\documentclass[%
reprint,
nofootinbib,
mathtools,
amssymb,
aps,
prx,
superscriptaddress,
usenames,dvipsnames,
floatfix
]{revtex4-2}

\usepackage{graphicx} % Include figure files
\usepackage[caption=false]{subfig} % use this instead of caption and subcaption, since those are incompatible with revtex4-2!
\usepackage{tabularray}
\usepackage{multirow}
\usepackage{comment}

\usepackage{bm}% bold math
\usepackage[tbtags]{amsmath}
\usepackage{amsthm}
\usepackage{nicefrac}
\usepackage{bbm}
\usepackage{physics}
\usepackage{titlesec}

\usepackage{tikz}
\usetikzlibrary{quantikz2}

\usepackage{enumitem}

\usepackage{xcolor,colortbl}
\usepackage[colorlinks=true,linktocpage=true]{hyperref}%
\hypersetup{
    allcolors  = {NavyBlue},
}

\definecolor{IBMblue}{RGB}{100,143,255}
\definecolor{IBMpurple}{RGB}{120,94,240}
\definecolor{IBMpink}{RGB}{220,38,127}
\definecolor{IBMorange}{RGB}{254,97,0}
\definecolor{IBMyellow}{RGB}{255,176,0}

\theoremstyle{plain}
\newtheorem{theorem}{Theorem}
\newtheorem{lemma}[theorem]{Lemma}
\newtheorem{proposition}[theorem]{Proposition}

\theoremstyle{plain}

\theoremstyle{definition}

\newcommand{\ugrshort}{Quantum Thermodynamics and Computation Group, Electromagnetism and Matter Physics Department, University of Granada, Granada, Spain}
\newcommand{\jku}{Department for Quantum Information and Computation at Kepler (QUICK), Johannes Kepler University, Linz, Austria}
\newcommand{\uottawa}{Department of Mathematics and Statistics, University of Ottawa, Ottawa, Canada}
\newcommand{\luh}{Institut f{\"{u}}r Theoretische Physik, Leibniz Universit{\"{a}}t Hannover, Hannover, Germany}
\newcommand{\bilkent}{Department of Mathematics, Bilkent University, Ankara, Turkey}

\newcommand{\F}[1]{\{0,1\}^{#1}}
\newcommand{\Cl}[1]{\mathcal{C}\ell_{#1}} % Clifford group
\renewcommand{\P}[1]{\mathcal{P}_{#1}} % Pauli group

\newcommand{\str}[1]{\boldsymbol{#1}}

\newcommand{\SGate}{S}
\newcommand{\TGate}{T}
\newcommand{\CNOT}{CX}
\newcommand{\PREC}{\prec}
\newcommand{\papertitle}{
Normal-Form Characterization of Adaptivity and Entanglement\\in Pauli-Based Computation}

\begin{document}
\title{\papertitle}

\author{Filipa C.\,R. Peres}
\affiliation{\ugrshort}
\affiliation{\jku}

\author{Selman Ipek}
\affiliation{\luh}
\affiliation{\bilkent}

\author{Daniel Lovsted}
\affiliation{\uottawa}

\author{Joshua Nevin}
\affiliation{\uottawa}

\author{Anne Broadbent}
\affiliation{\uottawa}

\makeatletter \let\@date\@empty \makeatother

\begin{abstract}
    Pauli-based computation~(PBC) is a measurement-based model of quantum computation that samples from output distributions of universal Clifford+$T$ circuits. Here, we provide a normal form~(NF) that completely characterizes its adaptive and entanglement structures. These insights unify PBC with one-way quantum computing as dual or `time-reversed' counterparts. We use the NF to advance strategies for distributed computation and investigate magic-state injection rates. We also demonstrate an exponential improvement in the detection of unnecessary magic-state qubits, allowing their efficient removal. Our NF further enables ahead-of-time circuit compilation using PBC, reducing CNOT gate count and quantum memory of random quantum circuits on average by 50\% and 40\%, respectively, beating other state-of-the-art compilers.
\end{abstract}

\maketitle

\begin{figure*}
    \centering
    \includegraphics[width=0.99\linewidth]{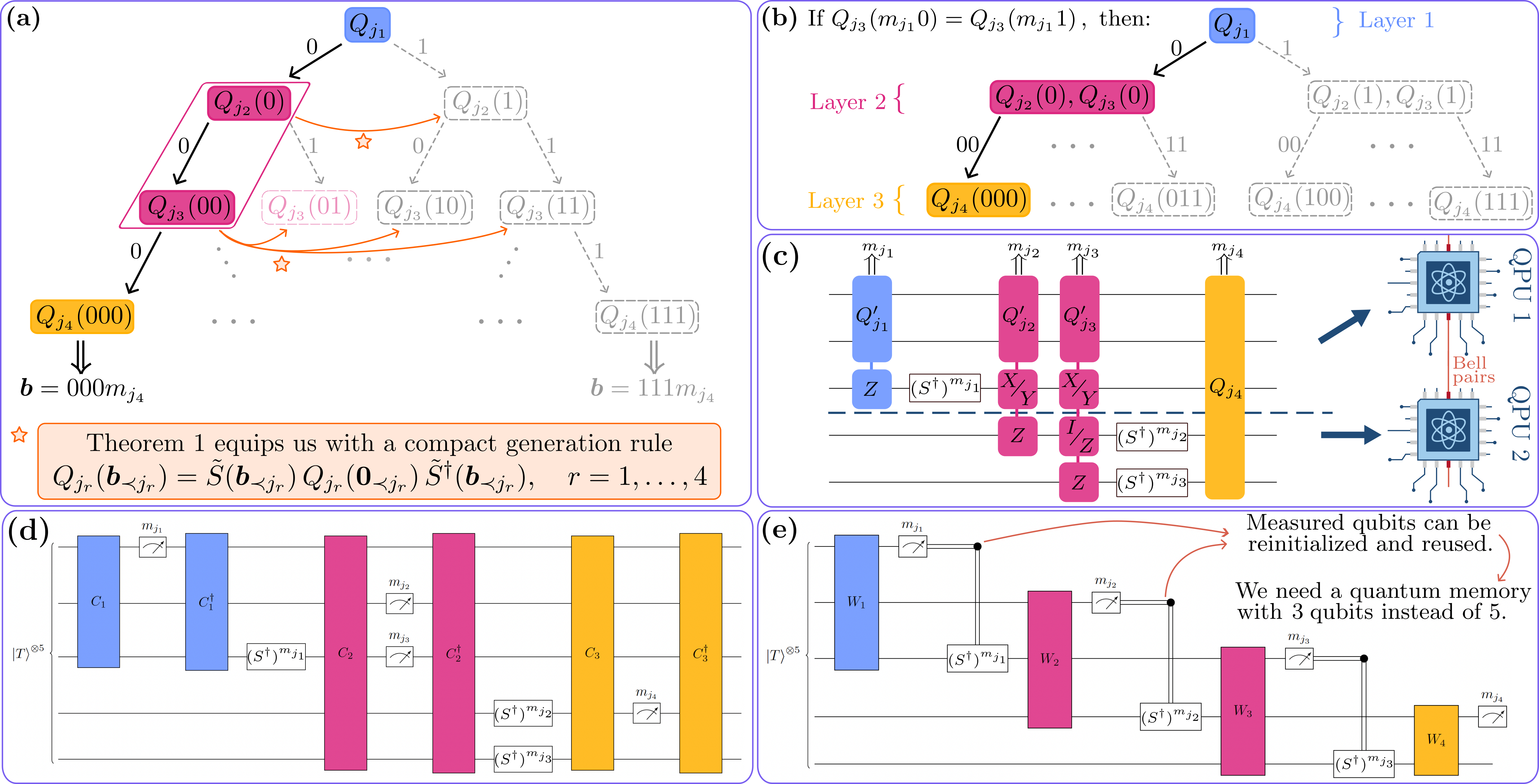}
    \caption{\textbf{(a)~PBC binary tree.} The reference branch (in black) can be determined efficiently; the NF provides a rule for obtaining any other node in the tree from the same-level node of the reference branch (orange arrows, marked with stars).
    \textbf{(b)~NF PBC tree.} The NF allows us to identify measurements that belong to the same adaptive layer.
    \textbf{(c)~Fault-tolerant architecture design.} The adaptive structure implies that not all magic states are needed at once (cf.~Fig.~\ref{fig: Depth-Injection}), while the entanglement structure helps establish a suitable strategy for distributed computation.
    In circuit drawings, boxes with rounded (resp.~sharp) edges denote measurements (resp.~unitary gates).
    \textbf{(d)~Naive ahead-of-time (AOT) circuit compilation.} The NF enables AOT compilation with PBC using a simple compute-and-uncompute strategy. \textbf{(e)~Enhanced AOT circuit compilation.} By removing the uncompute blocks, we reduce the quantum memory size and CNOT gate counts (cf.~Fig.~\ref{fig: AOT compiler RQCs}).
    }
    \label{fig: summary fig}
\end{figure*}

In Pauli-based computation (PBC), a supply of high-fidelity magic states and logical, multiqubit Pauli measurements are combined to enable universal quantum computation. Since its proposal~\cite{Bravyi2016trading}, the interest in PBC has been growing and its relevance across different areas of quantum computing research is now well established. For instance, PBC has proved useful in pushing the boundaries of efficient classical simulation~\cite{Bravyi2016trading, Zurel2020hidden, Zurel2024simulating}, demonstrating quantum advantage for specific families of quantum circuits~\cite{Yoganathan2019quantum}, and investigating dynamical phase transitions~\cite{Bejan2024dynamical}. It has also been used successfully to reduce the quantum memory size, gate counts, and depth of many quantum computations~\cite{Peres2023quantum, Peres2025reducing}, including when doing blind computation~\cite{Lovsted2026blind}. These capabilities render PBC relevant for near- to intermediate-term applications of quantum technology. Moreover, since its only operational primitives are logical Pauli measurements, PBC constitutes a natural scheme for fault tolerance via quantum low-density parity-check codes~\cite{He2025extractors, Zhang2025time, Cowtan2026parallel, Webster2026pinnacle, Litinski2019GoS, Kan2025SPARO, Yoder2025tour}. This establishes its long-term viability for quantum computation.

Since PBC is driven exclusively by measurements, it fits within the broader framework of measurement-based quantum computation~(MBQC)~\cite{gottesman1999demonstrating,raussendorf2001one,nielsen2003quantum,Leung2004unify,miller2016hierarchy,bartolucci2023fusion}. Despite the probabilistic nature of quantum measurements, deterministic computation can be recovered in these models by choosing upcoming measurements based on the outcomes of previous ones. In one-way quantum computing (1WQC), the measurement calculus provides a systematic framework for studying these dependencies~\cite{danos2007measurement}, while flow and its generalizations relate the sequence of adaptive corrections to the graph specifying the entangled resource state~\cite{danos2006determinism,browne2007generalized}. Explicit correction schemes have also been developed for computation on hypergraph states~\cite{gachechiladze2019changing} and with fusion measurements~\cite{de2019pauli,de2024fusion}. In contrast, a comparable characterization of adaptivity and entanglement in PBC has been missing.

In this Letter, we fully close this gap by providing a normal form (NF) for PBC that identifies its adaptive corrections as local Clifford phase gates applied to a precomputed reference sequence of measurements.
This reveals a surprisingly regular structure in PBC, not only of adaptivity, but also of entanglement, since the local nature of the corrections preserves the entanglement of the reference measurement sequence. For a pictorial summary of our results, see Fig.~\ref{fig: summary fig}.

From a theoretical perspective, our NF establishes a structural connection between PBC and 1WQC: Up to Pauli factors, the adaptive corrections in PBC are local Clifford symmetries of the magic resource state. This parallels 1WQC~\cite{raussendorf2001one}, where stabilizer symmetries of the graph state control the corrections and constrain the adaptive structure~\cite{raussendorf2016symmetry}. Our NF also establishes an explicit duality between these two models through which the Pauli observables of PBC's reference measurement sequence specify the resource state of the dual one-way computation. We therefore provide a unified view of PBC and 1WQC, as done previously for other measurement-based models~\cite{Leung2004unify, Perdrix2005unify, ChildsLeung2005unify}.

One advantage of PBC is that all Pauli measurements commute and our results provide a clear prescription for its adaptive measurement sequence. We leverage this for fault-tolerant applications, discussing methods for optimizing the distribution of a large PBC over a set of small(er) quantum processing units~(QPUs) and determining the magic-state generation rate needed for time-optimal implementation of specific computations. We also show how to efficiently detect and remove unnecessary magic-state qubits in a PBC, allowing us to sample from the desired distribution using fewer magic states. Previous methods could only accomplish this in time exponential in the original number of magic states input to the PBC. In contrast, our results allow us to do it in polynomial time.

Given current hardware limitations, the optimization of quantum circuits is a crucial task. Our NF allows us to perform ahead-of-time~(AOT) circuit compilation with PBC, demonstrating reductions in CNOT count and quantum memory size of 40--75\,\% and 20--60\%, respectively, for random quantum circuits (RQCs) as compared to state-of-the-art results, and 20--70\,\% and 10--60\,\% for a suite of benchmark circuits. 

\emph{Normal form.}---We consider an $n$-qubit Clifford+$T$ circuit containing $t$ $T$ gates, where $T \coloneqq \operatorname{diag}(1,e^{i\pi/4})$, and ending with $w$ {computational basis} readout measurements. Each $T$ gate can be replaced by a measurement gadget that injects an auxiliary magic state $\ket{T} \coloneqq T\ket{+}$ and applies a Clifford correction $S^m$, with $S\coloneqq \operatorname{diag}(1,i)$ and $m\in \F{}$ the gadget's measurement outcome~\cite{bravyi2005universal}.
Gadget qubits and their corresponding measurements are indexed by $1,\dots,t\,,$ while readouts are indexed by $t+1,\dots,t+w$.
The outcome of the $j$th measurement is denoted $m_j$.

The corresponding PBC is obtained by back-propagating the $t+w$ measurements through the adaptive Clifford circuit and applying the reduction procedure of Ref.~\cite{Bravyi2016trading}. As a result, a PBC can be represented using a binary tree, where each node is labeled by a multiqubit Pauli observable and each edge by the corresponding binary measurement outcome~[Fig.~\ref{fig: summary fig}(a)]. Each branch of the tree, i.e., root-to-leaf path, consists of $\ell$ gadget measurements $Q_{j_1},\dots,Q_{j_\ell}\,,$
followed by $\ell'$ readout measurements $Q_{j_{\ell+1}},\dots,Q_{j_{\ell+\ell'}}$, where $\ell+\ell'\leq t$. All measurements pairwise commute and are performed in succession on the $t$-qubit magic state.
Their measurement outcomes form the binary string $\str b = m_{j_1}\cdots m_{j_{\ell+\ell'}}$, which labels the branch $\mathcal B_{\str b}$.
Any remaining outcomes are simulated or inferred classically.
Since gadget outcomes trigger the adaptive $\SGate$ corrections, we write $Q_{j_r}(\str b_{\PREC j_r})$, where $\str b_{{\PREC  j_r}} := m_{j_1}\cdots m_{\min(j_r-1,\ell)}$ denotes this causal dependence.
We now present our main result.

\begin{theorem}[PBC normal form]
\label{thm:pbc-normal-form}
The Pauli observables measured in a branch $\mathcal B_{\str b}$ of a PBC are related to those of the all-zero reference branch $\mathcal B_{\str 0}$ by
\[
Q_{j_{r}}(\str b_{\PREC j_r})
=
\tilde \SGate(\str b_{\PREC j_r})\,Q_{j_{r}}(\str 0_{\PREC j_r})\,\tilde \SGate^\dagger(\str b_{\PREC j_r})\,,
\]
where $r = 1,\dots, \ell+\ell'$ and $\tilde S(\str b_{\PREC j_r}):=\prod_{u=1}^{\min{r-1,\ell}}S_{j_u}^{m_{j_u}}$.
\end{theorem}
To provide an intuitive proof sketch, each gadget correction $S^m$ can be acausally teleported onto the corresponding magic state, yielding $(S^{\dagger})^m\ket{T}$; see Fig.~\ref{fig:correction-teleportation}. Back-propagating the measurements through the resulting correction-free Clifford circuit gives the reference observables; transferring the phase corrections from the input state to the observables then gives the conjugation relation. Despite the acausal construction, each observable depends only on preceding outcomes. Crucially, the PBC reduction procedure preserves the conjugation relation and retains the same measurement indices $j_1,\dots,j_{\ell+\ell'}$ on every branch. The full proof is given in Appendix~B.
\begin{figure}[t]
\centering
\begin{tikzpicture} \node[scale=0.75]{
\begin{quantikz}[thin lines]
\lstick{$\ket{\psi}$}
    & \ctrl{1}
    & \qw
    & \gate{\SGate^{m}}
    & \rstick{$T\ket{\psi}$}
    \\
\lstick{$\ket{T}$}
    & \targ{}
    & \meter[label style={xshift=6mm,yshift=-3mm}]{m} 
    &\cwbend{-1}\setwiretype{n}
\end{quantikz}%
\quad $=$\quad
\begin{quantikz}[thin lines]
\lstick{$\ket{\psi}$}
    & \ctrl{1}
    & \qw
    & \rstick{$T\ket{\psi}$}\qw
    \\
\lstick{$(\SGate^\dagger)^m\ket{T}$}
    & \targ{}
    & \meter[label style={xshift=6mm,yshift=-3mm}]{m}
    & \setwiretype{n}\cw
\end{quantikz}
};
\end{tikzpicture}
\caption{\textbf{Acausal magic-state injection.} For a fixed outcome~$m$, the correction can be transferred to the magic state.}
\label{fig:correction-teleportation}
\end{figure}
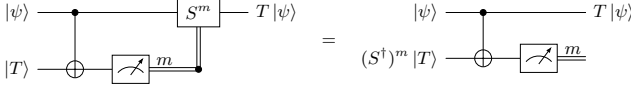

We now show that symmetries of the magic state resource~\cite{heinrich2019robustness} can be used to simplify the binary PBC tree. Indeed, $SX\ket{T} = e^{i\pi/4} \ket{T}$, and $S$ and $SX$ have the same conjugation action on Pauli observables, up to an overall sign. 
The distinct measurement sequences therefore form the orbit of the reference measurement sequence under the group
\[
G = \langle S_{j_1},\dots,S_{j_{\ell}}\rangle/\langle Z_{j_1},\dots,Z_{j_{\ell}}\rangle \cong \mathbb F_2^\ell.
\]
As a simple example, if the reference branch is composed of diagonal Pauli observables, there is a single distinct measurement sequence in the orbit and the PBC is therefore nonadaptive. Furthermore, consecutive observables belong to the same adaptive layer precisely when the phase corrections generated by outcomes within the block leave its remaining observables unchanged. We thus introduce the NF PBC tree, a more structured graphical representation of a PBC where nodes are labeled by observables in the same adaptive layer and edges by their outcome strings. The resulting graph will generally not be a binary tree~[Fig.~\ref{fig: summary fig}(b)].

Since local Clifford unitaries generate the PBC tree, its entanglement structure is branch-invariant. This can be characterized precisely
when all qubits in the original circuit are measured. In that case, the corresponding PBC reference-branch projector has rank one (see SM~\ref{sec:pbc-reduction}) and it defines the pure stabilizer state resource of a dual, computation-specific 1WQC~\cite{raussendorf2003measurement,Vijayan+2024compilation,KaldenbachHeller2025mapping,Okay2024classical}.
\begin{proposition}
A $t$-qubit PBC with $t$ independent reference-branch generators admits a computation-specific 1WQC with resource state
\[
\lvert\Psi_0\rangle\langle\Psi_0\rvert
=
\prod_{r=1}^{t}
\frac{I+Q_{j_r}(\str0_{\PREC r})}{2}
\]
and adaptive single-qubit measurements in the $T/ST$ bases.
\label{pro:1wqc-dual}
\end{proposition}

Connections between PBC and 1WQC have previously been observed in Refs.~\cite{Peres2025reducing,Lovsted2026blind}.
Our result, however, goes further by constructing an efficiently computable and invertible, probability-preserving bijection between the PBC and 1WQC outcome records. See Appendix~C and SM~\ref{sec:sm-dual-one-way-qc} for further details. The central observation is that exchanging the roles of preparation and measurement gives the PBC reduction procedure a time-reversed interpretation: The stabilizer inputs removed by the PBC reduction become local Pauli measurements on a stabilizer resource state (Appendix~A).
Removing these inputs therefore corresponds to the removal of Pauli measurements in 1WQC~\cite{raussendorf2003measurement}, with the two procedures eliminating the same stabilizer part of the computation.
\begin{figure}
    \centering
    \includegraphics[width=0.85\linewidth]{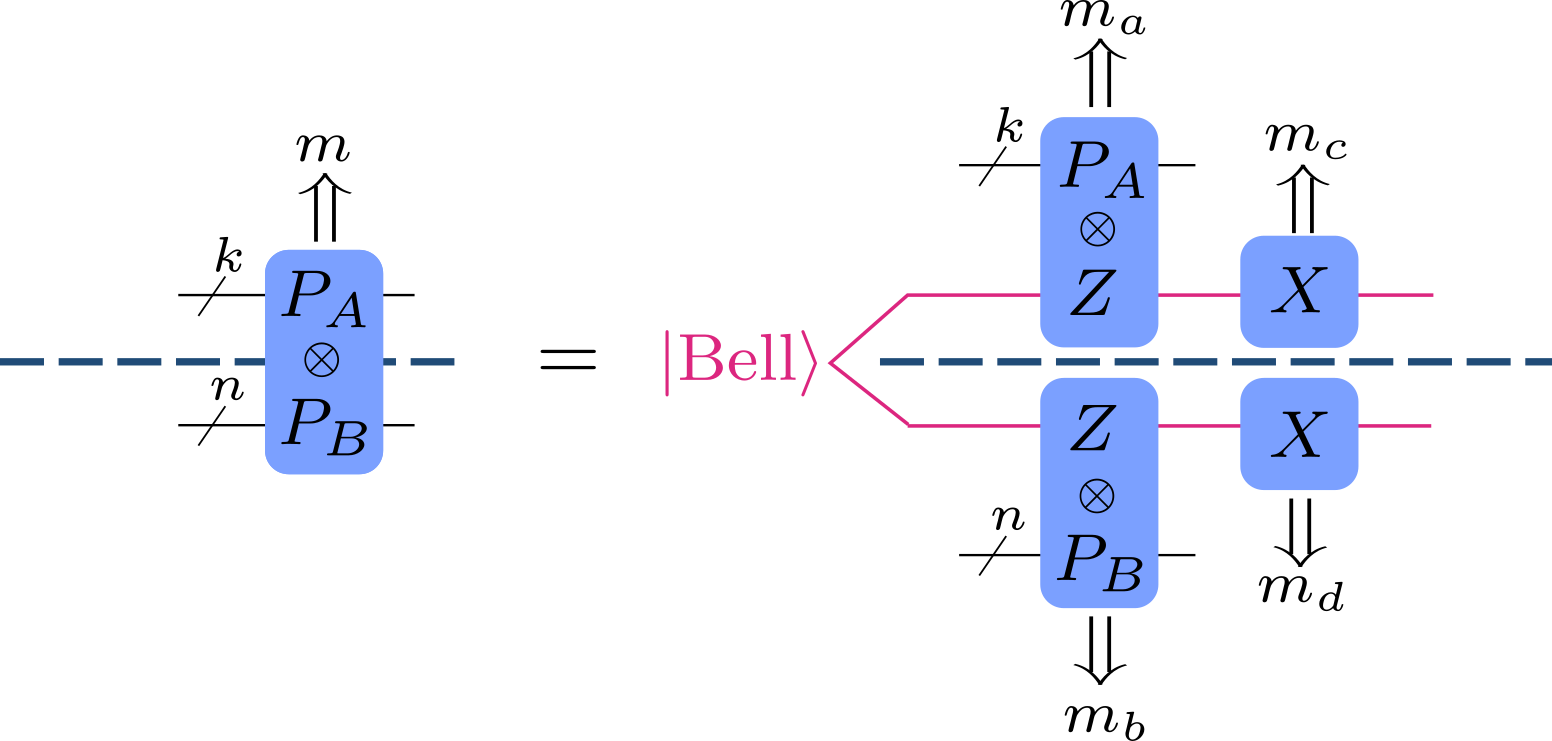}
    \caption{\textbf{Distributing a Pauli measurement across two processors.} Provided there is preshared entanglement between the two QPUs in the form of Bell pairs, each measurement in a PBC can be distributed, with the corresponding outcome obtained as $m = m_a \oplus m_b$. When $m_c \oplus m_d = 1\,,$ a Pauli correction needs to be applied (see SM~\ref{sm: distributed qc}).}
    \label{fig: Distributed QC}
\end{figure}

\emph{Fault-tolerant architecture design.---}The tensor-product structure of the reference branch identifies when the PBC can be naturally distributed among smaller QPUs requiring only \emph{classical} communication between them. Absent this, the future of fault-tolerant quantum computing will likely rely on distributing computations across such QPUs communicating via \emph{quantum} resources (see e.g.~\cite{He2025extractors}). In PBC, this involves splitting multiqubit Pauli measurements into smaller instances performed on separate QPUs. One possible strategy is depicted in Fig.~\ref{fig: Distributed QC}. Crucially, using just the observables of the reference branch, we can propose and optimize different distribution strategies for PBC. One option consists of first analyzing the support of all the Pauli measurements and using that information to minimize the number of measurements that are distributed across multiple QPUs. This strategy minimizes Bell-pair usage. One can also seek alternative distribution schemes that instead minimize the depth of the distributed computation.  Other space-time trade-offs arising when the QPUs share one or multiple Bell pairs at a time can also be investigated (SM~\ref{sm: distributed qc}).

Since, in PBC, all measurements commute, for architectures where Pauli measurements with overlapping support can be measured simultaneously~\cite{Cowtan2026parallel, Zhang2025time}, the PBC depth coincides exactly with the adaptive depth obtained with the NF.
In general, this implies that the PBC depth can be much smaller than the total number of Pauli measurements, contrary to prior expectations~\cite{Bravyi2016trading, Beverland2022surface, Chamberland2022universal}. This is illustrated in Fig.~\ref{fig: Depth-Injection} (left), where RQCs with hundreds of $T$ gates are transformed into PBCs with adaptive depth significantly smaller than $t$. Moreover, except for special instances (e.g.,~nonadaptive PBCs), not all magic states need to be available immediately [cf.~Fig.~\ref{fig: summary fig}(c)], with the rate of magic state injection from layer to layer being approximately constant~[Fig.~\ref{fig: Depth-Injection}(right)]. This rate can be lowered, if needed, at the expense of serializing the measurements. Thus, the NF can inform optimal space-time implementations of quantum algorithms.
\begin{figure}
    \centering
    \includegraphics[width=0.97\linewidth]{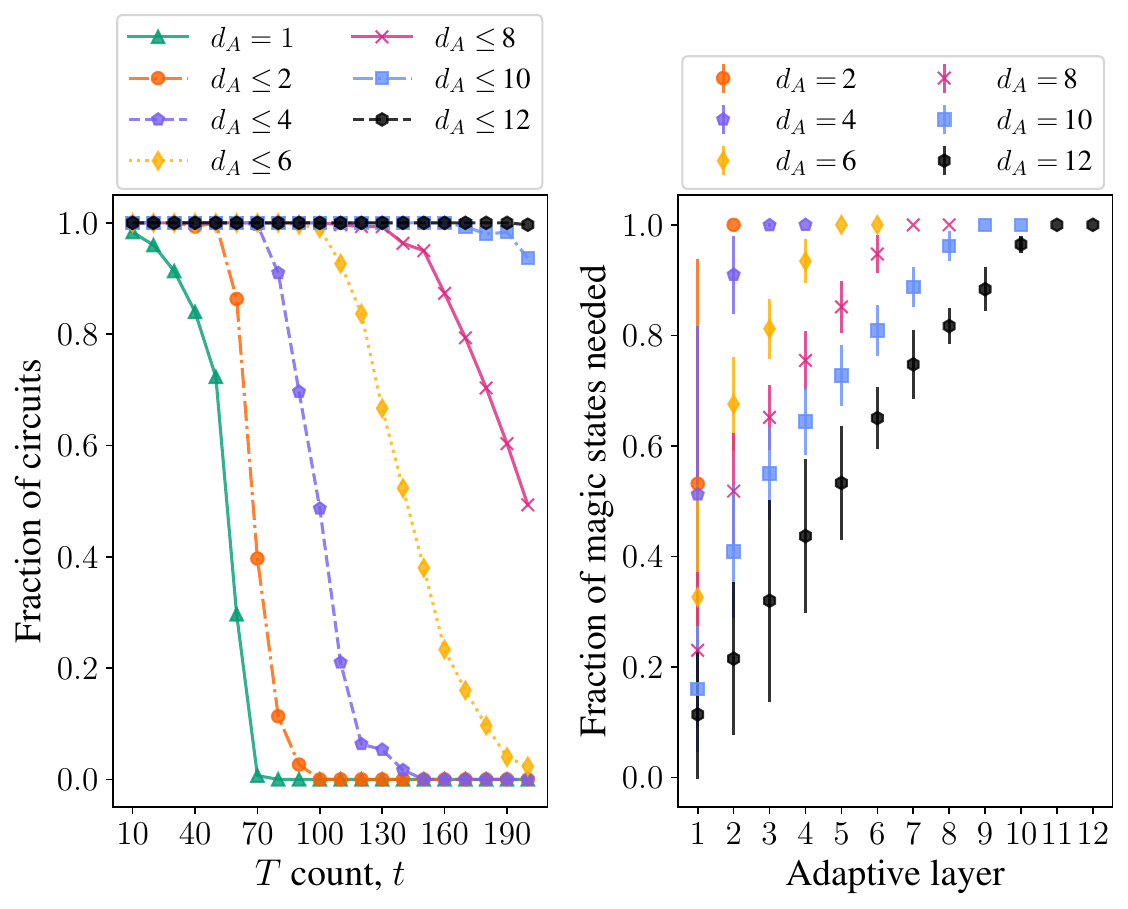}
    \caption{\textbf{PBC features associated with a suite of RQCs.} The circuits were generated as described in Ref.~\cite{Peres2023quantum}, acting on a $7\times 7$ qubit grid, with a varying $T$ count, $t=\{10,\dots,200\}$, and 100 circuits for each $t$. As $t$ increases, so does the average adaptive depth (left). For PBCs with larger adaptive depth, the fraction of magic qubits that need to be supplied at each layer is lower (right). Uncertainty bars correspond to 2 standard deviations.}
    \label{fig: Depth-Injection}
\end{figure}

\emph{Free virtual qubits.---}Our NF detects that \emph{some} qubits are strictly unnecessary to perform the desired computation. For a given PBC, a `\emph{free virtual qubit}' is a magic-state qubit that is only ever measured in the~$Z$ basis.

\begin{proposition}[Free virtual qubit detection]\label{prop: Free virtual qubits}
    Given a computation in the form of a Clifford+$T$ circuit, the free virtual qubits in the corresponding PBC can be efficiently identified.
\end{proposition}

Previously, no method for this task was known besides enumerating the entire exponentially large PBC tree. Our NF implies that inspecting a single branch of the tree suffices, thus offering an exponential improvement.

Since~$\ket{T} \coloneqq T\ket{+}$ and~$T$ commutes with~$Z\,,$ free virtual qubits can be treated as stabilizer inputs to a PBC and thus removed using conventional techniques~\cite{Bravyi2016trading}.
\begin{proposition}[Free virtual qubit removal]\label{prop: free qubit removal}
    Given a $t$-qubit PBC with $k$ free virtual qubits identified, it is efficient to find a new PBC on $(t-k)$ qubits that samples from the same output distribution.
\end{proposition}

Free virtual qubits can thus be removed at no cost beyond efficient classical processing. In contrast, conventional qubit virtualization~\cite{Bravyi2016trading, Peres2023quantum} and related techniques~\cite{PengHarrowOW2020simulating, Piveteau2024circuit} incur an exponential sampling overhead and produce expectation value estimates (rather than output samples). The proofs of Propositions~\ref{prop: Free virtual qubits} and~\ref{prop: free qubit removal} can be found in Appendix~D.

We applied our methods to three different sets of circuits to gain some intuition for the amount of free virtual qubits generally present in computations. For $49$-qubit RQCs with $t=10,\dots,200$, an average 9\% reduction in the number of magic states was obtained; for $25$-qubit RQCs with $t=4,\dots,22\,,$ the reduction is between 15--30\%; finally, for the benchmark circuits discussed in SM~\ref{sm: More on AOT compilation performance}, we see reductions between 40 and 80\%.

\emph{Circuit compilation.---}Mid-circuit measurements are very powerful primitives for several quantum computing tasks~\cite{CaoEisert2026mmtdriven, Grilo2026power, Lu2022mmt, Smith2023deterministic}. This motivates the development of compilation strategies that capitalize on this power and transform nonadaptive quantum circuits into \emph{adaptive} ones~\cite{DeCross+2023qubitreuse, Fang+2026, KaldenbachHeller2025mapping, Baumer+2024qft, Vijayan+2024compilation, Peres2023quantum, Peres2025reducing}. The PBC compiler was one of the first among these, demonstrating significant reductions in circuit gate counts~\cite{Peres2023quantum, Peres2025reducing}, but was limited to just-in-time~(JIT) compilation coupled to the computation itself. This required quantum hardware to remain coherent during the intricate classical side-processing~\cite{Peres2023quantum, Peres2025reducing}, constituting a challenge for current quantum technology.

\begin{figure}
    \centering
    \includegraphics[width=0.99\linewidth]{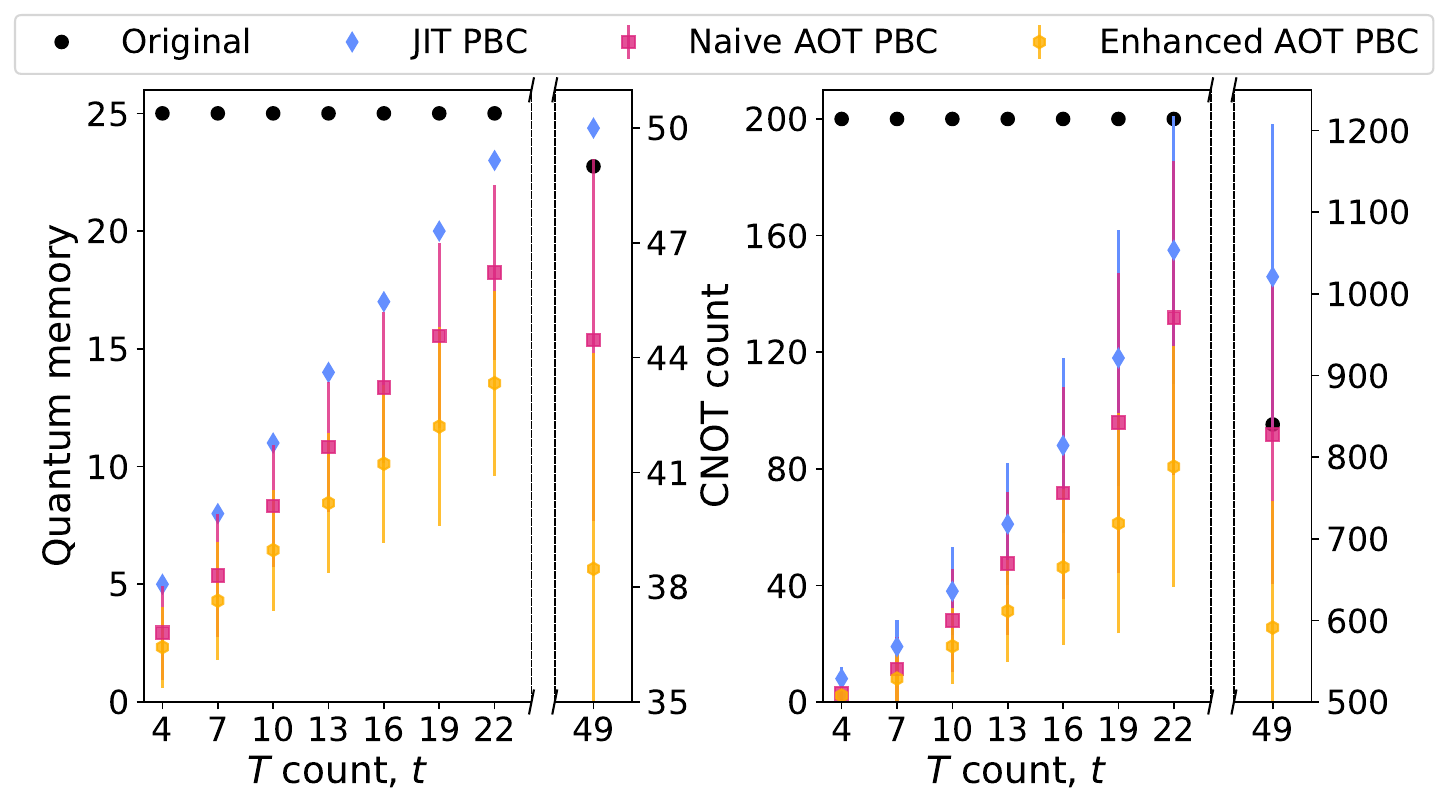}
    \caption{\textbf{Performance of our AOT PBC compiler for RQCs.} We take the RQCs analyzed in Refs.~\cite{Peres2023quantum,Peres2025reducing}. The $x$-axis reports the $T$ count of the original circuits input to the compilers. The naive and enhanced AOT PBC compilers consistently outperform the JIT compiler, leading to circuits with fewer qubits~(left) and CNOT gates~(right). Uncertainty bars correspond to 2 standard deviations.}
    \label{fig: AOT compiler RQCs}
\end{figure}

Our NF eliminates this drawback, allowing the PBC compiler to fully specify the compiled, \emph{adaptive} circuit AOT. Our techniques rely exclusively on the reference-branch Pauli observables to synthesize suitable Clifford subcircuits which are then interleaved with the local $S^{\dagger}$ corrections.
A simple `compute-uncompute' strategy, which we term \emph{naive} AOT PBC compilation, is depicted in Fig.~\ref{fig: summary fig}(d): Each block of parallelizable Pauli measurements is mapped to a sequence of computational basis measurements via conjugation by a Clifford subcircuit.

We now lay out a more sophisticated strategy that we call \emph{enhanced} AOT PBC compilation. Suppose that, for the $i$th PBC layer, we can find a Clifford subcircuit, $W_i$, so that $W_i^{\dagger}$ commutes with every ensuing layer of $S^{\dagger}$ gates. We can then push $W^{\dagger}_i$ to the end of the circuit---suitably updating the upcoming Pauli observables---and then drop it. This strategy effectively lowers the quantum memory because updated Pauli measurements have trivial support on qubits that have already been measured [Fig.~\ref{fig: summary fig}(e)]. Moreover, by avoiding the direct implementation of $W_i^\dagger$, it also achieves smaller CNOT gate counts and quantum depth than naive AOT compilation. Implementation details can be found in SM~\ref{sm: DJ} and~\ref{sm: Removing uncompute block}.

Fig.~\ref{fig: AOT compiler RQCs} shows that both the naive and enhanced AOT PBC strategies beat the JIT PBC compiler, which already outperformed competing state-of-the-art compilers~\cite{KissWet2020reduceT, Wetering2020zxcalculus, Staudacher2023} for two-qubit gate count~\cite{Peres2025reducing}.
In naive AOT PBC compilation, the observed reduction in quantum memory comes exclusively from free virtual qubit detection and removal, with further reductions achieved by the enhanced AOT strategy. Beyond RQCs, for several benchmark circuits, we also consistently outperform the state-of-the-art ZX-compiler~\cite{KissWet2020reduceT, Wetering2020zxcalculus, Staudacher2023} in terms of memory size, CNOT counts, and quantum depth (see SM~\ref{sm: More on AOT compilation performance}).

\emph{Outlook.---}In this Letter, we have characterized the adaptive and entanglement structure of PBC, resulting in applications in both the near and long term. Given PBC's past success in the context of classical simulation~\cite{Bravyi2016trading, Zurel2020hidden, Zurel2024simulating}, we anticipate that our NF may play a complementary role in exploring the boundaries of quantum advantage by identifying hidden structures in families of circuits that render them efficiently simulable classically (see SM~\ref{sm: DJ}).

Our methods for circuit optimization lead to reduced quantum computational resources and facilitate near-term quantum hardware implementation. This can likely be pushed even further by integration with other optimization strategies. For example, computation-specific 1WQC has recently been investigated as a means of reducing resource costs~\cite{Vijayan+2024compilation,kaldenbach2025mapping}. Our PBC/1WQC duality provides a possible route for utilizing these graph-based optimizations within PBC. Similarly, our NF allows the greedy algorithm from~\cite{Peres2025reducing} to be run AOT, suggesting additional performance gains.

PBC has several appealing features for the fault-tolerant regime, such as dispensing with stabilizer-state qubits and having only commuting Pauli measurements. Our results provide a clear prescription for the sequence of Pauli measurements, enabling a more precise account of the quantum resources needed for implementing quantum algorithms under specific fault-tolerance architectures.
For instance, our methods clarify the rate at which magic states are needed for time-optimal implementation, allowing us to adjust the number of magic state distillation factories per computational block. The NF therefore enables a systematic analysis of PBC, facilitating a more direct study of the domains in which it is advantageous over alternative schemes for fault-tolerant computing~\cite{Litinski2019GoS}.

\emph{Data availability statement.---}A code implementation of Theorem~\ref{thm:pbc-normal-form} and Proposition~\ref{prop: free qubit removal} is publicly available on GitHub~\cite{GitHubPACTS} and was used to produce all the numerical results presented in this Letter.

\emph{Acknowledgments.---}S.I. thanks R. Raussendorf for discussions. F.C.R.P. acknowledges support from Ayuda Consolidación CNS2023-145392 (MICIU\slash AEI\slash 10.13039\slash 501100011033, NextGenerationEU\slash PRTR) and the European Research Council (ERC) via the Starting grant q-shadows (101117138).
S.I. acknowledges support from Nieders{\"{a}}chsisches Ministerium f{\"{u}}r Wissenschaft und Kultur and the Alexander von Humboldt
Foundation. A.B., D.L., and J.N. acknowledge the support of the Natural Sciences and Engineering Research Council of Canada (NSERC) (ALLRP 569582-21), and of the Canada Research Chairs Program (CRC-2023-00173).
All authors acknowledge funding from EU HORIZON RIA FoQaCiA GA 101070558. AI tools were used to assist with code debugging and final manuscript review. All scientific contributions are original work by the authors.

\emph{Appendix A: PBC reduction procedure.---}For a Clifford+$T$ circuit $C$ with $t$ $T$ gates and $w$ computational-basis readout measurements, let $C'$ be the corresponding gadgetized circuit on $n+t$ qubits with adaptive $S$ corrections. We write a fixed complete outcome record as $\str m := (\str g,\str y)$ with gadget outcomes $\str g := m_1\cdots m_t$ and readout outcomes $\str y = m_{t+1}\cdots m_{t+w}$.
For the $i$th measurement, we denote its adaptive dependence on gadget outcomes by $\str g_{\PREC i}:=m_1\cdots m_{\min(i-1,t)}$; the first gadget being non-adaptive.
Back-propagation of gadget and readout observables $Z_1,\dots,Z_{t+w}$ through the adaptive circuit gives independent, commuting Pauli observables $P_1(\str g_{\PREC 1}),\ldots,P_{t+w}(\str g_{\PREC t+w})$. When the outcome dependence is clear, we suppress it and write simply $P_i$.

Organize the back-propagated Pauli observables into an initial stabilizer tableau $\mathcal T_0 := (P_1,\dots,P_{t+w})$
and let $\mathcal Z := \{Z_{t+i}\}_{i=1}^n$.
Write $\mathcal T_i = (P_1^{(i)},\dots,P_{t+w}^{(i)})$ for the tableau after processing its first $i$ rows with $P_h^{(0)} = P_h$.
The reduction procedure maintains two invariants: All rows commute and every processed row commutes with every element of $\mathcal Z$.
At stage $i=1,\dots,t+w$, process $P_i^{(i-1)}$ as follows:
\begin{itemize}
    \item[I.] $P_i^{(i-1)}$ commutes with all elements in $\mathcal Z$. No tableau update $\mathcal T_{i} = \mathcal T_{i-1}$.
    \item[II.] $P_i^{(i-1)}$ anticommutes with some computational $Z_{c_i} \in \mathcal{Z}$; if there are multiple, choose the one with the least index. Sample $m_i$ uniformly and update $\mathcal T_i  = V_i \mathcal T_{i-1} V_i^\dagger$ and $P_h^{(i)} := V_i P_h^{(i-1)} V_i^\dagger$.
\end{itemize}
Here, 
$V_i=[Z_{c_i}+(-1)^{m_i}P_i^{(i-1)}]/\sqrt2$
is a Clifford unitary satisfying
\[
P_h^{(i)}
=
\begin{cases}
P_h^{(i-1)}               
&:~     [P_h^{(i-1)},Z_{c_i}] = 0,\\
(-1)^{m_i} Z_{c_i}P_i^{(i-1)}P_h^{(i-1)}    
&:~    \{P_h^{(i-1)},Z_{c_i}\} = 0.
\end{cases}
\]
The sampled values in Case II match their entries in $\str m$.
For the normal-form proof below, we use the canonical value zero for classically simulated gadget outcomes.
Correctness of this procedure and comparison to Ref.~\cite{Bravyi2016trading} can be found in SM~\ref{sec:pbc-reduction}.

We can interpret the rows of the tableau $\mathcal T_{i-1}$ in dual fashion as stabilizers of a quantum resource state.
The tableau update in the anticommuting case above corresponds precisely to the stabilizer measurement update for a local computational $Z_{c_i}$ measurement.
To see this, we recall the tableau update of Ref.~\cite{AarosonGottesman2004improved}: All non-pivot rows, i.e., $h\neq i$, that anticommute with $Z_{c_i}$ are multiplied by $P_i^{(i-1)}\,,$ then the pivot row $P_i^{(i-1)}$ is replaced by $(-1)^{m_i}Z_{c_i}$.
The two resulting tableaus differ only by multiplication of the new generator $(-1)^{m_i}Z_{c_i}$ and thus generate the same stabilizer group.
This identifies the anticommuting PBC reduction steps with local computational-$Z$ measurements on the dual resource state.

Either way, this procedure performs Gaussian elimination on the $X$ block of the $n$-qubit register of the tableau.
Therefore, all updated observables have the final form $P_i'=Q_i\otimes\widehat Z_i$, where $Q_i$ acts on the magic-state register and $\widehat Z_i$ is a product of computational $Z$ operators.
Thus, in PBC, $\widehat Z_i$ acts trivially on $\ket{0}^{\otimes n}$, so we restrict to the auxiliary register.
The $Q_i$ commute but may become dependent under this restriction.
We extract an ordered subsequence of $\ell$ gadget observables $Q_{j_1},\ldots,Q_{j_\ell}$ and $\ell'$ readout observables $Q_{j_{\ell+1}},\ldots,Q_{j_{\ell+\ell'}}$ that give the quantum measurements on the $t$-qubit processor.
All other outcomes are simulated or inferred classically and retained to reconstruct the circuit readouts.

\emph{Appendix B: Proof of normal form.---}Fix a complete outcome record $\str m$. Figure~\ref{fig:correction-teleportation} allows all phase-gate corrections to be moved to the input magic states.
Defining $\widetilde \Sigma(\str g)=\bigotimes_{i=1}^tS_i^{m_i}$, this produces
$\widetilde \Sigma^\dagger(\str g)\ket{T}^{\otimes t}$ and leaves the correction-free Clifford circuit associated with the zero branch. Back-propagating each measurement $Z_i$, with $i=1,\dots,t+w$, yields the Pauli $P_i(\str 0_{\PREC i})$.
Moving the $\widetilde \Sigma(\str g)$ rotation from the input state onto the back-propagated observables then gives $\widetilde P_i(\str g)=\widetilde \Sigma(\str g)P_i(\str 0_{\PREC i})\widetilde \Sigma^\dagger(\str g)$.
This representation is causal: The $i$th zero-branch gadget observable has the form
\begin{equation}
P_i(\str0)
=
R_{i}\otimes Z\otimes I^{\otimes(t-i)}\otimes\widehat R_i ,
\label{eq:end-matter-Pauli}    
\end{equation}
where $R_{i}$ acts only on the preceding $i-1$ gadget qubits and $\widehat R_i$ on the computational register.
Since local phase gates commute with local $I$ and $Z$ Paulis, the actions of $\widetilde \Sigma(\str g)$ and its causal restriction $\widetilde \Sigma(\str g_{\PREC i}) := \prod_{j=1}^{\min(i-1,t)}S_j^{m_j}$ on $P_i(\str 0_{\PREC i})$ are identical.
Thus, we have $\widetilde P_i(\str g ) = \widetilde P_i(\str g_{\PREC i}) := \widetilde \Sigma(\str g_{\PREC i})P_i(\str 0_{\PREC i})\widetilde \Sigma^\dagger(\str g_{\PREC i})$, which depends only on earlier outcomes. The same result also holds for readout measurements.

We now show that the reduction procedure is branch-invariant.
Suppose the current tableau observables are related by
$\widetilde P_h^{(i-1)}(\str g_{\PREC h}) = \widetilde \Sigma(\str g)P_h^{(i-1)}(\str 0_{\PREC h})\widetilde \Sigma^\dagger(\str g)$,
with $h=1,\dots,t+w$.
Since $\widetilde \Sigma(\str g)$ fixes every computational $Z_{c}$, we have that $\widetilde P_i^{(i-1)}(\str g)$ anticommutes with $Z_c$ if and only if $P_i^{(i-1)}(\str 0)$ also does.
Thus, both branches enter the same reduction case. The commuting case is trivial.
In the anticommuting case, with matching simulated outcomes, the corresponding reduction Cliffords satisfy $V_i(\str g_{\PREC i})=\widetilde \Sigma(\str g)V_i(\str0_{\PREC i})\widetilde \Sigma^\dagger(\str g)$
so that
$\widetilde P_h^{(i)}(\str g_{\PREC h})
=
\widetilde \Sigma(\str g)
P_h^{(i)}(\str 0_{\PREC h})
\widetilde \Sigma^\dagger(\str g)$,
$h=1,\ldots,t+w$.

The final processed observables
$\widetilde P_h^{(t+w)}(\str g_{\PREC h}) = \widetilde \Sigma(\str g_{\PREC h})Q_h^{(t+w)}(\str 0_{\PREC h})\widetilde \Sigma^\dagger(\str g_{\PREC h})\otimes \widehat Z_{h}$ therefore preserve causality and act trivially on the circuit input $\ket{0}^{\otimes n}$. Clifford conjugation preserves Pauli dependence, so every branch retains the same indices after reduction to $t$ qubits. Fixing simulated gadget outcomes to zero identifies $\widetilde\Sigma(\str g_{\PREC j_r})$ with $\widetilde\SGate(\str b_{\PREC j_r})$ for every retained index $j_r$. This establishes the normal form.

\emph{Appendix C: Proof of Proposition~\ref{pro:1wqc-dual}.---}We assume that we measure all qubits in the original circuit, i.e., $w=n$. The reference branch projector is given by
\[
\Pi_{\str 0}
=
\prod_{i=1}^{t+n}
\frac{I+P_i(\str0_{\PREC i})}{2},
\qquad
\Pi_{\str0}=\ket{\mathcal R_{\str 0}}\!\bra{\mathcal R_{\str 0}},
\]
This is shown to be rank one in SM~\ref{sec:pbc-reduction}.
Introduce now the Pauli operators
\[
D(\str e(\str m), \str f(\str m)) \simeq Y_G^{\otimes \str e}\otimes X_C^{\otimes \str f}Z_C^{H\str f},
\]
where $G$ and $C$ distinguish the register of gadget and computational qubits, respectively.
Here $\simeq$ denotes equality up to phase, $\str e\in \mathbb{F}_2^t$, $\str f\in \mathbb{F}_2^n$, and $H\in \mathbb{F}_2^{n\times n}$ is a binary matrix.
The bit strings $\str e$ and $\str f$ are related to gadget and readout outcomes by an invertible linear map
\[
\str g = K\str e+M\str f,
\quad
\str y = B\str e+N\str f.
\]
The stabilizer tableau (see e.g.,~\cite{AarosonGottesman2004improved}) of the reference branch determines a unit lower-triangular matrix $K$ owing to the `triangular' form of the back-propagated Pauli observables in Eq.~\eqref{eq:end-matter-Pauli}.
The matrices $M$ and $N$ depend on the choice of $H$, which can be chosen so that $J=N+BK^{-1}M$ is invertible.
The resulting correspondence between frame coordinates and outcome records is therefore bijective.
The construction is given in SM~\ref{sec:sm-dual-one-way-qc}.

In particular, the operators corresponding to the record $\str m$ satisfy
\[
D(\str e,\str f)P_i(\str 0_{\PREC i})D^\dagger(\str e,\str f) = (-1)^{m_i}P_i(\str 0_{\PREC i})\,,
\]
for $i=1,\dots,n+t$.
Consequently, the projector
$D(\str e,\str f)\Pi_{\str 0}D^\dagger(\str e,\str f)$
is orthogonal to $\Pi_{\str 0}$, except when the outcome records are equal.

Let the initial PBC resource state be denoted $\ket{\Omega} := \ket{T}^{\otimes t}\otimes\ket{0}^{\otimes n}$.
Up to signs, we can choose the generators of the stabilizer group defining the branch projector to be
\[
\widetilde \Sigma(\str g)P_i(\str 0_{\PREC i})\widetilde \Sigma^\dagger(\str g),\quad i=1,\dots,t+n.
\]
The probability of obtaining outcome $\str m = (\str g,\str y)$ is
\[
\text{Pr}(\str g,\str y) = \lvert \bra{\Omega}\widetilde \Sigma(\str g)D(\str e,\str f)\ket{\mathcal R_{\str 0}} \rvert^2.
\]
The key idea is to show that we can reinterpret this probability as arising from a 1WQC with $\ket{\mathcal R_{\str 0}}$ as the resource. For each gadget qubit $i$, move the phase correction and Pauli onto the input $Y^{e_i}S^{-m_i}\ket T_i \simeq Z^{m_i}S^{e_i+m_i}\ket T_i$, with the exponent evaluated modulo two.
Thus, $\alpha_i=e_i+m_i$ selects the basis $B_{\alpha_i} := \{S^{\alpha_i}\ket T_i,ZS^{\alpha_i}\ket T_i\}$, and $o_i=m_i$ labels its outcome.

We measure the computational qubits in $Z$, obtaining the string $\str f$.
These measurements commute and have no dependence on prior measurement outcomes, but we can choose to measure them in the order reflecting the anticommuting PBC measurements.
By construction, the remaining computational measurements are deterministic, since we have a full tableau and all Pauli stabilizers are of the form $Q\otimes \widehat Z$.

The outcomes for the gadget qubits satisfy $\str o=\str g$ and we also have $\str g = K\str e+M\str f$, $\str y = B\str e+N\str f$.
Solving yields $\str e = K^{-1}(\str o + M\str f)$ and $\str y = BK^{-1}\str o + J\str f$.
The local basis choices are given by $\str \alpha = \str e + \str g$, which can now be explicitly written as
\[
\str \alpha = (K^{-1}+I)\str o+K^{-1}M\str f.
\]
Measure the gadget qubits using the gadget ordering. Since $K^{-1}+I$ is strictly lower triangular, the basis choice is causal.
The outcome map has inverse $\str o=\str g$, $\str f=J^{-1}(\str y+BK^{-1}\str g)$.
The overlap identity above equates the corresponding branch probabilities.
Since the complete circuit record is reversibly determined by the PBC outcomes and simulated bits, this establishes the desired bijection of complete records.

By Appendix~A, removing the computational sites leaves, up to known Pauli corrections, the stabilizer resource generated by the reduced reference observables:
\[
\ket{\Psi_0}\!\bra{\Psi_0}
=
\prod_{r=1}^{t}
\frac{I+Q_{j_r}(\str0_{\PREC  r})}{2}.
\]
The computational measurements can therefore be simulated classically, with the corresponding Pauli corrections incorporated into the remaining measurement bases and outcome labels.
Retaining $\str f$ preserves the probability-preserving bijection of complete records.

\emph{Appendix D: Proof of Propositions~\ref{prop: Free virtual qubits} and~\ref{prop: free qubit removal}.---}We now prove our results concerning free virtual qubits.

\begin{proof}[Proof of Proposition~\ref{prop: Free virtual qubits}]
    Given a computation in the form of a Clifford+$T$ quantum circuit, we can obtain the reference branch of the corresponding PBC. Detecting if a given magic-state qubit is a free virtual qubit reduces to checking that every Pauli observable in this reference branch acts with either $I$ or $Z$ on that qubit. Theorem~\ref{thm:pbc-normal-form} then guarantees that the same holds for every other branch of the PBC, since conjugation by the phase gate leaves these two Paulis unchanged. Finding the reference branch of the PBC and iterating over all qubits and Pauli measurements to detect free virtual qubits can be done in polynomial time.
\end{proof}

\begin{proof}[Proof of Proposition~\ref{prop: free qubit removal}]
    Suppose a free virtual qubit exists. Let us rewrite the initial state of that qubit as $\ket{T} \coloneqq T \ket{+}\,.$ Since, by definition, only $I$ or $Z$ act on it, the $T$ gate can effectively be conjugated past the whole PBC branch without modifying it. The qubit is then left in the stabilizer state $\ket{+}$. Since the conventional PBC reduction procedure naturally allows the efficient removal of any qubit initialized in a stabilizer state~\cite{Bravyi2016trading}, the proposition holds.
\end{proof}

\begingroup
\makeatletter
\let\savedbibsection\bibsection
\renewcommand{\bibsection}{%
  \begingroup
  \let\addcontentsline\@gobblethree
  \savedbibsection
  \endgroup
}
\endgroup

\newpage
\clearpage

\onecolumngrid
\begin{center}
    \noindent\textbf{Supplemental Material for: \emph{\papertitle}}
\end{center}

\setcounter{equation}{0}
\setcounter{figure}{0}
\setcounter{table}{0}
\setcounter{theorem}{0}
\makeatletter
\renewcommand{\thetable}{S\Roman{table}}
\counterwithout{equation}{section}
\renewcommand{\theequation}{S\arabic{equation}}
\renewcommand{\thefigure}{S\arabic{figure}}
\renewcommand{\bibnumfmt}[1]{[S#1]}
\renewcommand{\citenumfont}[1]{S#1}
\renewcommand{\thetheorem}{S\arabic{theorem}}
\renewcommand{\thelemma}{S\arabic{lemma}}
\renewcommand{\theproposition}{S\arabic{proposition}}
\renewcommand*{\thesection}{\Alph{section}}
\renewcommand*{\thesubsection}{\arabic{subsection}}
\renewcommand*{\thesubsubsection}{\alph{subsubsection}}
\titleformat{\section}[block]{\normalfont\bfseries}{Supplemental Material \thesection:}{0.5em}{\centering}

\twocolumngrid

\tableofcontents

\section{Stabilizer formalism}
\label{sec:stabilizer-formalism}

The $n$-qubit Hilbert space is the $n$-fold tensor product $\mathcal{H}_n := (\mathbb{C}^{2})^{\otimes n}$. We denote $n$-qubit Pauli operators by
\[
P_{u} := i^{u_{x}\cdot u_{z}}X^{u_{x,1}}Z^{u_{z,1}}\otimes \cdots \otimes X^{u_{x,n}}Z^{u_{z,n}} \,.
\]
Here $u = (u_{x},u_{z})$ is an element of the binary symplectic vector space $E_{n}:=\mathbb Z_{2}^{2n}$ with canonical basis $x_{1},\dots,x_{n},z_{1},\dots,z_{n}$. This space is equipped with a symplectic bilinear product $[u,v] = u_{x}\cdot v_{z}-u_{z}\cdot v_{x}$ evaluated over $\mathbb Z_{2}$. With this convention we have that $P_{u}P_{v}= (-1)^{[u,v]}P_{v}P_{u}$ so that $P_{u}$ and $P_{v}$ commute if and only if $[u,v]=0$ and anticommute otherwise.

The $n$-qubit Clifford group $\Cl{n}$ is the normalizer of $\mathcal P_{n}$ inside the unitary group, up to phases. The action of a Clifford unitary $U$ by conjugation on a Pauli operator therefore gives $UP_{u}U^{\dagger} = \pm P_{F_{U}(u)}$. Here $F_{U}:E_{n}\to E_{n}$ is the symplectic map induced by the Clifford $U$. Such maps are invertible, linear, and preserve the symplectic product, i.e., $[F_{U}(u),F_{U}(v)] = [u,v]$.
For completeness, we recall that a transvection $\tau_{v}$ is a symplectic map satisfying $\tau_{v}(u) = u+[u,v]v$. The symplectic version of the phase gate $\SGate_{i}$, denoted $\sigma_{i}$, acts as the transvection $\tau_{z_{i}}$ according to $\sigma_{i}(u) = u+[u,z_{i}]z_{i}$.

We say that a subspace $L \subset E_{n}$ is isotropic if all of its elements pairwise commute. An isotropic subspace $L $ of dimension $k$ corresponds, up to signs, to a stabilizer group $\mathcal S$ and we refer to the latter as also having dimension $k$. Maximal stabilizer groups, which define pure stabilizer states, therefore correspond to isotropic subspaces of dimension $k=n$. For an isotropic subspace $L\subset E_{n}$ of dimension $k$ its image $F(L)$ under a symplectic map $F$ is again an isotropic subspace of the same dimension.

\section{Pauli-based computation reduction procedure}
\label{sec:pbc-reduction}

\subsection{Back-propagation of measurements}
\label{sec:back-propagation}

Let us consider a Clifford+$T$ circuit $C$ and let us write its unitary part, before the readout measurements, as
\begin{equation*}
 U =\mathcal C_t\TGate_{\lambda_t}\mathcal C_{t-1}\cdots \mathcal C_1\TGate_{\lambda_1}\mathcal C_0 \,,
\end{equation*}
where each $\mathcal C_i$ is a fixed Clifford on the computational register and
$\lambda_{i}\in\{t+1,\ldots,t+n\}$ is the computational wire of the $i$th
$\TGate$ gate. We will refer to the $\mathcal C_{i}$ as computational Cliffords. In the adaptive Clifford circuit $C'$ obtained by magic-state injection, the $t$ magic-state qubits precede the $n$ computational qubits so that the $i$th gadget uses gadget wire $i$. We denote the gadget CNOTs as $\CNOT_{i} := \CNOT_{\lambda_{i},i}$ and write the corresponding adaptive $T$ gadget as
\begin{equation*}
A_i(m_i):=\SGate_{\lambda_{i}}^{m_i}\CNOT_i \,.
\end{equation*}
The layers $\mathcal C_i$ and the gadget CNOTs are nonadaptive, whereas the Clifford correction $\SGate_{\lambda_{i}}^{m_i}$ is applied after the measurement of gadget wire $i$.

Fix a provisional record $\str{m}=m_1\cdots m_{t+w}$. We write $\str g  = m_1\cdots m_t$ for the gadget outcomes and $\str y = m_{t+1}\cdots m_{t+w}$ for the readout outcomes. For fixed gadget outcomes $\str g$, we have that $C'$ is a fixed Clifford circuit.
The corresponding unitary layer is given by
\begin{equation*}
    U_{\textrm {out}}(\str g) :=\mathcal C_tA_t(m_t)\mathcal C_{t-1}\cdots \mathcal C_1A_1(m_1)\mathcal C_0 \,.
\end{equation*}
The unitary immediately after the $i$th gadget CNOT, and hence immediately before the corresponding $Z_i$ measurement and correction, is
\begin{equation*}
U_i(\str{g}) :={\CNOT}_i\mathcal C_{i-1}A_{i-1}(m_{i-1})\cdots  \mathcal C_1A_1(m_1)\mathcal C_0 \,,
\end{equation*}
for $i=1,\dots,t$, where dependence on the appropriate substring $\str g_{\PREC i} = m_{1}\cdots m_{\min(i-1,t)}$ is assumed. For convenience, we will drop the explicit dependence on the gadget record $\str g$ when it is evident from context. Additionally, for gadget indices $i < j$ we introduce the Clifford unitary $U_{j,i}$ defined by $U_{j,i} = U_{j}U_{i}^{\dagger}$. The symplectic maps induced by these Clifford unitaries are denoted $F_i:=F_{U_i}$ and $F_{\textrm {out}}:=F_{U_{\textrm {out}}}$ so that $G_{j,i} := F_{U_{j,i}} = F_{j}F_{i}^{-1}$.

For a complete outcome record $\str{m}\in\F{t+w}$ we have the Pauli observables
\begin{align*}
P_i & = U_i^\dagger Z_i U_i \,, \quad
         i=1,\ldots,t \,,\\
P_{t+r} & =U^{\dagger}_{\textrm {out}} Z_{t+r} U_{\textrm {out}} \,, \quad
         r=1,\ldots,w\,.
\end{align*}
Accordingly, their symplectic labels are given by
\begin{equation*}
\begin{aligned}
    p_i&=F_i^{-1}(z_i) \,,\\
    p_{t+r}&=F_{\mathrm{out}}^{-1}(z_{t+r}) \,.
\end{aligned}
\end{equation*}
The gadget observables are therefore not conjugated by one common Clifford. Independence and pairwise commutativity (proved below) therefore follow from the explicit structure of the gadgetized circuit.

\begin{lemma}[Pairwise commutativity]
\label{lem:sm-pairwise-commuting-t}
For a fixed gadget-outcome record, the labels
$p_1,\ldots,p_{t+w}$ generate an isotropic, i.e. are pairwise commuting.
\end{lemma}

\begin{proof}
Fix two gadget indices $i<j$ and recall that $F_{j} = G_{j,i}F_{i}$. After gadget wire $i$ is measured, no subsequent gate acts on it since later phase corrections act on the computational register and later CNOTs target later gadget wires. Hence, $G_{j,i}$ acts trivially on wire $i$. Symplectic invariance gives
\begin{equation*}
\begin{aligned}
    [p_i,p_j]
    &=[F_i^{-1}(z_i),F_j^{-1}(z_j)]\\
    &=[z_i,F_iF_j^{-1}(z_j)]
      =[z_i,G_{j,i}^{-1}(z_j)]=0 \,.
\end{aligned}
\end{equation*}
The final equality holds because $G_{j,i}^{-1}(z_j)$ has no $x_i$ component. The same argument applies to a gadget label and a readout label.  Writing
$F_{\mathrm{out}}=G_{\mathrm{out},i}F_i$, the intervening circuit again does
not act on gadget wire $i$, and hence
\begin{equation*}
[p_i,p_{t+r}] =[z_i,G_{\mathrm{out},i}^{-1}(z_{t+r})]=0 \,.
\end{equation*}
Finally, all readout labels are back-propagated through the same fixed
Clifford circuit, so their pairwise commutativity is immediate. This proves all three types of pairings.
\end{proof}

\begin{lemma}[Independence]
\label{lem:sm-independence}
For every fixed gadget-outcome record, the labels
$p_1,\ldots,p_{t+w}$ are independent.  Consequently,
\begin{eqnarray}
L:= \langle p_1,\ldots,p_{t+w}\rangle  \subseteq E_{t+n}\notag
\end{eqnarray}
is isotropic and has dimension $t+w$.
\end{lemma}

\begin{proof}
The proof is combinatorial. Begin by back-propagating an auxiliary $z_{i}$ associated with the $i$th gadget through the preceding CNOT controlled on some computational $c_{i}$, yielding the label $z_{i} + z_{c_{i}}$. Earlier computational Clifford gates and phase corrections act only on the computational register. Back-propagation through an earlier $h < i$ gadget CNOT can only add an
$x_{h}$ on that earlier gadget wire but cannot add an auxiliary $z_{h}$. Thus the gadget Pauli labels are independent. The readout Pauli labels $z_{t+1},\dots,z_{t+w}$ are transformed by the same symplectic map to $p_{t+1},\dots,p_{t+w}$ so that they remain independent. To see that they are each independent from the gadget Pauli labels, observe that when a readout is back-propagated through a gadget CNOT, it can add an auxiliary $x_{i}$, but never a $z_{i}$. Thus no $\mathbb Z_{2}$-linear combination of gadget labels can produce a readout label or vice versa.
\end{proof}

\subsection{Elimination of stabilizer qubits}
\label{sec:sm-reduction-t}

The resource savings of PBC arise from a reduction procedure that compresses the computation to the $t$-qubit magic-state register; see, e.g.,~\cite{Bravyi2016trading-sm, Yoganathan2019quantum-sm}. We describe this reduction using symplectic vector spaces.

For a fixed outcome record $\str m$, organize the ordered Pauli labels $p_1,\ldots,p_{t+w}$ into the rows of a stabilizer tableau.  We use the superscript $(i)$ to denote the tableau after $i$ rows have been processed. Thus
\[
    p_h^{(0)}:=p_h \,,
    \qquad h=1,\ldots,t+w \,,
\]
and define $K_{0} := \{0\}$. Immediately before stage $i+1$, where $i=0,\ldots,t+w-1$, let
$p_h^{(i)}$, $h=1,\ldots,t+w$, denote the current tableau rows and define
\begin{equation*}
K_i := \langle p_1^{(i)},\ldots,p_i^{(i)}\rangle \,.
\end{equation*}
The computational stabilizers, i.e., stabilizers of the $n$-qubit input state, span the isotropic subspace
\[
K_C := \left\langle z_{t+1} ,\ldots,z_{t+n}\right\rangle ,
\]
and the running isotropic space is $L_i:=K_C+K_i$.

The following are invariants of the reduction:
\begin{align}
&[p_h^{(i)},p_k^{(i)}]=0, &&h,k=1,\ldots,t+w,
\label{eq:tableau-orthogonality-invariant}\\
&K_i \subseteq K_C^\perp.
\label{eq:processed-space-in-KC-perp}
\end{align}
In particular, since $K_i$ is generated by the first $i$ (suitably updated) tableau rows, Eq.~\eqref{eq:tableau-orthogonality-invariant} implies
\begin{equation}
p_h^{(i)}\in K_i^\perp, \qquad h\geq i+1.
\label{eq:tail-orthogonal-to-processed}
\end{equation}
Thus, the current row $p_{i+1}^{(i)}$ commutes with every previously processed row.

These properties hold initially at $i=0$.  They also show why it suffices to consider anticommutation with the computational generators.  Suppose that the current row $p_{i+1}^{(i)}$ anticommutes with some $g\in L_i$. Writing $g = d+r$, where $d\in K_C$ and $r\in K_i$, we obtain
\[
1 =[p_{i+1}^{(i)},g]=[p_{i+1}^{(i)},d]+[p_{i+1}^{(i)},r]=[p_{i+1}^{(i)},d],
\]
where $[p_{i+1}^{(i)},r]=0$ owing to Eq.~\eqref{eq:tail-orthogonal-to-processed}. This implies that there exists a computational generator $z_c$ satisfying $[p_{i+1}^{(i)},z_c]=1$. The anticommuting case can therefore always be paired with a computational $z_c$.

At stage $i+1$, there are two cases.  If $p_{i+1}^{(i)}\in K_C^\perp$, no computational elimination is required.  Set $p_h^{(i+1)}:=p_h^{(i)}$, $h=1,\ldots,t+w$, so that
\begin{equation*}
K_{i+1} = K_i+\langle p_{i+1}^{(i)}\rangle\,.
\end{equation*}
If instead
\[
p_{i+1}^{(i)}\notin K_C^\perp \,,
\]
choose a computational generator $z_c$ satisfying $[p_{i+1}^{(i)},z_c]=1$ and define
\[
e_{i+1}:=p_{i+1}^{(i)}+z_c\,.
\]
Apply the corresponding transvection $\tau_{e_{i+1}}$ to the complete tableau. This is precisely the symplectic analogue of the random Clifford unitary used in the standard PBC reduction procedure~\cite{Bravyi2016trading-sm,Yoganathan2019quantum-sm}. The current row is mapped to
\[
p_{i+1}^{(i+1)} = \tau_{e_{i+1}}(p_{i+1}^{(i)}) = z_c \,,
\]
and therefore $K_{i+1} = K_i+\langle z_c\rangle $.

To verify that the invariants propagate, first note that the transvection
fixes every $r\in K_i$, since
\[
[r,e_{i+1}] =[r,p_{i+1}^{(i)}]+[r,z_c]=0+0 \,.
\]
The first term vanishes by Eq.~\eqref{eq:tail-orthogonal-to-processed}, while the second vanishes because $K_i\subseteq K_C^\perp$.  Moreover, the transvection is symplectic,
so it preserves pairwise orthogonality, and its new processed generator $z_c$ belongs to $K_C\subseteq K_C^\perp$.  Thus Eqs.~\eqref{eq:tableau-orthogonality-invariant} and
\eqref{eq:processed-space-in-KC-perp} hold at stage $i+1$.

After all $t+w$ stages, the final processed space is
\begin{equation*}
K_{t+w} = \langle p_1^{\prime},\ldots,p_{t+w}^{\prime}\rangle \subseteq K_C^\perp,
\end{equation*}
where have defined $p'_{i} := p_{i}^{(t+w)}$. Each corresponding row therefore admits the decomposition
\[
p_i'=q_i\oplus\widehat z_i,
\qquad q_i\in E_{t}, \quad \widehat z_i\in K_C.
\]
Although the $(t+n)$-qubit tableau rows $p_{i}'$ are all independent, the corresponding projected $t$-qubit Pauli elements $q_i$ need not be. In the quantum computation, the dependent Pauli observables are not measured and their outcomes are inferred via classical post-processing. We discuss this reduction further in the next section.

\subsection{Number of quantum measurements}

We now analyze the Pauli measurements performed on the $t$-qubit magic state register. Since
\[
K_{t+w}\subseteq K_C^\perp,
\]
every $p\in K_{t+w}$ admits a decomposition
\[
p=q\oplus\widehat z, \qquad q\in E_t,\quad \widehat z\in K_C.
\]
Define the linear projection
\[
\pi:K_{t+w}\longrightarrow E_t, \qquad \pi(q\oplus\widehat z):=q.
\]
Its kernel consists of those Pauli elements in $K_{t+w}$ supported only on the computational register, e.g., those involved in the anticommuting cases. Therefore
\begin{equation*}
\ker\pi=K_{t+w}\cap K_C.
\end{equation*}

The projected tableau rows span $\pi(K_{t+w})$, and one can extract an ordered basis
from these rows that gives the independent Pauli measurements performed on the
magic-state register. Without loss of generality, $\ell$ of these will correspond to gadget measurements and $\ell'$ to readout measurements.
 Using the rank-nullity theorem, we have
\begin{align}
\ell+\ell'
&:=
\dim\pi(K_{t+w}) \notag\\
&=
\dim K_{t+w}-\dim\ker\pi \notag\\
&=
t+w-\dim(K_{t+w}\cap K_C) \,.
\label{eq:sm-measurement-rank-t}
\end{align}

Below we bound the number of independent measurements performed in PBC. These bounds are equivalent to those established in~\cite{pashayan2022fast-sm}. Here, we recover them directly from the PBC reduction and identify these generators with the $\ell$ gadget and $\ell'$ readout measurements performed on the magic-state register.

\begin{proposition}[PBC measurement-count bounds]
\label{prop:sm-measurement-bound-t}
For an $n$-qubit Clifford+$\TGate$ circuit containing $t$ $\TGate$ gates and $w$ readout measurements, the corresponding PBC contains $\ell+\ell'$ independent measurements on the $t$ magic-state qubits, where
\begin{equation*}
\max\{0,t+w-n\} \leq \ell+\ell' \leq t \,.
\end{equation*}
If $w=n$, the PBC contains precisely $t$ independent measurements.
\end{proposition}

\begin{proof}
Since $K_{t+w}\cap K_C\subseteq K_C$ and $\dim K_C=n$, we have from Eq.~\eqref{eq:sm-measurement-rank-t} that
\[
\begin{aligned}
    \ell+\ell'
    &=
    t+w-\dim(K_{t+w}\cap K_C)\\
    &\geq t+w-n \,.
\end{aligned}
\]
Combining this with the trivial bound $\ell+\ell'\geq0$ proves
\[
    \ell+\ell'\geq\max\{0,t+w-n\} \,.
\]
For the upper bound, note that $\pi(K_{t+w})$ is an isotropic subspace in $E_{t}$ and an isotropic subspace of $E_t$ has dimension at most $t$, thus proving the result.
\end{proof}
Note that the count above is for the conventional PBC before any optional,
application-specific removal of free virtual qubits.

\section{Symplectic proof of the PBC normal form}\label{sec:sm-normal-form-t}

Back-propagating the $k$th gadget Pauli $z_{k}$ through the magic-state injection circuit $C'$ yields $p_{k} := F^{-1}_{k}(z_{k})$. This can be written as
\begin{equation*}
    p_k= r_k+z_k
\end{equation*}
where $r_k := F^{-1}_{k}(z_{c_{k}})$ and $c_{k}$ is the control qubit of the $k$th gadget. This follows since no gates after the gadget CNOT act on the $k$th gadget wire.

\begin{lemma}[Single-flip]\label{lem:sm-sibling-t}
Let two partial gadget records $\str g^{(0)}$ and $\str g^{(1)}$ agree except
at an earlier gadget index $k$, where their values are $0$ and $1$.  For any
later gadget or readout element $p_i$,
\begin{equation}\label{eq:sm-sibling-t}
p_i(\str g^{(1)}) =
\sigma_k\left(p_i(\str g^{(0)})\right ) +[p_i(\str g^{(0)}),z_k]p_k.
\end{equation}
\end{lemma}

\begin{proof}
Let $\zeta_i$ denote the computational-basis Pauli measured at stage $i$ and define
\[
v_i:=G_{i,k}^{-1}(\zeta_i).
\]
The two records have the same circuit until the $k$th gadget.  Thus we have
\[
p_i(\str g^{(0)})=F_k^{-1}(v_i).
\]
Flipping the $k$th gadget outcome from zero to one inserts the phase correction on the computational wire $c_k$, whose symplectic action is
$\sigma_{c_{k}}:=\tau_{z_{c_k}}$.  Therefore
\begin{align}
p_i(\str g^{(1)}) 
&=F_k^{-1}\sigma_{c_{k}}(v_i) \notag\\
&=F_k^{-1}\!\left(v_i+[v_i,z_{c_k}]z_{c_k} \right) \notag\\
&=p_i(\str g^{(0)}) +[v_i,z_{c_k}]r_k,
\label{eq:sm-sibling-r-t}
\end{align}
where $r_k=F_k^{-1}(z_{c_k})$.

Symplectic invariance gives
\[
[v_i,z_{c_k}] = [p_i(\str g^{(0)}),r_k],
\]
and Lemma~\ref{lem:sm-pairwise-commuting-t} further implies
\[
[p_i(\str g^{(0)}),p_k]=0,
\]
where $p_k$ is the same in both branches.
Since $p_k=r_k+z_k$, it follows that
\[
[p_i(\str g^{(0)}),r_k] = [p_i(\str g^{(0)}),z_k].
\]
Substituting $r_k=p_k+z_k$ into Eq.~\eqref{eq:sm-sibling-r-t}, we obtain
\begin{eqnarray*}
p_i(\str g^{(1)})
&=&p_i(\str g^{(0)})+[p_i(\str g^{(0)}),z_k](p_k+z_k) \notag\\
&=&\sigma_k\!\left(p_i(\str g^{(0)})\right)+[p_i(\str g^{(0)}),z_k] p_k,
\end{eqnarray*}
where $\sigma_k=\tau_{z_k}$ is the symplectic action of the phase gate on gadget wire $k$. This is Eq.~\eqref{eq:sm-sibling-t}. Note that the same identity holds when the bit is flipped from one to zero.
\end{proof}

\begin{lemma}[Multi-flip]
\label{lem:sm-multiflip-t}
Let two gadget records $\str g$ and $\str g'$ differ exactly at the gadget
positions $k_1<\cdots<k_s < i$ for some later gadget or readout index $i$. Define $\str g^{(a)}$, $a=0,\dots,s$,
\begin{comment}
\begin{equation}
\str g^{(a)} := \str g+\sum_{v=1}^{a}e_{k_v}, \qquad a=0,\ldots,s
\label{eq:sm-intermediate-records-t}
\end{equation}
\end{comment}
as the records differing at positions $k_{1},\dots,k_{a}$ so that $\str g^{(0)}=\str g$ and $\str g^{(s)}=\str g'$. Introduce also
\begin{eqnarray*}
\alpha_a
&:=& [p_i(\str g^{(a-1)}),z_{k_a}],\\
v_a
&:=& p_{k_a}\left( \str g^{(a-1)}\right),\\
\end{eqnarray*}
Then
\begin{equation}
p_i(\str g')
= \sigma_{k_s}\cdots\sigma_{k_1} \left(p_i(\str g)\right) +\eta_i(\str g,\str g'),
\label{eq:sm-multiflip-t}
\end{equation}
where
\begin{equation}
\eta_i(\str g,\str g') = \sum_{a=1}^{s}\alpha_a v_a.
\label{eq:sm-multiflip-correction-t}
\end{equation}
In particular,
 $\eta_i(\str g,\str g')$ is a linear combination of gadget generators
preceding $p_i$, all evaluated on the final branch~$\str g'$.
\end{lemma}

\begin{proof}
For measurement record $\str g^{(a)}$ we denote
\[
x_a:=p_i(\str g^{(a)}), \qquad a=0,\ldots,s.
\]
Consider the consecutive records $\str g^{(a-1)}$ and $\str g^{(a)}$ differing only at $k_a$. Applying the single-flip lemma gives
\begin{equation}
x_a = \sigma_{k_a}(x_{a-1})+\alpha_a v_a.
\label{eq:sm-one-flip-step-t}
\end{equation}
We prove by induction on $h=1,\ldots,s$ that
\begin{equation}
x_h
= \sigma_{k_h}\cdots\sigma_{k_1}(x_0) +\sum_{a=1}^{h}\alpha_a v_a.
\label{eq:sm-multiflip-induction-t}
\end{equation}
For $h=1$, this is Eq.~\eqref{eq:sm-one-flip-step-t}.  Suppose that Eq.~\eqref{eq:sm-multiflip-induction-t} holds for $h-1$.  Applying Eq.~\eqref{eq:sm-one-flip-step-t} and using linearity gives
\begin{eqnarray}
x_h
&=& \sigma_{k_h}\cdots\sigma_{k_1}(x_0)\notag\\
&+& \sum_{a=1}^{h-1} \alpha_a\sigma_{k_h}(v_a) +\alpha_h v_h.
\label{eq:sm-multiflip-induction-step-t}
\end{eqnarray}
For every $a<h$, we have $k_a<k_h$.  The gadget Pauli label $v_a$ has trivial support on gadget wires with index greater than $k_a$ so that $[v_a,z_{k_h}]=0$. This implies that
\[
\sigma_{k_h}(v_a) = v_a+[v_a,z_{k_h}]z_{k_h} = v_a.
\]
Substitution into Eq.~\eqref{eq:sm-multiflip-induction-step-t} proves Eq.~\eqref{eq:sm-multiflip-induction-t} for $h$.  Taking $h=s$, together with $x_0=p_i(\str g)$ and $x_s=p_i(\str g')$, proves Eqs.~\eqref{eq:sm-multiflip-t} and \eqref{eq:sm-multiflip-correction-t}.

It remains to establish that $\eta_i(\str g,\str g')$ is a linear combination of gadget generators preceding $p_{i}$, all evaluated on branch $\str g'$. For this, observe that we evaluate $v_a$ with the positions $k_1,\ldots,k_{a-1}$ having already been changed to their values in $\str g'$.  These are precisely the positions at which $\str g$ and $\str g'$ differ before $k_a$.  Hence $v_{a}$ depends only on this preceding record that agrees with~$\str g'$. The additional terms $\eta_i(\str g,\str g')$ therefore consist of a strictly lower-triangular stabilizer row operation on the final branch.
\end{proof}

For a branch record $\str g$, we call the original back-propagated labels
$p_i(\str g)$ the raw labels.
 Define the corresponding normal-form labels by
\begin{equation}
    \widetilde p_i(\str g)
    :=
    \sigma_{\str g}\bigl(p_i(\str 0)\bigr),
    \qquad
    \sigma_{\str g}
    :=
    \prod_{k=1}^t\sigma_k^{m_k}.
    \label{eq:sm-phase-I-representatives-t}
\end{equation}
The factors in $\sigma_{\str g}$ commute because they are transvections along the mutually orthogonal vectors $z_1,\ldots,z_t$.

\begin{lemma}[Causal generator replacement]
\label{lem:sm-generator-replacement-t}
For every branch record $\str g$ and every $a=1,\ldots,t+w$, the raw and normal-form Pauli elements span the same subspace:
\begin{equation}
\begin{aligned}
    \left\langle
        p_1(\str g),\ldots,p_a(\str g)
    \right\rangle
    =
    \left\langle
        \widetilde p_1(\str g),\ldots,
        \widetilde p_a(\str g)
    \right\rangle .
\end{aligned}
    \label{eq:sm-prefix-span-t}
\end{equation}
Equivalently, the two ordered lists are related by stabilizer row operations preserving causality.
\end{lemma}

\begin{proof}
Fix $i\leq a$ and apply Lemma~\ref{lem:sm-multiflip-t} to the zero record and $\str g$. Only the positions of the gadget preceding $i$ can affect $p_i(\str g)$. The transvections associated with later gadget positions are fixed $p_i(\str 0)$ by causality, so the product of phase transvections can be extended to all $k=1,\ldots,t$.  The multi-flip lemma therefore gives
\[
p_i(\str g)
=
\widetilde p_i(\str g) +\eta_i(\str 0,\str g),
\]
where
 $\eta_i(\str 0,\str g) \in \left\langle p_1(\str g),\ldots,p_{i-1}(\str g) \right\rangle$.
Thus the $i$th raw label is the $i$th normal-form label plus a combination of strictly earlier raw labels. Induction on $i$ proves Eq.~\eqref{eq:sm-prefix-span-t}.
\end{proof}

Lemma~\ref{lem:sm-generator-replacement-t} permits us to run the PBC
reduction on the normal-form generators.  Indeed, immediately before any reduction
stage, replacing the current row by itself plus a combination of earlier
processed rows does not change whether computational elimination is required or not since the preceding rows are contained in $K_C^\perp$. Indeed, run the reduction of SM~\ref{sec:sm-reduction-t} on the reference branch and extract a causal basis on the final projected rows. Denote the retained indices by
\[
j_1<\cdots<j_\ell <j_{\ell+1}<\cdots<j_{\ell+\ell'},
\]
where $j_1,\ldots,j_\ell\leq t$ are gadget indices and $j_{\ell+1},\ldots,j_{\ell+\ell'}>t$ are readout indices.  Their genuine quantum outcomes form the record
\[
\str b = m_{j_1}\cdots m_{j_{\ell+\ell'}} \in\F{\ell+\ell'}.
\]

A gadget row not retained on the magic-state register is processed in the anticommuting case using a computational $Z$ pivot.  Its outcome is a bit classically simulated using a fair coin toss.  At the symplectic level, its two values induce the same transvection and differ only in the sign, so we choose the canonical value zero. Hence, the gadget record associated with  $\str b$ is defined by
\begin{equation*}
g_k(\str b)
:=
\begin{cases}
 m_{j_u}, 	& k=j_u,\quad u=1,\ldots,\ell,\\
 0,			& k\notin\{j_1,\ldots,j_\ell\}.
\end{cases}
\end{equation*}
We use the shorthand
\[
p_i(\str b):=p_i(\str g(\str b))
\]
and define
\begin{equation*}
\widetilde S(\str b) := \prod_{u=1}^{\ell} \SGate_{j_u}^{m_{j_u}},
\qquad
\sigma_{\str b} := \prod_{u=1}^{\ell} \sigma_{j_u}^{m_{j_u}}.
\end{equation*}
The phase corrections act only on the surviving gadget lines.

We now show that the reduction procedure is invariant across branches. For every symplectic map $F$ and every $e\in E_{t+n}$, we have
\[
F\tau_e=\tau_{F(e)}F
\]
since
\[
F\tau_e(v) = F(v)+[F(v),F(e)]F(e) = \tau_{F(e)}F(v).
\]
For each branch $\str b$, let $\widetilde p_h^{(i)}(\str b)$ denote the tableau rows obtained after processing the first $i$ rows of the normal-form list, with
\[
\widetilde p_h^{(0)}(\str b) := \widetilde p_h(\str b).
\]
Define also the processed space
\begin{equation*}
\widetilde K_i(\str b)
:=
\left\langle \widetilde p_1^{(i)}(\str b),\ldots, \widetilde p_i^{(i)}(\str b )\right\rangle .
\end{equation*}

\begin{lemma}[Branch-invariant reduction]
\label{lem:sm-branch-invariant-pivots-t}
Fix a deterministic rule for choosing a computational $Z$ pivot.  For every
branch $\str b$ and every reduction stage $i=0,\ldots,t+w$,
\begin{align}
\widetilde p_h^{(i)}(\str b)
&=
\sigma_{\str b} \bigl(\widetilde p_h^{(i)}(\str 0)\bigr), &&h=1,\ldots,t+w,
\label{eq:sm-stage-invariant-t}\\
\widetilde K_i(\str b)
&=
\sigma_{\str b} \bigl(\widetilde K_i(\str 0)\bigr).
\label{eq:sm-processed-space-invariant-t}
\end{align}
Consequently, computational elimination occurs at the same processed row on every branch, and the same computational pivot is selected.
\end{lemma}

\begin{proof}
At $i=0$, Eq.~\eqref{eq:sm-stage-invariant-t} follows from
Eq.~\eqref{eq:sm-phase-I-representatives-t}, while both processed spaces are
zero.  Suppose the two identities hold after processing the first $i$ rows.
Because $\sigma_{\str b}$ acts only on the magic-state register, it fixes
$K_C$ pointwise and therefore
\[
    \sigma_{\str b}(K_C^\perp)=K_C^\perp.
\]
We will show that
\[
    \widetilde p_{i+1}^{(i)}(\str b)\in K_C^\perp
    \quad\Longleftrightarrow\quad
    \widetilde p_{i+1}^{(i)}(\str 0)\in K_C^\perp.
\]
on every branch.
In the commuting case, no transvection is applied, so both identities
propagate immediately. In the anticommuting case, for every computational
generator $z_c$,
\begin{eqnarray*}
[\widetilde p_{i+1}^{(i)}(\str b),z_c]
&=&
[\sigma_{\str b} (\widetilde p_{i+1}^{(i)}(\str 0)), z_{c}] \notag\\
&=&
[\widetilde p_{i+1}^{(i)}(\str 0), \sigma_{\str b} (z_{c})] \notag\\
&=&
[\widetilde p_{i+1}^{(i)}(\str 0),z_c].
\end{eqnarray*}
The same computational pivots are therefore available, and the deterministic
rule selects the same $z_{c_{i+1}}$. Letting
\[
e_{i+1}(\str b)
:=
\widetilde p_{i+1}^{(i)}(\str b)+z_{c_{i+1}}.
\]
and since $\sigma_{\str b}$ fixes $z_{c_{i+1}}$ we have that
\[
e_{i+1}(\str b) = \sigma_{\str b}\bigl(e_{i+1}(\str 0)\bigr).
\]
This further implies that
\[
\tau_{e_{i+1}(\str b)} = \sigma_{\str b}\tau_{e_{i+1}(\str 0)}\sigma_{\str b}^{-1}.
\]
Applying the transvection along $e_{i+1}(\str b)$ to all subsequent rows with index $h > i+1$ then gives
\begin{eqnarray*}
\widetilde p_h^{(i+1)}(\str b)
&=&
\tau_{e_{i+1}(\str b)} \bigl(\widetilde p_h^{(i)}(\str b)\bigr)\\
&=&
\sigma_{\str b} \tau_{e_{i+1}(\str 0)} \bigl(\widetilde p_h^{(i)}(\str 0)\bigr)\\
&=&
\sigma_{\str b} \bigl(\widetilde p_h^{(i+1)}(\str 0)\bigr).
\end{eqnarray*}
Eq.~\eqref{eq:sm-processed-space-invariant-t} follows since all generators have the same form $\widetilde p_{i}(\str b) = \sigma_{\str b}(\widetilde p_{i}(\str 0))$. This completes the induction.
\end{proof}

After all $t+w$ rows have been processed, define the projected auxiliary
labels
\begin{equation*}
    q_h(\str b)
    :=
    \pi\bigl(\widetilde p_h^{(t+w)}(\str b)\bigr),
    \qquad h=1,\ldots,t+w.
\end{equation*}
Because $\sigma_{\str b}$ acts only on the magic-state register, it commutes with the projection:
\[
\pi\sigma_{\str b} = \sigma_{\str b}\pi.
\]
Lemma~\ref{lem:sm-branch-invariant-pivots-t} therefore gives
\begin{equation}
q_h(\str b)
=
\sigma_{\str b}\bigl(q_h(\str 0)\bigr), \qquad h=1,\ldots,t+w.
\label{eq:sm-projected-row-covariance-t}
\end{equation}
Since $\sigma_{\str b}$ is symplectic, it also preserves the dimension of the isotropic:
\begin{eqnarray*}
\dim\left\langle q_1(\str b),\ldots,q_h(\str b) \right\rangle
=
\dim\left\langle q_1(\str 0),\ldots,q_h(\str 0) \right\rangle .
\end{eqnarray*}
Consequently, when extracting an ordered basis, we retain the same indices $j_1,\ldots,j_{\ell+\ell'}$ on every branch.
 Dependent projected readout rows occur at the same positions and their outcomes are obtained by applying the corresponding Gaussian elimination row operations to the preceding quantum outcomes.

\begin{theorem}[Symplectic PBC NF]
\label{thm:sm-normal-form-t}
The genuine-measurement indices
$j_1,\ldots,j_{\ell+\ell'}$ are the same on every branch, and their reduced
auxiliary labels satisfy
\begin{equation}
q_{j_u}(\str b)
=
\sigma_{\str b}\bigl(q_{j_u}(\str 0)\bigr),
\qquad u=1,\ldots,\ell+\ell'.
\label{eq:sm-reduced-normal-form-t}
\end{equation}
\end{theorem}

\begin{proof}
Lemma~\ref{lem:sm-branch-invariant-pivots-t} proves covariance of the complete tableau reduction, while Eq.~\eqref{eq:sm-projected-row-covariance-t} proves covariance after restriction to the magic-state register.  Since ordered basis extraction retains the same indices on every branch, Eq.~\eqref{eq:sm-reduced-normal-form-t} follows.
\end{proof}

\section{Construction of the dual one-way quantum computation}
\label{sec:sm-dual-one-way-qc}

We now make the connection to one-way quantum computing (1WQC) by establishing Proposition~\ref{pro:1wqc-dual}.
The central idea is to exchange the roles of preparation and measurement in a manner preserving the desired outcome probabilities.
In particular, the reference Pauli measurements define a stabilizer resource, while the magic-state inputs become local measurement effects.
The main issue is to identify a \emph{causal} 1WQC measurement pattern with an invertible outcome map.
The binary matrices that we introduce below ensure that this is possible.

\subsection{Reference resource and branch probabilities}

We consider the situation where all qubits in the original circuit are measured, i.e. $w=n$.
Before any stabilizer reduction, the Pauli observables on the zero-branch are given by $P_i(\str 0)$ so that the corresponding zero-branch projector is
\[
\Pi_{\str 0}
:= \prod_{i=1}^{t+n}\frac{I+P_i(\str 0)}{2}
 = \ket{\mathcal R_{\str 0}}\bra{\mathcal R_{\str 0}}.
\]
That this projector is rank one follows from Lemma~\ref{lem:sm-independence}.
Moreover, there exists a set of outcome-generating Pauli observables $D(\str m)$, labeled by $\str m$, described more explicitly below, satisfying
\[
D(\str m)P_i(\str 0)D^\dagger(\str m) = (-1)^{m_i}P_i(\str 0),
\quad
i=1,\dots,t+n.
\]
Such Pauli operators always exist and are often chosen as conjugate elements, such as $X_i$ and $Z_i$, but our specific $D(\str m)$ need not be of this form.

Recall that the input state to the gadgetized circuit is
\[
\ket{\Omega} := \ket{T}^{\otimes t}\otimes \ket{0}^{\otimes n}.
\]
For an outcome record $\str m$, we have that the corresponding rank-one PBC projector is given by
\[
\Pi_{\str m} := \widetilde \Sigma(\str g)D(\str m)\Pi_{\str 0}D^\dagger(\str m)\widetilde \Sigma^\dagger(\str g),
\]
where $\widetilde \Sigma(\str g) := \prod_{i=1}^t S_i^{m_i}$. The probability in PBC of this outcome record to occur is given by
\begin{equation}
\text{Pr}(\str m) := \bra{\Omega} \widetilde \Sigma(\str g)D(\str m)\Pi_{\str 0}D^\dagger(\str m)\widetilde \Sigma^\dagger(\str g)\ket{\Omega}.
\label{eq:sm-probability-pbc}
\end{equation}
Here we have used from SM~\ref{sec:sm-normal-form-t} that, up to signs, the branch observables corresponding to $\str m$ can be written in canonical form as
\[
\widetilde P_i(\str m) = \widetilde \Sigma(\str g) P_i(\str 0) \widetilde \Sigma^\dagger(\str g).
\]

\subsection{Outcome-generating Pauli operators}

We consider the Pauli labels of the reference branch generators.
Let $\mathcal T_X$ and $\mathcal T_Z$ be the $X$-block and $Z$-block, respectively, of the corresponding stabilizer tableau $\mathcal T$.
More explicitly, writing rows as (gadgets, readouts) and columns as (magic, stabilizer), we have
\begin{eqnarray*}
\mathcal T_X=\begin{pmatrix}L&A\\B&C\end{pmatrix},
 \qquad
 \mathcal T_Z=\begin{pmatrix}I_t&D\\0&E\end{pmatrix}.
\end{eqnarray*}
Here $L\in\mathbb F_2^{t\times t}$ is strictly lower triangular,
$A,D\in\mathbb F_2^{t\times n}$,
$B\in\mathbb F_2^{n\times t}$, and
$C,E\in\mathbb F_2^{n\times n}$.\footnote{The binary matrix $C$ is distinct from the Clifford circuit $C$.}
This tableau is highly structured due to the causal ordering of the magic state gadgets.
In the $i$th gadget the measurement $Z_i$ is back-propagated through $CX_i$ yielding $Z_iZ_{c_i}$.
Further back-propagation can pick up $X$-factors due to earlier gadget $CX_h$, $h < i$, but it maintains the $Z_i$-factor throughout.
This explains the presence of $I_t$ in $\mathcal T_Z$ as well as the strictly lower triangular $L$ in $\mathcal T_X$.
Likewise, a readout measurement can pick up $X$-factors in the magic register, but not $Z$-factors, hence the zero block in $\mathcal T_Z$.

We seek a collection of Pauli operators that flips the signs of the reference branch observables to the desired outcome record $\str m$.
First, introduce the operator
\[
D(\str e, \str u, \str v) \simeq Y_G^{\otimes \str e}\otimes X_C^{\otimes \str u}Z_C^{\otimes \str v},
\]
where $G$ and $C$ represent the register of gadget and computational qubits, respectively, and
where $\simeq$ is equality up to phase. Additionally, we have $\str e \in \mathbb F_2^t$, $\str u,\str v \in \mathbb F_2^{n}$.
The corresponding Pauli label $d(\str e,\str u,\str v) \in E_{t+n}$ has $x$-coordinate $(\str e,\str u)$ and $z$-coordinate $(\str e,\str v)$.

To induce the desired sign flip on our reference branch, we must have
\[
D(\str e,\str u,\str v)\widetilde P_i(\str 0) = (-1)^{m_i}\widetilde P_i(\str 0)D(\str e,\str u,\str v),
\]
or, equivalently, that the symplectic product satisfies
\[
[\tilde p_i(\str 0),d(\str e,\str u,\str v)] = m_i.
\]
More explicitly, we have
\begin{eqnarray}
    \str g
    &=& (I_t + L)\str e + D\str u + A\str v,\notag \\
    \str y
    &=& B\str e + E\str u + C\str v.\notag
\end{eqnarray}
The binary matrix $K := I_t + L$ is unit lower triangular, and thus invertible over $\mathbb F_2$, making it central in constructing a causal measurement sequence; see e.g.,~\cite{Raussendorf2013contextuality-sm}.

\subsection{Measurement order and outcome map}

We wish to establish a linear mapping from the outcome record $\str m$ to the Pauli coordinates $(\str e,\str u,\str v)$.
For fixed $\str u$ and $\str v$ we can solve for
\[
\str e = K^{-1}(\str g +D\str u +A\str v),
\]
since $K$ is invertible.
Consequently, we also have
\[
\str y = BK^{-1}\str g + \overline E \str u + \overline C \str v,
\]
where
\[
\overline E := E + BK^{-1}D,\quad \overline C := C + BK^{-1}A.
\]

We now establish a useful simplification.
From the dual perspective, each computational qubit, labeled $t+1,\dots,t+n$, will be measured in the $Z$ basis.
Thus, we have $X^uZ^v\ket{0} \simeq \ket{u}$, independently of $v$.
We can therefore make the following convenient choice
\[
\str u = \str f,
\quad \text{and}\quad
\str v = H\str f.
\]
Furthermore, we write
\[
D(\mathbf e,\mathbf f):=D(\mathbf e,\mathbf f,H\mathbf f),
\qquad
D(\mathbf m):=D(\mathbf e(\mathbf m),\mathbf f(\mathbf m)).
\]

The computational $Z$ factors are invisible to the local $Z$ measurements, but they still change the reference generator signs, as desired.
We use this freedom to encode the remaining readout information through the linear relationship
\[
\str y = BK^{-1}\str g + J\str f,
\quad \text{where} \quad
J := \overline E+\overline C H.
\]
Crucially, as we show next, we can choose $H$ so that the binary matrix $J$ is invertible.

\begin{lemma}
\label{sm:1wqc-completion}
The matrix $(\overline C\mid\overline E)$ has row rank $n$.
There exists $H\in\mathbb F_2^{n\times n}$ such that
\begin{equation*}
 J:=\overline E+\overline C H
\end{equation*}
is invertible.
Such an $H$ can be found by Gaussian elimination and basis extension.
\end{lemma}

\begin{proof}
We first prove the rank statement.
Define the invertible matrix
\[
W
=
\begin{pmatrix}
I_t     &   0\\
BK^{-1} &   I_n\\
\end{pmatrix}.
\]
Applying $W$ to the reference stabilizer tableau $\mathcal T$, this map leaves the gadget measurements fixed while ensuring that all readout measurements have $I$ or $Y$ on the magic-state register. Thus, readouts commute with the magic-state part of $D(\str e, \str f)$.
More explicitly, the new tableau has the form
\begin{equation*}
 \begin{aligned}
 W\mathcal T_X
 &=
 \begin{pmatrix}
 L&A\\
 BK^{-1}&\overline C
 \end{pmatrix},\\
 W\mathcal T_Z
 &=
 \begin{pmatrix}
 I_t&D\\
 BK^{-1}&\overline E
 \end{pmatrix}.
 \end{aligned}
\end{equation*}

The matrix $(\overline C, \overline E)$ encodes the computational part of the readout measurements and we wish to show that this $n\times 2n$ binary matrix has row rank $n$.
To show this, we assume the opposite and prove a contradiction.
In particular, assume there is a nonzero vector $\mu \in \mathbb F_2^n$ such that
\[
\mu^T \overline C = \mu^T\overline E = 0.
\]
We use this to define the Pauli label
\[
\tau := \mu^T \left (BK^{-1}~\overline C~|~BK^{-1}~\overline E\right) = (( a^T,0),( a^T,0)),
\]
where $ a^T := \mu^TBK^{-1}$, which corresponds to a Pauli with trivial support on the computational register.
Because the transformed tableau remains isotropic, this label commutes with each gadget row.
The $i$th gadget row has auxiliary $x$ coordinates given by the strictly lower triangular $L$ and auxiliary $z$ coordinates given by row $i$ of $I_t$, so
\[
0=[\tilde p_i(\str0),\tau]
=\sum_{h=1}^t L_{ih}a_h+a_i
=(K a)_i.
\]
Thus $K a=0$, and invertibility of $K$ implies $ a=0$. Consequently $\tau=0$ on both registers.
This contradicts our assumption of full row rank.

Since $(\overline C\mid\overline E)$ has full row rank, there exists $H\in\mathbb F_2^{n\times n}$ such that $J=\overline E+\overline C H$ is invertible.
Indeed, $\overline E$ already supplies all directions modulo $\operatorname{im}\overline C$;
the term $\overline C H$ can be chosen to complete these to a basis of $\mathbb F_2^n$.
Such an $H$ can be constructed by Gaussian elimination and basis extension over $\mathbb F_2$, which is efficient in the number of qubits.
\end{proof}

For a suitable choice of $H$ as described above, we have that the readout outcomes are given by
\begin{eqnarray}
\str y
&=&
BK^{-1}\str g + \overline E\str f + \overline C H \str f\notag\\
&=&
BK^{-1}\str g + \left(E+CH + BK^{-1}(D+AH)\right)\str f.\notag
\end{eqnarray}
This motivates the definitions
\begin{eqnarray}
M
&=&
D+AH,
\quad
N
=
E+CH,
\quad
J
=
N+BK^{-1}M.\notag
\end{eqnarray}
We then have that the outcome-generating Pauli operators are given by
\[
D(\str e,\str f)
\simeq
Y^{\otimes\str e}\otimes X^{\otimes\str f}Z^{H\str f},
\]
and we have the linear map
\[
\begin{pmatrix}\str g\\\str y\end{pmatrix}
=
\begin{pmatrix}K&M\\B&N\end{pmatrix}
\begin{pmatrix}\str e\\\str f\end{pmatrix}.
\]
This block matrix is invertible since $K$ and $J$ are invertible.
Namely, observe that
\[
\begin{pmatrix}I_t&0\\BK^{-1}&I_n\end{pmatrix}
\begin{pmatrix}K&M\\B&N\end{pmatrix}
=
\begin{pmatrix}K&M\\0&J\end{pmatrix}.
\]
The left matrix on the left-hand side is invertible, as is the matrix on the right-hand side, since both $K$ and $J$ are invertible, thus implying the result.

\subsection{Causal adaptive local measurements}

We now use the constructions above to establish the 1WQC measurement sequence that will eventually yield the promised duality.
Define the orthonormal basis
\[
B_\alpha := \{S^\alpha \ket{T}, ZS^\alpha \ket{T}\},\quad \alpha \in \mathbb F_2,
\]
where we interpret the first element as outcome zero and the second as outcome one. Recalling Eq.~\eqref{eq:sm-probability-pbc}, for $e,m\in \mathbb F_2$, we have
\[
Y^e S^{-m}\ket{T}\simeq Z^m S^{e+m}\ket{T},
\]
where $e+m$ is evaluated modulo two.
This suggests that $\alpha := e+m$ determines the measurement angle while $m$ is interpreted as the outcome.
At first glance, this appears to be acausal since $m$ in PBC is a measurement outcome and it is determining the measurement angle $\alpha$, however, we now demonstrate that the whole construction is causal.
Define the separable state vector
\[
\ket{\phi(\str o, \str f)}
:=
\left (\bigotimes_{i=1}^t Z_i^{o_i} S_i^{\alpha_i} \ket{T}_i\right )\otimes \ket{f_1}\otimes \cdots\otimes \ket{f_n},
\]
so that the probability of outcome record $\str m$ is equivalently given by
\[
\text{Pr}(\str g,\str y) = \lvert \bra{\phi(\str o,\str f)} \mathcal R_{\str 0} \rangle \rvert^2
\]
From the above analysis, we identify the outcomes of local measurements on the gadget qubits, $i=1,\dots,t$, with $\str o = \str g$, while the measurement bases are determined by $\str \alpha = \str e + \str g$.
Eliminating $\str e$ for the observed outcomes $\str o$, we have
\[
\str \alpha = (K^{-1}+I_t)\str o + K^{-1}M\str f.
\]
The dual computation proceeds as follows:
\begin{enumerate}
    \item Measure all computational qubits, in no particular order, in $Z$ basis, obtaining $\str f$,
    \item Measure all gadget qubits, in injection order, in basis $B_{\alpha_i}$, according to $\str \alpha$, obtaining $\str o$.
\end{enumerate}
Since $K^{-1}+I_t$ is strictly lower triangular, this procedure is causal.
Thus, the 1WQC outcomes determine the PBC outcomes according to
\begin{equation}
\str g = \str o,
\quad
\str y = BK^{-1}\str o + J\str f,
 \label{eq:sm-1wqc-decoder}   
\end{equation}
and the inverse is given by
\begin{equation}
\str o = \str g,
\quad
\str f = J^{-1}(\str y + BK^{-1}\str g).
\label{eq:sm-1wqc-decoder-inverse}  
\end{equation}
The encoding/decoding between these outcomes is efficient due to Lemma~\ref{sm:1wqc-completion}, owing to the efficiency of constructing $H$ and $J$.

At a high level, these mappings allow us to construct the tree of 1WQC measurement patterns from data available from the PBC.
The vectors $\ket{\phi(\str o, \str f)}$ are the leaves of this tree and two distinct leaves either have different computational-register outcomes or differ at a first auxiliary outcome where they use opposite vectors in the same basis. They are therefore orthogonal and form a complete basis.

\subsection{Elimination of computational register}

The reduction of SM~\ref{sec:sm-reduction-t} has a direct interpretation as local $Z$ measurements on the computational qubits of the dual resource.
To see this, let $P$ be a current reference generator, with eigenvalue $\epsilon\in\{+1,-1\}$, whose Pauli label $p$ satisfies $[p,z_c]=1$.
Measuring $Z_c$ gives either eigenvalue $\mu\in\{+1,-1\}$ with probability $1/2$ and replaces the signed stabilizer $\epsilon P$ by $\mu Z_c$.
Every other generator $R$, with eigenvalue $\eta$, is updated as
\[
(R,\eta)\longmapsto
\begin{cases}
(R,\eta),&[R,Z_c]=0,\\
(RP,\eta\epsilon),&\{R,Z_c\}=0.
\end{cases}
\]
Here $R$ commutes with $P$, so $RP$ is Hermitian.
The updated generators commute with $Z_c$ and can therefore be restricted to its measured eigenspace.

This is precisely the PBC reduction step if we also remove the measured qubit.
Indeed, since the Pauli labels $[r,p]=0$, the transvection used in the reduction acts as
\[
\tau_{p+z_c}(r)
=r+[r,z_c](p+z_c).
\]
For a row anticommuting with $Z_c$, this gives $r+p+z_c$.
After deleting site $c$, the additional $z_c$ has no remaining support, so the resulting label agrees with that of the measurement update $RP$. The pivot itself becomes $z_c$ and disappears upon deleting the measured site.
Outcome-dependent signs are retained in the classical record.

Applying this argument successively identifies the computational elimination with local $Z$-measurement preprocessing of the resource.
For the complete pure resource tableau, a computational $Z_c$ commuting with every remaining generator has a deterministic outcome; otherwise the anticommuting-pivot rule above applies.
Since these computational measurements commute with $\widetilde\Sigma(\str g)$, they may be performed before the auxiliary measurements.

Crucially, owing to the symmetry of the 1WQC and PBC reduction steps, it is tempting to identify the simulated PBC pivot outcomes with the dual computational outcomes --- these need not coincide.
Rather, the outcome records are related by Eqs.~\eqref{eq:sm-1wqc-decoder} and~\eqref{eq:sm-1wqc-decoder-inverse}, together with the recorded row operations.
Whether such an identification can be made, while maintaining causality, has not been established.

\subsection{Reduced resource state}

We now show that the computational qubit measurements yield the same underlying resource state, up to efficient classical side processing. Since the $n+t$ qubit resource $\ket{\mathcal R_{\str0}}$ is a stabilizer state, its computational $Z$ measurements can be simulated efficiently: each successive outcome is either deterministic or uniformly random, as determined by the current stabilizer tableau.
As noted above, if there are $k$ anticommuting cases in the PBC procedure, then there will be $k$ computational $Z$ measurements that are random; the remaining are determined by this random sample.
Let $\str f$ collectively denote the resulting record, including both random and deterministic outcomes, and let $\pi_C(\str f)$ be its probability.

Conditioned on any supported record $\str f$, the remaining $t$ qubits are in a pure stabilizer state.
In particular, the unsigned tableau updates are independent of the measurement outcomes and only the generator signs change~\cite{AarosonGottesman2004improved-sm}.
These sign changes can, in turn, be implemented by an efficiently computable $t$-qubit Pauli observable $F(\str f)$.
By the correspondence with the PBC reduction established above,
we may therefore write
\begin{equation*}
 \left(
  I \otimes \bigotimes_{i=1}^n \bra{f_i}
\right)
\lvert \mathcal R_{\str0}\rangle
 \simeq
 \sqrt{\pi_C(\str f)} \, F(\str f)\ket{\Psi_0},
\end{equation*}
where $\ket{\Psi_0}$ is the fixed $t$-qubit stabilizer resource defined by the independent reduced reference generators $Q_{j_1}(\str 0),\dots,Q_{j_{\ell+\ell'}}(\str 0)$.

Thus, the reduction procedure samples $\str f$ classically, prepares $\ket{\Psi_0}$, and absorbs $F(\str f)$ into the auxiliary measurement bases. Explicitly, the local measurement vectors become
\[
 F_i(\str f)^\dagger Z_i^{o_i}\SGate_i^{\alpha_i}\ket T_i,
 \qquad
 F(\str f)=\bigotimes_{i=1}^t F_i(\str f).
\]
A single-qubit Pauli merely exchanges the two $T/ST$ bases and/or swaps their outcome labels.
These changes are known before the auxiliary measurements begin, so locality and causality are preserved.

\section{Details on distributed quantum computation}\label{sm: distributed qc}

Here, we show that any Pauli measurement can be split into two QPUs using the strategy depicted in Fig.~\ref{fig: Distributed QC} of our Letter, provided there is preshared entanglement between them in the form of Bell pairs. This is in the spirit of recent architectural proposals, such as seen in Ref.~\cite{He2025extractors-sm}. 

\begin{proposition}[Measurement distribution]
    Let $P= P_A \otimes P_B\,,$ with $P_A \in \P{k}$ and $P_B \in \P{n}$, be a generic $(n+k)$-qubit Pauli observable that we want to measure across two different QPUs with, respectively, (at least) $k + 1$ and $n + 1$ qubits. Assume that the two QPUs share a Bell pair $\ket{\Phi^{+}}_{12} \coloneqq (\ket{00}_{12} + \ket{11}_{12})/\sqrt{2}$. Then, it is possible to distribute the measurement of $P$ across the two QPUs using the following procedure:
    \begin{enumerate}
        \item Measure commuting observables $Q_{a} = P_A \otimes I_B \otimes Z_1\otimes I_2$ and $Q_{b} = I_A \otimes P_B \otimes I_1 \otimes Z_2$, and record the corresponding outcomes $m_a,m_b \in \F{}\,.$
        \item Measure commuting observables $Q_{c} = I_A \otimes I_B \otimes X_1 \otimes I_2$ and $Q_{d} = I_A \otimes I_B \otimes I_1 \otimes X_2$, and record the corresponding outcomes $m_c,m_d \in \F{}\,.$
        \item Apply a Pauli correction when $m_c \oplus m_d = 1 \,.$
    \end{enumerate}
\end{proposition}
\begin{proof}
    For the correctness of the protocol, two conditions must be guaranteed:
    \begin{itemize}
        \item[(i)] The final state of the $(n+k)$ computational qubits after the quantum measurement must be the same in the distributed and non-distributed cases.
        \item[(ii)] The outcome probability distribution of the distributed computation must be the same as in the non-distributed case.
    \end{itemize}
Let $\ket{\Psi_{\text{in}}}$ denote the generic $(n+k)$-qubit quantum state before the measurement takes place. Trivially, in the non-distributed case, we have
    \begin{equation}\label{eq: Psi_out non-dist}
        \ket{\Psi_{\text{out}}^{(m)}} = \frac{1}{\sqrt{p(m)}} \Pi_{P}^{(m)} \ket{\Psi_{\text{in}}}_{AB}\,,
    \end{equation}
    where
    \begin{equation*}
        p(m) = \frac{1}{2}
        + \frac{(-1)^m}{2}\,
        \left\langle \Psi_{\text{in}}
        \middle| P \middle|
        \Psi_{\text{in}} \right\rangle
\end{equation*}
and $\Pi_A^{(a)} \coloneqq [ I + (-1)^a A]/2\,$.

We now show that the distributed procedure returns the same output state for the computational qubits with the same output probability distribution.
We start with the first two measurements, after which we obtain
    \begin{align*}
        & \ket{\Psi^{(m_a,m_b)}} \propto \Pi^{(m_b)}_{Q_b} \Pi^{(m_a)}_{Q_a} \ket{\Psi_{\text{in}}}_{AB}\ket{\Phi^{+}}_{12} \\
        & = \frac{1}{2} \left[ \Pi_{P}^{(m_{a}\oplus m_{b})}\left|\Psi_{\text{in}}\right\rangle_{AB} \right] \otimes \left|\Phi^{+}\right\rangle_{12} + \\ 
        & \frac{(-1)^{m_{a}}}{2} \left[ \left( P_{A}\otimes I_{B} \right)\Pi_{P}^{(m_{a}\oplus m_{b})}\left|\Psi_{\text{in}}\right\rangle_{AB} \right] \otimes \left|\Phi^{-}\right\rangle_{12} \,.
    \end{align*}
Thus, after these two measurements, the auxiliary qubits are entangled with the computational ones. To disentangle `$AB$' from `$12$', we measure $Q_c = I_{AB}\otimes I_1\otimes X_2$ and $Q_d = I_{AB}\otimes X_1\otimes I_2$. Doing so projects the computational qubits onto the desired state if $m_c \oplus m_d = 0\,.$ On the other hand, if the $m_c \oplus m_d = 1\,,$ an additional Pauli unitary $(-1)^{m_a} P_A\otimes I_B$ is present and needs to be corrected for. Mathematically, we write:
    \begin{align*}
        \ket{\Psi_{\text{out}}^{(m_a,m_b,m_c,m_d)}} \propto \Pi^{(m_d)}_{Q_d} \Pi^{(m_c)}_{Q_c} \ket{\Psi^{(m_a,m_b)}}\hspace{20mm} \\
        = \begin{cases}
            m' = 0: \, \frac{1}{2} \left( \Pi_{P}^{(m_{a}\oplus m_{b})} \left|\Psi_{\text{in}}\right\rangle_{AB} \right) \otimes \ket{\eta^{(m_c,m_d)}}_{12}\\
            m' = 1: \, \frac{1}{2} \left( U \Pi_{P}^{(m_{a}\oplus m_{b})}\left|\Psi_{\text{in}}\right\rangle_{AB} \right)  \otimes \ket{\eta^{(m_c,m_d)}}_{12} 
        \end{cases}\hspace{-2mm},
    \end{align*}
    where $m' =  m_c \oplus m_d$, $U = P_{A}\otimes I_{B}$ is a Pauli unitary acting only on the computational qubits on the first QPU, and $\ket{\eta^{(m_c,m_d)}}_{12} = Z_1^{m_c} \otimes Z_2^{m_d}\left|++\right\rangle_{12}$.
   
    Finally, applying the Pauli correction $U^{\dagger}$ (if needed), tracing out the auxiliary system, and assigning $m \coloneqq m_a \oplus m_b$ leads exactly to the result in Eq.~\eqref{eq: Psi_out non-dist} with the same output probability distribution. This concludes the proof.
\end{proof}
After the distribution procedure has been executed, the auxiliary qubits are disentangled from the main computational qubits, so they can be reinitialized and reused. We note that if only a single Bell pair is shared between two QPUs at any given time, then only one Pauli measurement can be distributed at a time, leading to a serialization of the PBC: Even measurements belonging to the same adaptive layer must be performed sequentially, one after another. Alternatively, if more Bell pairs are preshared between the QPUs, parallelization of same-layer measurements remains possible provided the distributed measurements still commute. The idea is illustrated via the simple example in Fig.~\ref{fig: multiple Bell pairs}.
Consider two commuting measurements $P=P_A\otimes P_B$ and $Q=Q_A\otimes Q_B$. 
Distributing them leads to two possibilities: either (I)~$[P_A,Q_A] = [P_B,Q_B] = 0\,,$ or (II)~$\{P_A,Q_A\} = \{P_B,Q_B\} = 0\,.$
In the first case, the distributed measurements can be performed simultaneously, followed by the single-qubit $X$ measurements. Consequently, a single layer of the non-distributed PBC leads to two layers in the distributed scenario.
In the second case, anticommutativity implies that the distributed measurements must be performed sequentially. Hence, the advantage of having more preshared Bell pairs is only that the $X$ measurements are parallelizable, leading to a linear improvement in the depth of the distributed computation as compared to the case where only one Bell pair is shared between the QPUs.
\begin{figure}
    \centering
    \includegraphics[width=0.9\linewidth]{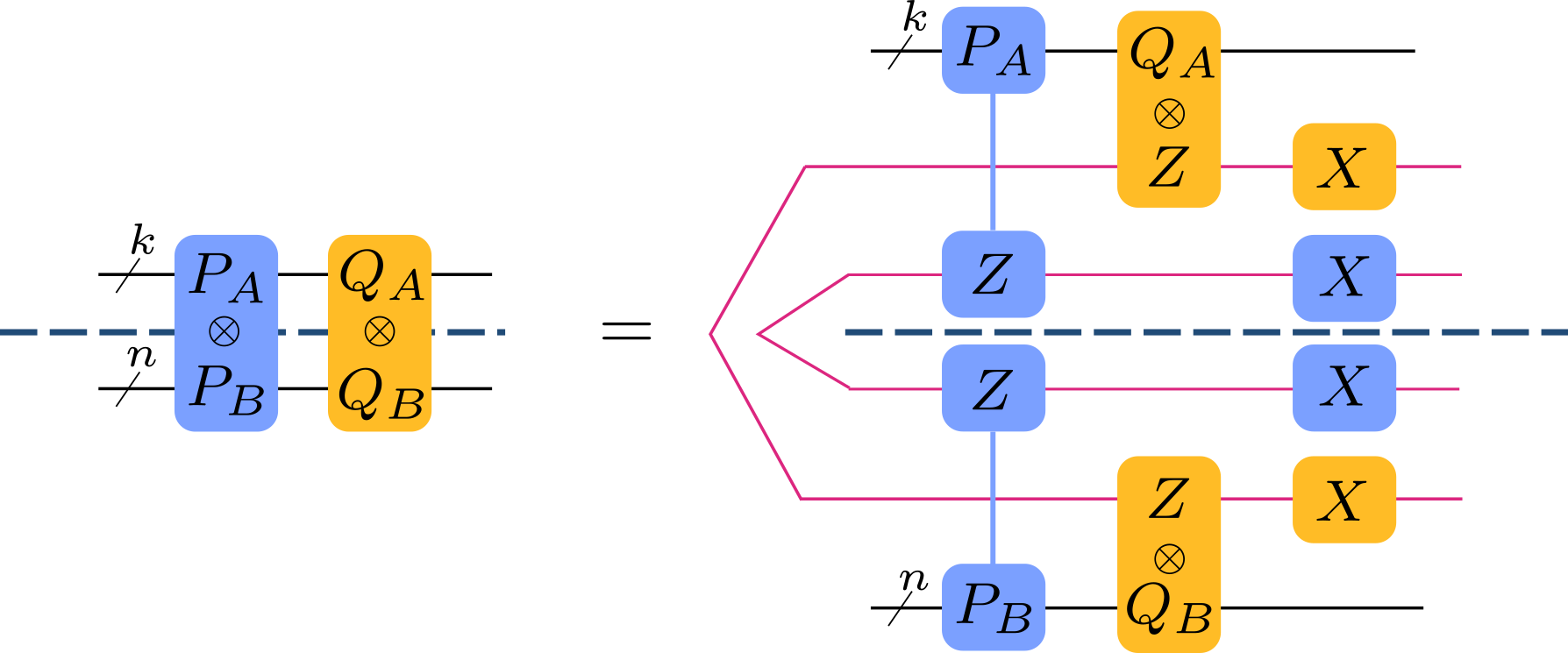}
    \caption{\textbf{Measurement distribution with multiple Bell pairs.} If two Bell pairs are available rather than only one, measurement serialization is avoided provided $[P_A,Q_A] = [P_B,Q_B] = 0\,.$ In this case, the construction depicted still allows the distributed measurements to be performed simultaneously.}
    \label{fig: multiple Bell pairs}
\end{figure}

More generally, for a layer with more than two measurements, cases (I) and (II) still hold. If the distributed measurements all pairwise commute, then the layer of the non-distributed PBC turns into two layers, provided a suitable number of Bell pairs is available; if the distributed measurements all anticommute, they need to be performed sequentially, and having the appropriate number of Bell pairs contributes only to parallelizing the otherwise sequential $X$ measurements. Furthermore, an intermediate regime exists between (I) and (II) where the distributed Pauli measurements are divided into sets of pairwise commuting elements that can still be simultaneously measured, but are followed sequentially by the measurement of subsequent sets whose elements anticommute with previously performed measurements.

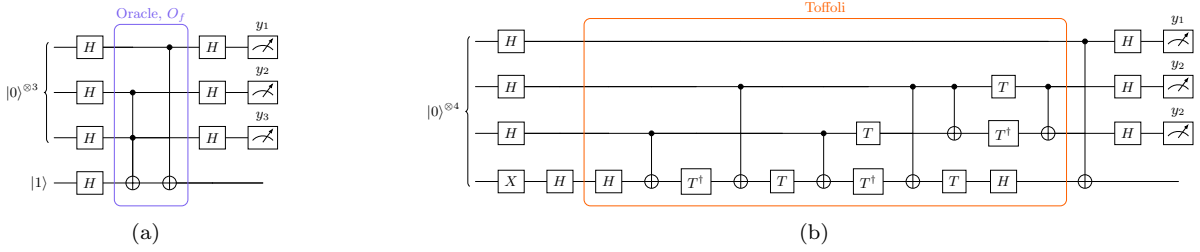
\begin{figure*}[t]
    \centering
    \subfloat[]{\begin{tikzpicture} \node[scale=0.6]{ \begin{quantikz}[thin lines] \lstick[wires=3]{$\ket{0}^{\otimes 3}$} &\gate{H} &\qw\gategroup[wires=4, steps=2, style={draw=IBMpurple, rounded corners}, label style={color=IBMpurple}]{Oracle, $O_f$} &\ctrl{3} &\gate{H} &\meter{y_1}\\ &\gate{H} &\ctrl{2} &\qw &\gate{H} &\meter{y_2} \\ &\gate{H} &\ctrl{1} &\qw &\gate{H} &\meter{y_3}\\ \lstick{$\ket{1}$} &\gate{H} &\targ{} &\targ{} &\qw &\qw \end{quantikz} }; \end{tikzpicture}
    }\hspace{15mm}
    \subfloat[]{\begin{tikzpicture} \definecolor{IBMorange}{RGB}{254,97,0} \node[scale=0.6]{ \begin{quantikz}[thin lines] \lstick[wires=4]{$\ket{0}^{\otimes 4}$} &\gate{H} &\qw &\qw\gategroup[wires=4, steps=11, style={draw=IBMorange, rounded corners}, label style={color=IBMorange}]{Toffoli} &\qw &\qw &\qw &\qw &\qw &\qw &\qw &\qw &\qw &\qw &\ctrl{3} &\gate{H} &\meter{y_1}\\ &\gate{H} &\qw &\qw &\qw &\qw &\ctrl{2} &\qw &\qw &\qw &\ctrl{2} &\ctrl{1} &\gate{T} &\ctrl{1} &\qw &\gate{H} &\meter{y_2} \\ &\gate{H} &\qw &\qw &\ctrl{1} &\qw &\qw &\qw &\ctrl{1} &\gate{T} &\qw &\targ{} &\gate{T^{\dagger}} &\targ{} &\qw &\gate{H} &\meter{y_2}\\ &\gate{X} &\gate{H} &\gate{H} &\targ{} &\gate{T^{\dagger}} &\targ{} &\gate{T} &\targ{} &\gate{T^{\dagger}} &\targ{} &\gate{T} &\gate{H} &\qw &\targ{} &\qw &\qw \end{quantikz} }; \end{tikzpicture}
    }
    \caption{\textbf{Circuit for solving an instance of the Deutsch-Josza problem.} \textbf{(a)~}~The circuit written using a Toffoli gate. \textbf{(b)~}~The Toffoli gate can be rewritten in terms of the Clifford+$T$ set, using $7$ $T$ gates.}
    \label{fig: Deutsch-Josza example}
\end{figure*}
\begin{figure*}[t]
    \centering
    \includegraphics[width=0.95\linewidth]{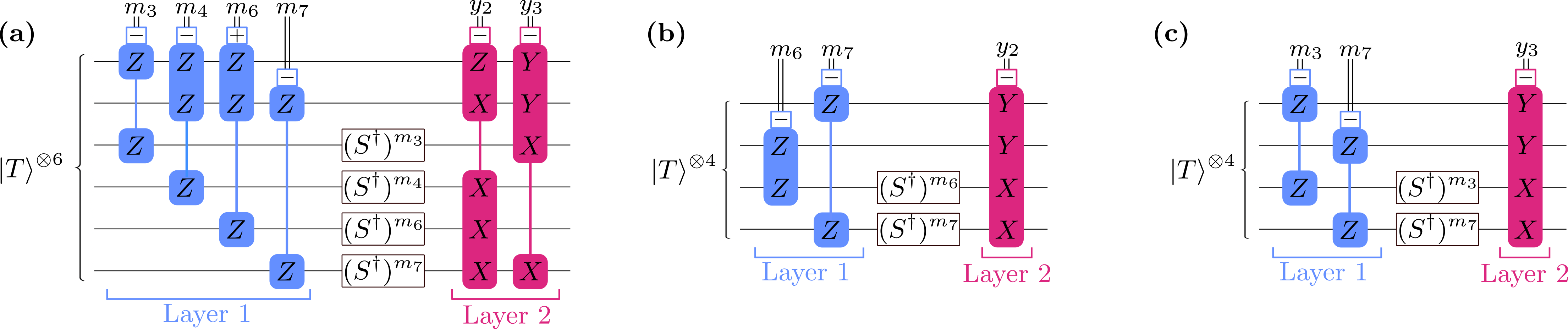}
    \caption{\textbf{PBCs associated with the Deutsch-Josza example under consideration.} \textbf{(a)} The PBC obtained when considering exactly the circuit in Fig.~\ref{fig: Deutsch-Josza example}(b). Alternatively, one can consider variants of the input circuit that include only a single readout measurement. This leads to the PBCs in \textbf{(b)} and \textbf{(c)} corresponding, respectively, to an input circuit measuring only the second and third qubits. Intriguingly, if we are only interested in sampling from the output distribution of the first readout, we can do so in a fully classical manner.}
    \label{fig: PBCs associated with the DJ example}
\end{figure*}

Thus, using the NF to understand the properties of the distributed measurements allows us to decide on the appropriate number of preshared Bell pairs either for time-, space-, or space-time optimal implementations.

\section{Compiling a small Deutsch-Josza circuit}\label{sm: DJ}

For pedagogical purposes, we provide a fully worked-out example illustrating how our NF enables ahead-of-time~(AOT) quantum circuit compilation with PBC. Importantly, we explicitly sketch the procedure for doing \emph{enhanced} AOT PBC compilation, effectively removing the uncompute blocks needed in the naive approach (compare Figs.~\ref{fig: summary fig}(d) and (e) of our Letter). We end with a discussion of how this procedure is key to achieving better gate counts and depth than those of the input circuit. 

Our starting point is the 4-qubit quantum circuit depicted in Fig.~\ref{fig: Deutsch-Josza example}(a), which solves a particular instance of the Deutsch-Josza problem. We defer details concerning this problem to other references, e.g.,~\cite[Secs.~1.4 and 6.1]{NielsenChuang-sm}. For the present discussion, it suffices to know that the goal is to determine whether a given function $f: \F{n} \mapsto \F{}$ is constant or balanced and that the Deutsch-Josza algorithm is constructed in such a way that if the outcome of the quantum circuit is $\str{y} = \str{0}$, $f$ is constant, whereas if $\str{y} \neq \str{0}$, $f$ is balanced. This solves the problem in a single query, while classically, exponentially many queries are generally necessary for a deterministic solution. Remarkably, we will see below that, for the example considered here, our NF enables immediate, fully classical identification of the function as balanced.

The first step in our AOT PBC compilation method consists of converting the circuit in Fig.~\ref{fig: Deutsch-Josza example}(a) to a new circuit written in terms of the Clifford+$T$ gate set. We do this by decomposing the Toffoli gate using the optimal, nonadaptive decomposition by Shende and Markov~\cite{ShendeMarkov08-sm}, which requires 7 $T$ gates [Fig.~\ref{fig: Deutsch-Josza example}(b)]. We then input this quantum circuit to our code~\cite{GitHubPACTS-sm}, which leverages Theorem~\ref{thm:pbc-normal-form} and Proposition~\ref{prop: free qubit removal} of our Letter, to efficiently obtain all the information we need to describe the NF PBC tree [Fig.~\ref{fig: summary fig}(b) of the main text]. We represent the NF PBC in circuit form in Fig.~\ref{fig: PBCs associated with the DJ example}(a).

\begin{figure*}
    \includegraphics[width=0.95\linewidth]{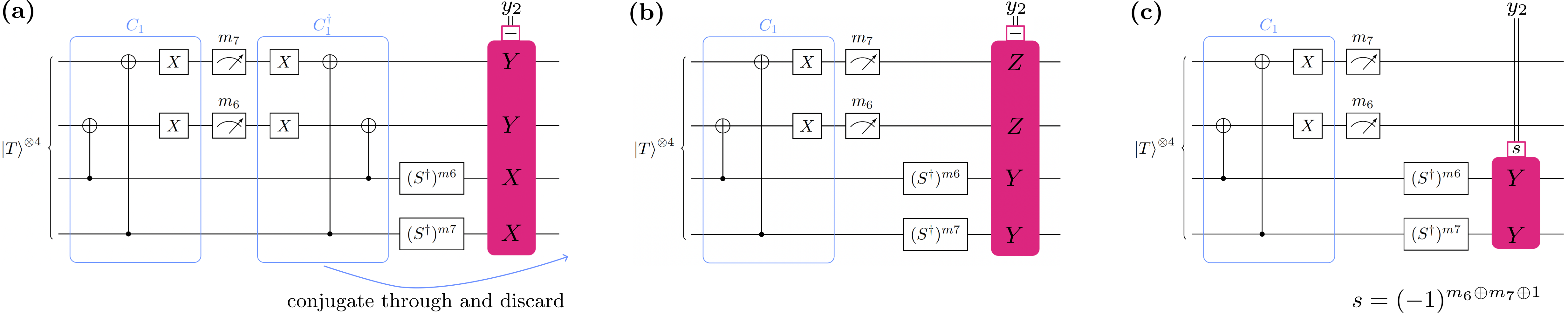}
    \caption{\textbf{Step-by-step construction the adaptive Clifford circuit.} \textbf{(a)~}~The first layer of Pauli measurements $\{-Z_1Z_4,\, -Z_2Z_3\}$ is recast as a set of two computational basis measurements conjugated by a Clifford unitary $C_1$. \textbf{(b)}~$C^{\dagger}_1$ commutes with $(S^{\dagger}_3)^{m_6} (S^{\dagger}_4)^{m_7}\,,$ is used to update the Pauli measurement $-Y_1Y_2X_3X_4$ under conjugation ($-Y_1Y_2X_3X_4 \longmapsto -Z_1Z_2Y_3Y_4$), and is then dropped from the circuit. \textbf{(c)}~Because the first and second qubits are in a computational basis state, performing the measurement $-Z_1Z_2Y_3Y_4$ is equivalent to performing the measurement $(-1)^{m_6\oplus m_7 \oplus 1} Y_3Y_4$. This observation is extremely profitable for reducing the quantum memory size for many circuit instances.}
    \label{fig: step-by-step}
\end{figure*}

\begin{figure}
    \centering
    \begin{tikzpicture} \definecolor{IBMblue}{RGB}{100,143,255} \definecolor{IBMpink}{RGB}{220,38,127} \node[scale=0.6]{ \begin{quantikz}[thin lines] \lstick[wires=4]{$\ket{T}^{\otimes 4}$} &\qw\gategroup[wires=4, steps=3, style={draw=IBMblue, rounded corners}, label style={color=IBMblue}]{$C_1$} &\targ{} &\gate{X} &\meter{m_7} &\ctrl[vertical wire=c]{1}\setwiretype{c} &\ctrl[vertical wire=c]{1} \\ &\targ{} &\qw &\gate{X} &\meter{m_6} &\ctrl[vertical wire=c]{2}\setwiretype{c} &\ctrl[vertical wire=c]{1} \\ &\ctrl{-1} &\qw &\qw &\qw &\gate{(S^{\dagger})^{m_6}} &\gate{X^{m_6 \oplus m_7 \oplus 1}}\gategroup[wires=2, steps=4, style={draw=IBMpink, rounded corners}, label style={color=IBMpink}]{$C_2$} &\gate{S^{\dagger}} &\targ{} &\qw \\ &\qw &\ctrl{-3} &\qw &\qw &\gate{(S^{\dagger})^{m_7}} &\qw &\gate{S^{\dagger}} &\ctrl{-1} &\gate{H} &\meter{y_2} \end{quantikz} }; \end{tikzpicture}
    \caption{\textbf{Final AOT PBC-compiled circuit.} This circuit can be used to sample the output distribution of the second readout qubit of the original Deutsch-Josza circuit in Fig.~\ref{fig: Deutsch-Josza example}, requiring fewer gates and depth.}
    \label{fig: final Deutsch-Josza AOT compiled circuit}
\end{figure}

Here, several features are noteworthy. First, the PBC has 6 qubits rather than 7; this highlights that one of the magic-state qubits has been identified as a free virtual qubit (Proposition~\ref{prop: free qubit removal}) and removed from the PBC. Second, the sequence of 6 Pauli measurements is divided into only 2 adaptive layers: The first layer is comprised of Pauli measurements stemming from magic-state injection gadgets, while the second layer holds two Pauli measurements corresponding to the readouts of the second and third qubits in the quantum circuit. This takes us to the third (and significant) observation: The readout measurement of the first qubit is identified by our NF procedure as a measurement that can be decided classically and deterministically returns the outcome 1. This is remarkable because it means that the PBC NF immediately allows us to solve the problem classically, without needing any quantum resources: If $y_1=1$ deterministically, then $\str{y} \neq \str{0} \iff f\text{ is balanced}\,.$ This observation suggests that our NF will likely become an intriguing new framework to investigate the dequantization of quantum algorithms~\cite{Tang2019quantuminspired-sm} and identify novel classes of quantum circuits amenable to efficient classical simulation.

The observations in the previous paragraph motivate the following question: What if, rather than sampling from the full output distribution of the circuits in Figs.~\ref{fig: Deutsch-Josza example}(a) and \ref{fig: Deutsch-Josza example}(b), $p(\str{y})$, we are content with sampling from the output distribution of each of the readout qubits $p(y_i)$? Put differently, what if the input circuit to our NF PBC procedure has only one readout measurement? From the perspective of solving the Deutsch-Josza problem, measuring only the first qubit would suffice, as it deterministically yields $1$. In that case, as mentioned above, Theorem~\ref{thm:pbc-normal-form} and Proposition~\ref{prop: free qubit removal} combine to inform us that no quantum computation is needed and that the problem can be solved fully classically. In general, this would conclude the algorithm and we would move on. However, for the purposes of illustrating how AOT PBC compilation works via a small example, we will now consider what happens if we are only interested in sampling from the output distribution of either the second or the third qubit of the input circuit. In these cases, our code returns the NF PBC trees depicted pictorially in Figs.~\ref{fig: PBCs associated with the DJ example}(b) and  \ref{fig: PBCs associated with the DJ example}(c). Importantly, we note that two more free virtual qubits were identified, effectively reducing the size of the quantum memory from 6 to 4. This is another intriguing feature enabled by our NF: The number of readout qubits in a given quantum circuit may considerably influence the size of the quantum memory needed to implement the corresponding sampling task.

Since the circuits in Figs.~\ref{fig: PBCs associated with the DJ example}(b) and \ref{fig: PBCs associated with the DJ example}(c) are essentially the same up to a permutation of the top two qubits, we will now focus on the former to describe step-by-step how the NF can be used to produce an adaptive PBC-compiled Clifford circuit with magic state inputs. In this case, because there are two adaptive layers in the PBC, the compiled circuit will be comprised of two Clifford subcircuits, which we dub $C_1$ and $C_2$. For the two Pauli measurements in the first layer $\{-Z_1Z_4,\, -Z_2Z_3 \}$, we find the Clifford circuit $C_1$ that maps them to computational basis measurements so that: $C_1 \{-Z_1 Z_4 ,\, -Z_2 Z_3 \}  C_1^{\dagger} = \{Z_1,\,Z_2\}\,.$ Using standard algorithms for Clifford synthesis~\cite{AarosonGottesman2004improved-sm, Bravyi2021clifford-sm, Bravyi2021hadamardfree-sm}, it is not hard to show that $C_1 = X_1 X_2 CX_{4,1} CX_{3,2}\,.$ We explicitly redraw the circuit in Fig.~\ref{fig: PBCs associated with the DJ example}(b) as seen in Fig.~\ref{fig: step-by-step}(a).

The next step consists of propagating $C_1^{\dagger}$ through the rest of the circuit and discarding it at the end. Because qubits 1 and 2 have been measured in the computational basis, they will be in one of the computational basis states. It is easy to see that $C_1^{\dagger}$ commutes with the layer of $S^{\dagger}$ gates. This is a crucial condition for removing the uncompute blocks. It can also be easily shown that updating the Pauli observable $-Y_1Y_2X_3X_4$ in the second layer leads to a new observable that acts trivially on these two qubits, as seen in Fig.~\ref{fig: step-by-step}(b), where $C_1^{\dagger}$ has been discarded after reaching the end of the quantum circuit. This result holds beyond this example (see SM~\ref{sm: Removing uncompute block}) and is central to enabling enhanced AOT PBC compilation with smaller quantum memories: Since $-Z_1Z_2Y_3Y_4$ acts trivially on qubits 1 and 2, it can be compressed to act only on qubits 3 and 4, provided we classically track its sign~[Fig.~\ref{fig: step-by-step}(c)]. 

Finally, we again use standard Clifford synthesis algorithms to find the second subcircuit $C_2$ so that: $C_2 \{ (-1)^{m_6\oplus m_7 \oplus 1} Y_3Y_4 \} C_2^{\dagger} = \{Z_4\}\,.$ This leads to the final AOT PBC-compiled quantum circuit depicted in Fig.~\ref{fig: final Deutsch-Josza AOT compiled circuit}.

Compared to the original Clifford+$T$ input circuit from Fig.~\ref{fig: Deutsch-Josza example}(b), this circuit uses 3 CNOT gates rather than~7, an average total number of gates of~17.5 rather than~21, and an average quantum depth of~10.25 rather than~12. (For a fair comparison, the total number of gates and the quantum depth take into consideration the transformation $TH$ that is needed to create the $\ket{T}$ state from the $\ket{0}$ state). This shows the power of the midcircuit measurements, leveraged via our NF.

\section{Ahead-of-time compilation without the `uncompute' block}\label{sm: Removing uncompute block}

SM~\ref{sm: DJ} provided a pedagogical illustration of how one might perform AOT PBC compilation using our NF result. Crucially, in that example, we saw that $C^{\dagger}_1$ was constructed in such a way as to commute with the subsequent layer of adaptively-chosen $S^{\dagger}$ corrections. This is an essential feature that these unitaries must possess in order to allow them to be used to update ensuing reference branch Pauli measurements and afterwards be discarded from the circuit.

We will now demonstrate that it is always possible to construct Clifford unitaries that have this property. This is at the core of the \emph{enhanced} AOT PBC compilation discussed in our Letter. 

Consider the reference branch of a generic $t$-qubit PBC, $\mathcal{B}_{\str{0}}\,,$ is comprised of Pauli observables $\{Q_{j_1},\dots,Q_{j_{\ell + \ell^{\prime}}}\}\,,$ with $1 \leq \ell + \ell^{\prime} \leq t\,,$ where we have omitted the explicit dependence on the all-zero outcome string to alleviate notation.
\begin{lemma}[Trivial support]\label{lemma: trivial support}
     Consider the $j_a$th Pauli measurement of $\mathcal{B}_{\str{0}}$, $Q_{j_a}$, and denote its support $\mathcal{S}_{j_a}\,.$ Let $W_a \in \Cl{t}$ be the Clifford unitary so that:
     \begin{equation*}
        W_a Q_{j_a} W^{\dagger}_a = Z_{s_{j_a}}\,,
    \end{equation*}
     for some $s_{j_a}\in \mathcal{S}_{j_a}\,.$ Then, for every upcoming Pauli measurement $Q_{j_{b}}$ with $a < b \leq \ell + \ell^{\prime}\,,$ it holds that:
     \begin{equation}\label{eq: trivial support}
         W_a Q_{j_b} W_a^{\dagger} = Z_{s_{j_a}}^c \otimes R_{j_b} \quad,\forall b\,,
     \end{equation}
     where $c\in\F{}$ and $R_{j_b}$ is a $(t-1)$-qubit Pauli operator not supported in $s_{j_a}\,.$
\end{lemma}
\begin{proof}
    Since $[Q_{j_a},\, Q_{j_b}] = 0 \,\forall b$ and because commutativity is preserved under conjugation, it must hold that
    \begin{align*}
         [W_aQ_{j_a}W_a^{\dagger},\, W_aQ_{j_b}W_a^{\dagger}] & = 0 \quad,\forall b \\
        \iff [Z_{s_{j_a}},\, W_aQ_{j_b}W_a^{\dagger}] & = 0 \quad,\forall b \,.
    \end{align*}
    The last condition necessarily implies the shape of the updated Paulis as written in Eq.~\eqref{eq: trivial support}. 
\end{proof}

In other words, if a Pauli observable from the reference branch is mapped, under conjugation by a Clifford unitary, to a computational basis measurement, then updating ensuing Pauli observables of that same branch under conjugation by that same unitary will lead to a set of updated observables that have trivial support on the previously measured qubit. This is exactly what we saw in Fig.~\ref{fig: step-by-step}. 

\begin{figure*}
    \centering
    \includegraphics[width=0.99\linewidth]{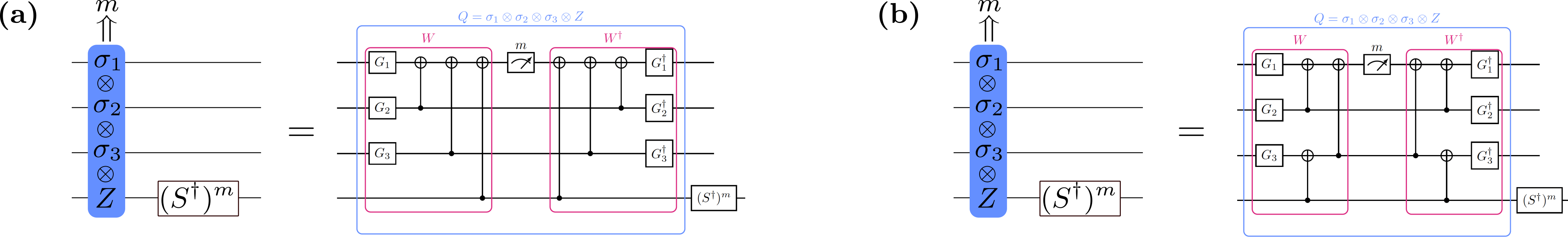}
    \caption{\textbf{Strategies for creating $W$ unitaries.} The Pauli measurement can be performed using a compute-and-uncompute strategy with \textbf{(a)} a linear-depth CNOT cascade structure or \textbf{(b)} a logarithmic-depth CNOT tree structure. The single-qubit Clifford gates $G_i$ are such that $G_i \sigma_i G_i^{\dagger} = Z_i\,,$ where $\sigma_i = \{X, Y,Z\}$ and, correspondingly, $G_i = \{H, HS^{\dagger}, I\}\,.$}
    \label{fig: Construction of W unitaries}
\end{figure*}

\begin{lemma}[Commutativity with phase corrections]\label{lemma: commutativity}
    Consider the $j_a$th Pauli measurement of $\mathcal{B}_{\str{0}}\,,$ $Q_{j_a}\,.$ Let $\mathcal{S}_{j_a}$ denote its support and $s_{j_a}$ label the smallest integer in $\mathcal{S}_{j_a}$. Then, there exists a Clifford unitary $W_a \in \Cl{|\mathcal{S}_{j_a}|}$ so that
    \begin{equation*}
        W_a Q_{j_a} W^{\dagger}_a = Z_{s_{j_a}} \text{ and } [ W_a,\, S_{j_u}^{\dagger} ] = 0 \,,\forall  a \leq u \leq \ell \,.
    \end{equation*}
\end{lemma}
\begin{proof}
    The existence of a $W_a$ meeting the first condition is a trivially known fact. Therefore, proving the lemma reduces to demonstrating that the second condition holds. We prove this constructively. 

    Since readout measurements have no ensuing adaptive $S^{\dagger}$ corrections, we can focus only on gadget measurements. Suppose, wlog, that $Q_{j_a} = R_{j_a} \otimes Z \otimes I^{\otimes (t - j_a - 1)} \,,$ meaning that it corresponds to the $a$th gadget measurement of the quantum circuit. Then, a simple $W_a$ that meets both conditions consists of a sequence of CNOT gates targeting qubit $s_{j_a}$ and controlled by each of the other qubits in the support of $Q_{j_a}\,,$ preceded by suitable single-qubit Cliffords. The construction is depicted in Fig.~\ref{fig: Construction of W unitaries}(a). It trivially guarantees that $[ W_a,\, S_{j_u}^{\dagger} ] = 0\,,$ for $u > a$, since $W_a$ is not supported in qubits beyond qubit $j_a\,.$ That $W_a$ commutes with $S^{\dagger}_{j_a}$ is guaranteed by the fact that only CNOT controls act on the $j_a$th qubit. 

    An alternative construction for $W_a$ exists that uses logarithmic, rather than linear, depth in the support. This construction is depicted in Fig.~\ref{fig: Construction of W unitaries}(b) and similarly guarantees the commutativity condition for the same reasons as the previous construction.
\end{proof}

Lemmas~\ref{lemma: trivial support} and \ref{lemma: commutativity} are at the heart of \emph{enhanced} AOT circuit compilation with PBC. The latter tells us how to construct suitable $W$ unitaries for each gadget measurement and the former that the updated ensuing Pauli operators will be trivially supported in whichever qubits have already been measured. This typically lowers the support of the measurements as the computation progresses (recall the simple Deutsch-Josza example explained previously).

Unlike gadget measurements, readout measurements are not hostage to limitations in the shape of the associated Clifford unitaries, because there are no more layers of adaptive $S^{\dagger}$ corrections. This gives us a bit more freedom in the way we treat them. Thus, using the ingredients explained here, we follow three different strategies for \emph{enhanced} AOT circuit compilation with PBC:

\vspace{2mm}\noindent \underline{\textbf{Strategy 1.}}
\begin{enumerate}
    \item Gadget Pauli measurements are individually mapped to a computational basis measurement using the linear-depth construction in Fig.~\ref{fig: Construction of W unitaries}(a).
    \item Readout measurements are similarly taken individually and sequentially, using the same linear-depth construction.
\end{enumerate}

\noindent \underline{\textbf{Strategy 2.}}
\begin{enumerate}
    \item Gadget Pauli measurements are individually mapped to a computational basis measurement using the linear-depth construction in Fig.~\ref{fig: Construction of W unitaries}(a).
    \item Readout measurements are taken jointly. We synthesize the associated Clifford subcircuit using different well-known algorithms~\cite{AarosonGottesman2004improved-sm, Bravyi2021clifford-sm, Bravyi2021hadamardfree-sm}, and choose the one that performs the best.
\end{enumerate}

\noindent \underline{\textbf{Strategy 3.}}
\begin{enumerate}
    \item Gadget Pauli measurements are individually mapped to a computational basis measurement using the logarithmic-depth construction in Fig.~\ref{fig: Construction of W unitaries}(b).
    \item Readout measurements are similarly taken individually and sequentially, using the same logarithmic-depth construction.
\end{enumerate}

The main drawback of these strategies is that they do \emph{not} take advantage of the adaptive structure of the PBC, in the sense that all gadget measurements are serialized, even if they belong to the same adaptive layer. This comes from the limitation of having to construct unitaries that commute with ensuing $S^{\dagger}$ gates, a condition that is easy to guarantee in general for a single measurement but harder for sets of multiple measurements. A trivial extension of the single-measurement case happens when two measurements have disjoint support (as seen, for instance, in the small Deutsch-Josza example), in which case, they can be mapped by a single Clifford (with a tensor-product form) to computational-basis measurements occurring in parallel. Beyond that, the development of more generic Clifford synthesis algorithms that operate under the commutativity condition might lead to interesting trade-offs in the performance of enhanced AOT PBC compilation.

To gain a better understanding of the advantages and disadvantages between handling measurements one at a time or jointly, we use the readouts. In Strategy 1, we handle them individually and sequentially, while in Strategy 2, we take them all together. We comment on these trade-offs in SM~\ref{sm: More on AOT compilation performance}, supported by extensive numerical results obtained with random quantum circuits~(RQCs).

\begin{figure*}[t]
    \centering
    \includegraphics[width=0.95\linewidth]{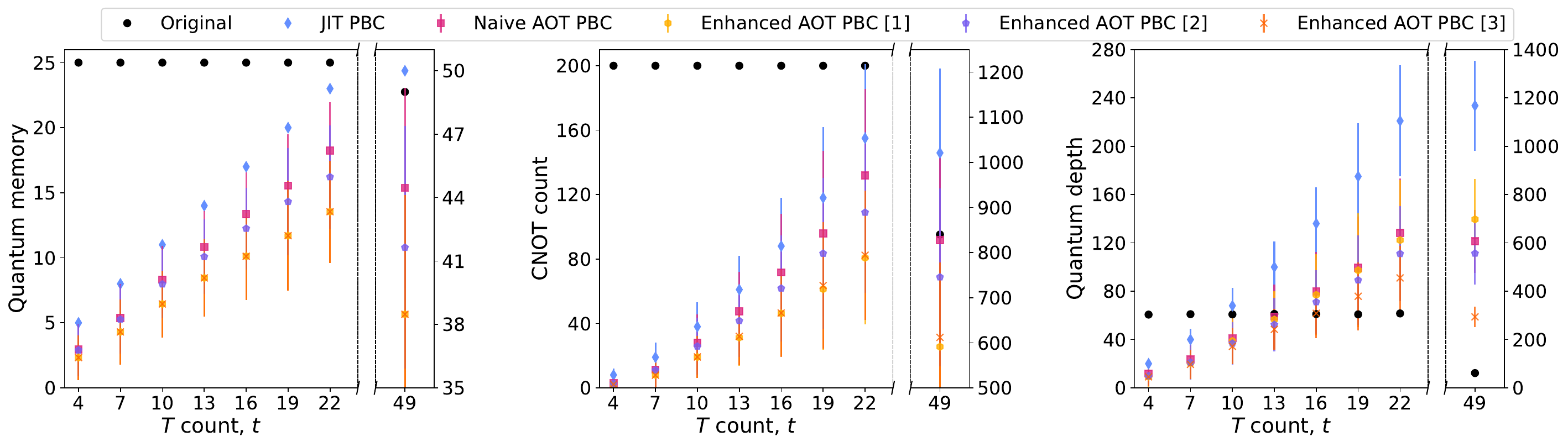}
    \caption{\textbf{Performance of the different PBC compilers.} We take the RQCs analyzed in Refs.~\cite{Peres2023quantum-sm,Peres2025reducing-sm}; these are 25-qubit RQCs with $T$ count varying between 4 and 22 and 49-qubit circuits with $t=49$. We see that, in general, the various enhanced versions of the AOT PBC compiler significantly outperform both JIT and naive AOT PBC compilation. We emphasize the quantum depth, where JIT PBC had clear limitations; the enhanced AOT PBC using Strategy 3 takes significant strides towards overcoming those, reducing the average quantum depth of JIT PBC-compiled circuits by 50\% for the smaller 25-qubit RQCs and by 75\% for the larger 49-qubit ones. Uncertainty bars represent 2 standard deviations.}
    \label{fig: All compilation results}
\end{figure*}

\begin{table*}[t]
    \centering
    \begin{tabular}{c|cccc|cccc|ccc}
    \hline\hline
    \multirow{2}{*}{Circuit} & \multicolumn{4}{c|}{Original circuits} & \multicolumn{4}{c|}{ZX-compiled circuits} & \multicolumn{3}{c}{Enhanced AOT PBC compiler} \\
     & $n$ & $t$ & $CX$ count & $d$ & $n$ & $t$ & $CX$ count & $d$ & $n$ & $CX$ count & $d$ \\
    \hline
    Deutsch-Josza\textunderscore n4 & $4$ & $7$ & $7$ & $16$ & $4$ & $7$ & $7$ & $16$ & \multicolumn{3}{c}{Function type is classically decidable!} \\
    Adder\textunderscore n4 & $4$ & $8$ & $10$ & $21$ & $4$ & $8$ & $10$ & $16$ & $2$ & $3$ & $9$ \\
    Adder\textunderscore n10 & $10$ & $56$ & $65$ & $176$ & $10$ & $32$ & $57$ & $134$ & $9$ & $43$ & $87$ \\
    SAT\textunderscore n11 & $11$ & $294$ & $252$ & $646$ & $11$ & $126$ & $214$ & $503$ & $12$ & $448$ & $670$ \\
    Multiply\textunderscore n13 & $13$ & $42$ & $40$ & $57$ & $13$ & $32$ & $40$ & $56$ & $7$ & $20$ & $47$ \\
    Multiplier\textunderscore n15 & $15$ & $252$ & $222$ & $354$ & $15$ & $70$ & $180$ & $244$ & $9$ & $92$ & $149$ \\
    QRAM\textunderscore n20 & $20$ & $140$ & $136$ & $188$ & $20$ & $96$ & $124$ & $187$ & $8$ & $98$ & $179$ \\
    \hline\hline
    \end{tabular}
    \caption{\textbf{Performance of our AOT PBC compiler for benchmark circuits.} The first column identifies the circuit under analysis. The ensuing four columns summarize the properties of the original circuits. We compiled the circuits using the ZX-calculus compiler~\cite{KissWet2020reduceT-sm, Wetering2020zxcalculus-sm, Staudacher2023-sm} and our AOT PBC compiler. For all circuits except one, we see an improvement in all circuit metrics: quantum memory size, CNOT gate count, and quantum depth.}
    \label{tab: AOT compiler benchmarks}
\end{table*}

\section{Further analysis of the performance of AOT PBC compilation}\label{sm: More on AOT compilation performance}

\subsection{Random quantum circuits}

In our Letter, we reported the improvements in the performance of the PBC compiler by turning it from a just-in-time~(JIT) compiler into an AOT one. To avoid confusion, we presented only the results for one of the three possible enhanced AOT strategies described in SM~\ref{sm: Removing uncompute block} (namely, Strategy 1). However, different strategies can lead to different trade-offs. We take this opportunity to analyze and discuss these details.

In Fig.~\ref{fig: All compilation results}, we report our results in full for the suite of RQCs that had been previously analyzed in Refs.~\cite{Peres2023quantum-sm, Peres2025reducing-sm}, yielding the results for the JIT PBC compiler as depicted using the blue diamonds. We notice that every AOT version (even the naive one) outperforms the JIT PBC compiler, which, in turn, already significantly outperformed the state-of-the-art ZX compilers~\cite{KissWet2020reduceT-sm, Wetering2020zxcalculus-sm, Staudacher2023-sm} in terms of CNOT gate count for all values of $t\,.$ The exception to this happens for $t=49$, where the ZX compiler achieved a result of $795\pm 26$ against the $1021\pm 187 $ of the JIT PBC compiler~\cite{Peres2025reducing-sm}. Thus, we now see that the AOT strategies significantly improve over state-of-the-art compilation results. 

A feature where the PBC compiler has traditionally underperformed is in the depth of compiled RQCs. We see that Strategy 3 for enhanced AOT PBC achieves significant improvements in depth with only a slight increase in the total number of CNOT gates as compared to Strategy 1. Interestingly, considering all readouts jointly leads to worse quantum memory and CNOT gate count results, indicating that the serialized approach comes with added benefits from the successive support reduction of the Pauli measurements. 

\subsection{Benchmark quantum circuits}

We also decided to check how AOT PBC compilation performed on a set of benchmark circuits. The results are presented in Table~\ref{tab: AOT compiler benchmarks}, where we see that, for most circuits, the enhanced AOT PBC compiler achieves significant improvements in all assessed measures: quantum memory size, CNOT gate count, and quantum depth. The single exception is the SAT circuit, for which PBC considerably underperforms, returning a circuit that is more resource-intensive than both the original and the ZX-compiled circuit. This is curious and points to some fundamentally different structure of this circuit (as compared to the others) that makes it unsuitable for PBC compilation. We defer a deeper study and understanding of the structures enabling or disabling good PBC compilation results to future work.


\begin{thebibliography}{53}%
\makeatletter
\providecommand \@ifxundefined [1]{%
 \@ifx{#1\undefined}
}%
\providecommand \@ifnum [1]{%
 \ifnum #1\expandafter \@firstoftwo
 \else \expandafter \@secondoftwo
 \fi
}%
\providecommand \@ifx [1]{%
 \ifx #1\expandafter \@firstoftwo
 \else \expandafter \@secondoftwo
 \fi
}%
\providecommand \natexlab [1]{#1}%
\providecommand \enquote  [1]{``#1''}%
\providecommand \bibnamefont  [1]{#1}%
\providecommand \bibfnamefont [1]{#1}%
\providecommand \citenamefont [1]{#1}%
\providecommand \href@noop [0]{\@secondoftwo}%
\providecommand \href [0]{\begingroup \@sanitize@url \@href}%
\providecommand \@href[1]{\@@startlink{#1}\@@href}%
\providecommand \@@href[1]{\endgroup#1\@@endlink}%
\providecommand \@sanitize@url [0]{\catcode `\\12\catcode `\$12\catcode `\&12\catcode `\#12\catcode `\^12\catcode `\_12\catcode `\%12\relax}%
\providecommand \@@startlink[1]{}%
\providecommand \@@endlink[0]{}%
\providecommand \url  [0]{\begingroup\@sanitize@url \@url }%
\providecommand \@url [1]{\endgroup\@href {#1}{\urlprefix }}%
\providecommand \urlprefix  [0]{URL }%
\providecommand \Eprint [0]{\href }%
\providecommand \doibase [0]{https://doi.org/}%
\providecommand \selectlanguage [0]{\@gobble}%
\providecommand \bibinfo  [0]{\@secondoftwo}%
\providecommand \bibfield  [0]{\@secondoftwo}%
\providecommand \translation [1]{[#1]}%
\providecommand \BibitemOpen [0]{}%
\providecommand \bibitemStop [0]{}%
\providecommand \bibitemNoStop [0]{.\EOS\space}%
\providecommand \EOS [0]{\spacefactor3000\relax}%
\providecommand \BibitemShut  [1]{\csname bibitem#1\endcsname}%
\let\auto@bib@innerbib\@empty
%</preamble>
\bibitem [{\citenamefont {Bravyi}\ \emph {et~al.}(2016)\citenamefont {Bravyi}, \citenamefont {Smith},\ and\ \citenamefont {Smolin}}]{Bravyi2016trading}%
  \BibitemOpen
  \bibfield  {author} {\bibinfo {author} {\bibfnamefont {S.}~\bibnamefont {Bravyi}}, \bibinfo {author} {\bibfnamefont {G.}~\bibnamefont {Smith}},\ and\ \bibinfo {author} {\bibfnamefont {J.~A.}\ \bibnamefont {Smolin}},\ }\href {https://doi.org/10.1103/PhysRevX.6.021043} {\bibfield  {journal} {\bibinfo  {journal} {Phys. Rev. X}\ }\textbf {\bibinfo {volume} {6}},\ \bibinfo {pages} {021043} (\bibinfo {year} {2016})}\BibitemShut {NoStop}%
\bibitem [{\citenamefont {Zurel}\ \emph {et~al.}(2020)\citenamefont {Zurel}, \citenamefont {Okay},\ and\ \citenamefont {Raussendorf}}]{Zurel2020hidden}%
  \BibitemOpen
  \bibfield  {author} {\bibinfo {author} {\bibfnamefont {M.}~\bibnamefont {Zurel}}, \bibinfo {author} {\bibfnamefont {C.}~\bibnamefont {Okay}},\ and\ \bibinfo {author} {\bibfnamefont {R.}~\bibnamefont {Raussendorf}},\ }\href {https://doi.org/10.1103/PhysRevLett.125.260404} {\bibfield  {journal} {\bibinfo  {journal} {Phys. Rev. Lett.}\ }\textbf {\bibinfo {volume} {125}},\ \bibinfo {pages} {260404} (\bibinfo {year} {2020})}\BibitemShut {NoStop}%
\bibitem [{\citenamefont {Zurel}\ \emph {et~al.}(2024)\citenamefont {Zurel}, \citenamefont {Okay},\ and\ \citenamefont {Raussendorf}}]{Zurel2024simulating}%
  \BibitemOpen
  \bibfield  {author} {\bibinfo {author} {\bibfnamefont {M.}~\bibnamefont {Zurel}}, \bibinfo {author} {\bibfnamefont {C.}~\bibnamefont {Okay}},\ and\ \bibinfo {author} {\bibfnamefont {R.}~\bibnamefont {Raussendorf}},\ }\href {https://doi.org/10.1103/PRXQuantum.5.030343} {\bibfield  {journal} {\bibinfo  {journal} {PRX Quantum}\ }\textbf {\bibinfo {volume} {5}},\ \bibinfo {pages} {030343} (\bibinfo {year} {2024})}\BibitemShut {NoStop}%
\bibitem [{\citenamefont {Yoganathan}\ \emph {et~al.}(2019)\citenamefont {Yoganathan}, \citenamefont {Jozsa},\ and\ \citenamefont {Strelchuk}}]{Yoganathan2019quantum}%
  \BibitemOpen
  \bibfield  {author} {\bibinfo {author} {\bibfnamefont {M.}~\bibnamefont {Yoganathan}}, \bibinfo {author} {\bibfnamefont {R.}~\bibnamefont {Jozsa}},\ and\ \bibinfo {author} {\bibfnamefont {S.}~\bibnamefont {Strelchuk}},\ }\href {https://doi.org/10.1098/rspa.2018.0427} {\bibfield  {journal} {\bibinfo  {journal} {Proc. R. Soc. A}\ }\textbf {\bibinfo {volume} {475}},\ \bibinfo {pages} {20180427} (\bibinfo {year} {2019})}\BibitemShut {NoStop}%
\bibitem [{\citenamefont {Bejan}\ \emph {et~al.}(2024)\citenamefont {Bejan}, \citenamefont {McLauchlan},\ and\ \citenamefont {B\'eri}}]{Bejan2024dynamical}%
  \BibitemOpen
  \bibfield  {author} {\bibinfo {author} {\bibfnamefont {M.}~\bibnamefont {Bejan}}, \bibinfo {author} {\bibfnamefont {C.}~\bibnamefont {McLauchlan}},\ and\ \bibinfo {author} {\bibfnamefont {B.}~\bibnamefont {B\'eri}},\ }\href {https://doi.org/10.1103/PRXQuantum.5.030332} {\bibfield  {journal} {\bibinfo  {journal} {PRX Quantum}\ }\textbf {\bibinfo {volume} {5}},\ \bibinfo {pages} {030332} (\bibinfo {year} {2024})}\BibitemShut {NoStop}%
\bibitem [{\citenamefont {Peres}\ and\ \citenamefont {Galv{\~{a}}o}(2023)}]{Peres2023quantum}%
  \BibitemOpen
  \bibfield  {author} {\bibinfo {author} {\bibfnamefont {F.~C.~R.}\ \bibnamefont {Peres}}\ and\ \bibinfo {author} {\bibfnamefont {E.~F.}\ \bibnamefont {Galv{\~{a}}o}},\ }\href {https://doi.org/10.22331/q-2023-10-03-1126} {\bibfield  {journal} {\bibinfo  {journal} {{Quantum}}\ }\textbf {\bibinfo {volume} {7}},\ \bibinfo {pages} {1126} (\bibinfo {year} {2023})}\BibitemShut {NoStop}%
\bibitem [{\citenamefont {Peres}\ and\ \citenamefont {Galv\~ao}(2025)}]{Peres2025reducing}%
  \BibitemOpen
  \bibfield  {author} {\bibinfo {author} {\bibfnamefont {F.~C.~R.}\ \bibnamefont {Peres}}\ and\ \bibinfo {author} {\bibfnamefont {E.~F.}\ \bibnamefont {Galv\~ao}},\ }\href {https://doi.org/10.1103/d3x5-cgky} {\bibfield  {journal} {\bibinfo  {journal} {Phys. Rev. A}\ }\textbf {\bibinfo {volume} {112}},\ \bibinfo {pages} {062604} (\bibinfo {year} {2025})}\BibitemShut {NoStop}%
\bibitem [{\citenamefont {Lovsted}\ \emph {et~al.}(2026)\citenamefont {Lovsted}, \citenamefont {Peres}, \citenamefont {Nevin}, \citenamefont {Ipek},\ and\ \citenamefont {Broadbent}}]{Lovsted2026blind}%
  \BibitemOpen
  \bibfield  {author} {\bibinfo {author} {\bibfnamefont {D.}~\bibnamefont {Lovsted}}, \bibinfo {author} {\bibfnamefont {F.~C.~R.}\ \bibnamefont {Peres}}, \bibinfo {author} {\bibfnamefont {J.}~\bibnamefont {Nevin}}, \bibinfo {author} {\bibfnamefont {S.}~\bibnamefont {Ipek}},\ and\ \bibinfo {author} {\bibfnamefont {A.}~\bibnamefont {Broadbent}},\ }\href@noop {} {\bibinfo {title} {{Blind Quantum Computation with a Small Quantum Server}}} (\bibinfo {year} {2026}),\ \Eprint {https://arxiv.org/abs/2609.20729} {arXiv:2609.20729 [quant-ph]} \BibitemShut {NoStop}%
\bibitem [{\citenamefont {He}\ \emph {et~al.}(2025)\citenamefont {He}, \citenamefont {Cowtan}, \citenamefont {Williamson},\ and\ \citenamefont {Yoder}}]{He2025extractors}%
  \BibitemOpen
  \bibfield  {author} {\bibinfo {author} {\bibfnamefont {Z.}~\bibnamefont {He}}, \bibinfo {author} {\bibfnamefont {A.}~\bibnamefont {Cowtan}}, \bibinfo {author} {\bibfnamefont {D.~J.}\ \bibnamefont {Williamson}},\ and\ \bibinfo {author} {\bibfnamefont {T.~J.}\ \bibnamefont {Yoder}},\ }\href@noop {} {\bibinfo {title} {Extractors: {QLDPC} {Architectures} for {Efficient} {Pauli}-{Based} {Computation}}} (\bibinfo {year} {2025}),\ \Eprint {https://arxiv.org/abs/2503.10390} {arXiv:2503.10390 [quant-ph]} \BibitemShut {NoStop}%
\bibitem [{\citenamefont {Zhang}\ and\ \citenamefont {Li}(2025)}]{Zhang2025time}%
  \BibitemOpen
  \bibfield  {author} {\bibinfo {author} {\bibfnamefont {G.}~\bibnamefont {Zhang}}\ and\ \bibinfo {author} {\bibfnamefont {Y.}~\bibnamefont {Li}},\ }\href {https://doi.org/10.1103/PhysRevLett.134.070602} {\bibfield  {journal} {\bibinfo  {journal} {Phys. Rev. Lett.}\ }\textbf {\bibinfo {volume} {134}},\ \bibinfo {pages} {070602} (\bibinfo {year} {2025})}\BibitemShut {NoStop}%
\bibitem [{\citenamefont {Cowtan}\ \emph {et~al.}(2026)\citenamefont {Cowtan}, \citenamefont {He}, \citenamefont {Williamson},\ and\ \citenamefont {Yoder}}]{Cowtan2026parallel}%
  \BibitemOpen
  \bibfield  {author} {\bibinfo {author} {\bibfnamefont {A.}~\bibnamefont {Cowtan}}, \bibinfo {author} {\bibfnamefont {Z.}~\bibnamefont {He}}, \bibinfo {author} {\bibfnamefont {D.~J.}\ \bibnamefont {Williamson}},\ and\ \bibinfo {author} {\bibfnamefont {T.~J.}\ \bibnamefont {Yoder}},\ }\href {https://doi.org/10.1103/gj8x-n5gg} {\bibfield  {journal} {\bibinfo  {journal} {PRX Quantum}\ }\textbf {\bibinfo {volume} {7}},\ \bibinfo {pages} {020325} (\bibinfo {year} {2026})}\BibitemShut {NoStop}%
\bibitem [{\citenamefont {Webster}\ \emph {et~al.}(2026)\citenamefont {Webster}, \citenamefont {Berent}, \citenamefont {Chandra}, \citenamefont {Hockings}, \citenamefont {Baspin}, \citenamefont {Thomsen}, \citenamefont {Smith},\ and\ \citenamefont {Cohen}}]{Webster2026pinnacle}%
  \BibitemOpen
  \bibfield  {author} {\bibinfo {author} {\bibfnamefont {P.}~\bibnamefont {Webster}}, \bibinfo {author} {\bibfnamefont {L.}~\bibnamefont {Berent}}, \bibinfo {author} {\bibfnamefont {O.}~\bibnamefont {Chandra}}, \bibinfo {author} {\bibfnamefont {E.~T.}\ \bibnamefont {Hockings}}, \bibinfo {author} {\bibfnamefont {N.}~\bibnamefont {Baspin}}, \bibinfo {author} {\bibfnamefont {F.}~\bibnamefont {Thomsen}}, \bibinfo {author} {\bibfnamefont {S.~C.}\ \bibnamefont {Smith}},\ and\ \bibinfo {author} {\bibfnamefont {L.~Z.}\ \bibnamefont {Cohen}},\ }\href@noop {} {\bibinfo {title} {The {Pinnacle} {Architecture}: {Reducing} the cost of breaking {RSA}-2048 to 100 000 physical qubits using quantum {LDPC} codes}} (\bibinfo {year} {2026}),\ \Eprint {https://arxiv.org/abs/2602.11457} {arXiv:2602.11457 [quant-ph]} \BibitemShut {NoStop}%
\bibitem [{\citenamefont {Litinski}(2019)}]{Litinski2019GoS}%
  \BibitemOpen
  \bibfield  {author} {\bibinfo {author} {\bibfnamefont {D.}~\bibnamefont {Litinski}},\ }\href {https://doi.org/10.22331/q-2019-03-05-128} {\bibfield  {journal} {\bibinfo  {journal} {Quantum}\ }\textbf {\bibinfo {volume} {3}},\ \bibinfo {pages} {128} (\bibinfo {year} {2019})}\BibitemShut {NoStop}%
\bibitem [{\citenamefont {Kan}\ \emph {et~al.}(2025)\citenamefont {Kan}, \citenamefont {Du}, \citenamefont {Liu}, \citenamefont {Wang}, \citenamefont {Ding}, \citenamefont {Li}, \citenamefont {Mao},\ and\ \citenamefont {Stein}}]{Kan2025SPARO}%
  \BibitemOpen
  \bibfield  {author} {\bibinfo {author} {\bibfnamefont {S.}~\bibnamefont {Kan}}, \bibinfo {author} {\bibfnamefont {Z.}~\bibnamefont {Du}}, \bibinfo {author} {\bibfnamefont {C.}~\bibnamefont {Liu}}, \bibinfo {author} {\bibfnamefont {M.}~\bibnamefont {Wang}}, \bibinfo {author} {\bibfnamefont {Y.}~\bibnamefont {Ding}}, \bibinfo {author} {\bibfnamefont {A.}~\bibnamefont {Li}}, \bibinfo {author} {\bibfnamefont {Y.}~\bibnamefont {Mao}},\ and\ \bibinfo {author} {\bibfnamefont {S.}~\bibnamefont {Stein}},\ }\href@noop {} {\bibinfo {title} {{SPARO}: {Surface}-code {Pauli}-based {Architectural} {Resource} {Optimization} for {Fault}-tolerant {Quantum} {Computing}}} (\bibinfo {year} {2025}),\ \Eprint {https://arxiv.org/abs/2504.21854} {arXiv:2504.21854 [quant-ph]} \BibitemShut {NoStop}%
\bibitem [{\citenamefont {Yoder}\ \emph {et~al.}(2025)\citenamefont {Yoder}, \citenamefont {Schoute}, \citenamefont {Rall}, \citenamefont {Pritchett}, \citenamefont {Gambetta}, \citenamefont {Cross}, \citenamefont {Carroll},\ and\ \citenamefont {Beverland}}]{Yoder2025tour}%
  \BibitemOpen
  \bibfield  {author} {\bibinfo {author} {\bibfnamefont {T.~J.}\ \bibnamefont {Yoder}}, \bibinfo {author} {\bibfnamefont {E.}~\bibnamefont {Schoute}}, \bibinfo {author} {\bibfnamefont {P.}~\bibnamefont {Rall}}, \bibinfo {author} {\bibfnamefont {E.}~\bibnamefont {Pritchett}}, \bibinfo {author} {\bibfnamefont {J.~M.}\ \bibnamefont {Gambetta}}, \bibinfo {author} {\bibfnamefont {A.~W.}\ \bibnamefont {Cross}}, \bibinfo {author} {\bibfnamefont {M.}~\bibnamefont {Carroll}},\ and\ \bibinfo {author} {\bibfnamefont {M.~E.}\ \bibnamefont {Beverland}},\ }\href@noop {} {\bibinfo {title} {Tour de gross: {A} modular quantum computer based on bivariate bicycle codes}} (\bibinfo {year} {2025}),\ \Eprint {https://arxiv.org/abs/2506.03094} {arXiv:2506.03094 [quant-ph]} \BibitemShut {NoStop}%
\bibitem [{\citenamefont {Gottesman}\ and\ \citenamefont {Chuang}(1999)}]{gottesman1999demonstrating}%
  \BibitemOpen
  \bibfield  {author} {\bibinfo {author} {\bibfnamefont {D.}~\bibnamefont {Gottesman}}\ and\ \bibinfo {author} {\bibfnamefont {I.~L.}\ \bibnamefont {Chuang}},\ }\href {https://doi.org/https://doi.org/10.1038/46503} {\bibfield  {journal} {\bibinfo  {journal} {Nature}\ }\textbf {\bibinfo {volume} {402}},\ \bibinfo {pages} {390} (\bibinfo {year} {1999})}\BibitemShut {NoStop}%
\bibitem [{\citenamefont {Raussendorf}\ and\ \citenamefont {Briegel}(2001)}]{raussendorf2001one}%
  \BibitemOpen
  \bibfield  {author} {\bibinfo {author} {\bibfnamefont {R.}~\bibnamefont {Raussendorf}}\ and\ \bibinfo {author} {\bibfnamefont {H.~J.}\ \bibnamefont {Briegel}},\ }\href {https://doi.org/10.1103/PhysRevLett.86.5188} {\bibfield  {journal} {\bibinfo  {journal} {Phys. Rev. Lett.}\ }\textbf {\bibinfo {volume} {86}},\ \bibinfo {pages} {5188} (\bibinfo {year} {2001})}\BibitemShut {NoStop}%
\bibitem [{\citenamefont {Nielsen}(2003)}]{nielsen2003quantum}%
  \BibitemOpen
  \bibfield  {author} {\bibinfo {author} {\bibfnamefont {M.~A.}\ \bibnamefont {Nielsen}},\ }\href {https://doi.org/https://doi.org/10.1016/S0375-9601(02)01803-0} {\bibfield  {journal} {\bibinfo  {journal} {Physics Letters A}\ }\textbf {\bibinfo {volume} {308}},\ \bibinfo {pages} {96} (\bibinfo {year} {2003})}\BibitemShut {NoStop}%
\bibitem [{\citenamefont {Aliferis}\ and\ \citenamefont {Leung}(2004)}]{Leung2004unify}%
  \BibitemOpen
  \bibfield  {author} {\bibinfo {author} {\bibfnamefont {P.}~\bibnamefont {Aliferis}}\ and\ \bibinfo {author} {\bibfnamefont {D.~W.}\ \bibnamefont {Leung}},\ }\href {https://doi.org/10.1103/PhysRevA.70.062314} {\bibfield  {journal} {\bibinfo  {journal} {Phys. Rev. A}\ }\textbf {\bibinfo {volume} {70}},\ \bibinfo {pages} {062314} (\bibinfo {year} {2004})}\BibitemShut {NoStop}%
\bibitem [{\citenamefont {Miller}\ and\ \citenamefont {Miyake}(2016)}]{miller2016hierarchy}%
  \BibitemOpen
  \bibfield  {author} {\bibinfo {author} {\bibfnamefont {J.}~\bibnamefont {Miller}}\ and\ \bibinfo {author} {\bibfnamefont {A.}~\bibnamefont {Miyake}},\ }\href {https://doi.org/https://doi.org/10.1038/npjqi.2016.36} {\bibfield  {journal} {\bibinfo  {journal} {npj Quantum Information}\ }\textbf {\bibinfo {volume} {2}},\ \bibinfo {pages} {16036} (\bibinfo {year} {2016})}\BibitemShut {NoStop}%
\bibitem [{\citenamefont {Bartolucci}\ \emph {et~al.}(2023)\citenamefont {Bartolucci}, \citenamefont {Birchall}, \citenamefont {Bombin}, \citenamefont {Cable}, \citenamefont {Dawson}, \citenamefont {Gimeno-Segovia}, \citenamefont {Johnston}, \citenamefont {Kieling}, \citenamefont {Nickerson}, \citenamefont {Pant} \emph {et~al.}}]{bartolucci2023fusion}%
  \BibitemOpen
  \bibfield  {author} {\bibinfo {author} {\bibfnamefont {S.}~\bibnamefont {Bartolucci}}, \bibinfo {author} {\bibfnamefont {P.}~\bibnamefont {Birchall}}, \bibinfo {author} {\bibfnamefont {H.}~\bibnamefont {Bombin}}, \bibinfo {author} {\bibfnamefont {H.}~\bibnamefont {Cable}}, \bibinfo {author} {\bibfnamefont {C.}~\bibnamefont {Dawson}}, \bibinfo {author} {\bibfnamefont {M.}~\bibnamefont {Gimeno-Segovia}}, \bibinfo {author} {\bibfnamefont {E.}~\bibnamefont {Johnston}}, \bibinfo {author} {\bibfnamefont {K.}~\bibnamefont {Kieling}}, \bibinfo {author} {\bibfnamefont {N.}~\bibnamefont {Nickerson}}, \bibinfo {author} {\bibfnamefont {M.}~\bibnamefont {Pant}}, \emph {et~al.},\ }\href {https://doi.org/https://doi.org/10.1038/s41467-023-36493-1} {\bibfield  {journal} {\bibinfo  {journal} {Nature Communications}\ }\textbf {\bibinfo {volume} {14}},\ \bibinfo {pages} {912} (\bibinfo {year} {2023})}\BibitemShut {NoStop}%
\bibitem [{\citenamefont {Danos}\ \emph {et~al.}(2007)\citenamefont {Danos}, \citenamefont {Kashefi},\ and\ \citenamefont {Panangaden}}]{danos2007measurement}%
  \BibitemOpen
  \bibfield  {author} {\bibinfo {author} {\bibfnamefont {V.}~\bibnamefont {Danos}}, \bibinfo {author} {\bibfnamefont {E.}~\bibnamefont {Kashefi}},\ and\ \bibinfo {author} {\bibfnamefont {P.}~\bibnamefont {Panangaden}},\ }\href {https://doi.org/https://doi.org/10.1145/1219092.1219096} {\bibfield  {journal} {\bibinfo  {journal} {Journal of the ACM (JACM)}\ }\textbf {\bibinfo {volume} {54}},\ \bibinfo {pages} {8} (\bibinfo {year} {2007})}\BibitemShut {NoStop}%
\bibitem [{\citenamefont {Danos}\ and\ \citenamefont {Kashefi}(2006)}]{danos2006determinism}%
  \BibitemOpen
  \bibfield  {author} {\bibinfo {author} {\bibfnamefont {V.}~\bibnamefont {Danos}}\ and\ \bibinfo {author} {\bibfnamefont {E.}~\bibnamefont {Kashefi}},\ }\href {https://doi.org/10.1103/PhysRevA.74.052310} {\bibfield  {journal} {\bibinfo  {journal} {Phys. Rev. A}\ }\textbf {\bibinfo {volume} {74}},\ \bibinfo {pages} {052310} (\bibinfo {year} {2006})}\BibitemShut {NoStop}%
\bibitem [{\citenamefont {Browne}\ \emph {et~al.}(2007)\citenamefont {Browne}, \citenamefont {Kashefi}, \citenamefont {Mhalla},\ and\ \citenamefont {Perdrix}}]{browne2007generalized}%
  \BibitemOpen
  \bibfield  {author} {\bibinfo {author} {\bibfnamefont {D.~E.}\ \bibnamefont {Browne}}, \bibinfo {author} {\bibfnamefont {E.}~\bibnamefont {Kashefi}}, \bibinfo {author} {\bibfnamefont {M.}~\bibnamefont {Mhalla}},\ and\ \bibinfo {author} {\bibfnamefont {S.}~\bibnamefont {Perdrix}},\ }\href {https://doi.org/10.1088/1367-2630/9/8/250} {\bibfield  {journal} {\bibinfo  {journal} {New Journal of Physics}\ }\textbf {\bibinfo {volume} {9}},\ \bibinfo {pages} {250} (\bibinfo {year} {2007})}\BibitemShut {NoStop}%
\bibitem [{\citenamefont {Gachechiladze}\ \emph {et~al.}(2019)\citenamefont {Gachechiladze}, \citenamefont {G{\"u}hne},\ and\ \citenamefont {Miyake}}]{gachechiladze2019changing}%
  \BibitemOpen
  \bibfield  {author} {\bibinfo {author} {\bibfnamefont {M.}~\bibnamefont {Gachechiladze}}, \bibinfo {author} {\bibfnamefont {O.}~\bibnamefont {G{\"u}hne}},\ and\ \bibinfo {author} {\bibfnamefont {A.}~\bibnamefont {Miyake}},\ }\href {https://doi.org/10.1103/PhysRevA.99.052304} {\bibfield  {journal} {\bibinfo  {journal} {Phys. Rev. A}\ }\textbf {\bibinfo {volume} {99}},\ \bibinfo {pages} {052304} (\bibinfo {year} {2019})}\BibitemShut {NoStop}%
\bibitem [{\citenamefont {de~Beaudrap}\ \emph {et~al.}(2020)\citenamefont {de~Beaudrap}, \citenamefont {Duncan}, \citenamefont {Horsman},\ and\ \citenamefont {Perdrix}}]{de2019pauli}%
  \BibitemOpen
  \bibfield  {author} {\bibinfo {author} {\bibfnamefont {N.}~\bibnamefont {de~Beaudrap}}, \bibinfo {author} {\bibfnamefont {R.}~\bibnamefont {Duncan}}, \bibinfo {author} {\bibfnamefont {D.}~\bibnamefont {Horsman}},\ and\ \bibinfo {author} {\bibfnamefont {S.}~\bibnamefont {Perdrix}},\ }\href {https://doi.org/10.4204/eptcs.318.6} {\bibfield  {journal} {\bibinfo  {journal} {Electronic Proceedings in Theoretical Computer Science}\ }\textbf {\bibinfo {volume} {318}},\ \bibinfo {pages} {85–105} (\bibinfo {year} {2020})}\BibitemShut {NoStop}%
\bibitem [{\citenamefont {de~Felice}\ \emph {et~al.}(2024)\citenamefont {de~Felice}, \citenamefont {Poór}, \citenamefont {Yeh},\ and\ \citenamefont {Cashman}}]{de2024fusion}%
  \BibitemOpen
  \bibfield  {author} {\bibinfo {author} {\bibfnamefont {G.}~\bibnamefont {de~Felice}}, \bibinfo {author} {\bibfnamefont {B.}~\bibnamefont {Poór}}, \bibinfo {author} {\bibfnamefont {L.}~\bibnamefont {Yeh}},\ and\ \bibinfo {author} {\bibfnamefont {W.}~\bibnamefont {Cashman}},\ }\href@noop {} {\bibinfo {title} {Fusion and flow: formal protocols to reliably build photonic graph states}} (\bibinfo {year} {2024}),\ \Eprint {https://arxiv.org/abs/2409.13541} {arXiv:2409.13541 [quant-ph]} \BibitemShut {NoStop}%
\bibitem [{\citenamefont {Raussendorf}\ \emph {et~al.}(2016)\citenamefont {Raussendorf}, \citenamefont {Sarvepalli}, \citenamefont {Wei},\ and\ \citenamefont {Haghnegahdar}}]{raussendorf2016symmetry}%
  \BibitemOpen
  \bibfield  {author} {\bibinfo {author} {\bibfnamefont {R.}~\bibnamefont {Raussendorf}}, \bibinfo {author} {\bibfnamefont {P.}~\bibnamefont {Sarvepalli}}, \bibinfo {author} {\bibfnamefont {T.-C.}\ \bibnamefont {Wei}},\ and\ \bibinfo {author} {\bibfnamefont {P.}~\bibnamefont {Haghnegahdar}},\ }\href {https://doi.org/https://doi.org/10.1016/j.ic.2016.02.010} {\bibfield  {journal} {\bibinfo  {journal} {Information and Computation}\ }\textbf {\bibinfo {volume} {250}},\ \bibinfo {pages} {115} (\bibinfo {year} {2016})}\BibitemShut {NoStop}%
\bibitem [{\citenamefont {Jorrand}\ and\ \citenamefont {Perdrix}(2005)}]{Perdrix2005unify}%
  \BibitemOpen
  \bibfield  {author} {\bibinfo {author} {\bibfnamefont {P.}~\bibnamefont {Jorrand}}\ and\ \bibinfo {author} {\bibfnamefont {S.}~\bibnamefont {Perdrix}},\ }in\ \href {https://doi.org/10.1117/12.620302} {\emph {\bibinfo {booktitle} {{Quantum Informatics 2004}}}},\ Vol.\ \bibinfo {volume} {5833},\ \bibinfo {editor} {edited by\ \bibinfo {editor} {\bibfnamefont {Y.~I.}\ \bibnamefont {Ozhigov}}},\ \bibinfo {organization} {International Society for Optics and Photonics}\ (\bibinfo  {publisher} {SPIE},\ \bibinfo {year} {2005})\ pp.\ \bibinfo {pages} {44 -- 51}\BibitemShut {NoStop}%
\bibitem [{\citenamefont {Childs}\ \emph {et~al.}(2005)\citenamefont {Childs}, \citenamefont {Leung},\ and\ \citenamefont {Nielsen}}]{ChildsLeung2005unify}%
  \BibitemOpen
  \bibfield  {author} {\bibinfo {author} {\bibfnamefont {A.~M.}\ \bibnamefont {Childs}}, \bibinfo {author} {\bibfnamefont {D.~W.}\ \bibnamefont {Leung}},\ and\ \bibinfo {author} {\bibfnamefont {M.~A.}\ \bibnamefont {Nielsen}},\ }\href {https://doi.org/10.1103/PhysRevA.71.032318} {\bibfield  {journal} {\bibinfo  {journal} {Phys. Rev. A}\ }\textbf {\bibinfo {volume} {71}},\ \bibinfo {pages} {032318} (\bibinfo {year} {2005})}\BibitemShut {NoStop}%
\bibitem [{\citenamefont {Bravyi}\ and\ \citenamefont {Kitaev}(2005)}]{bravyi2005universal}%
  \BibitemOpen
  \bibfield  {author} {\bibinfo {author} {\bibfnamefont {S.}~\bibnamefont {Bravyi}}\ and\ \bibinfo {author} {\bibfnamefont {A.}~\bibnamefont {Kitaev}},\ }\href {https://doi.org/10.1103/PhysRevA.71.022316} {\bibfield  {journal} {\bibinfo  {journal} {Phys. Rev. A}\ }\textbf {\bibinfo {volume} {71}},\ \bibinfo {pages} {022316} (\bibinfo {year} {2005})}\BibitemShut {NoStop}%
\bibitem [{\citenamefont {Heinrich}\ and\ \citenamefont {Gross}(2019)}]{heinrich2019robustness}%
  \BibitemOpen
  \bibfield  {author} {\bibinfo {author} {\bibfnamefont {M.}~\bibnamefont {Heinrich}}\ and\ \bibinfo {author} {\bibfnamefont {D.}~\bibnamefont {Gross}},\ }\href {https://doi.org/10.22331/q-2019-04-08-132} {\bibfield  {journal} {\bibinfo  {journal} {{Quantum}}\ }\textbf {\bibinfo {volume} {3}},\ \bibinfo {pages} {132} (\bibinfo {year} {2019})}\BibitemShut {NoStop}%
\bibitem [{\citenamefont {Raussendorf}\ \emph {et~al.}(2003)\citenamefont {Raussendorf}, \citenamefont {Browne},\ and\ \citenamefont {Briegel}}]{raussendorf2003measurement}%
  \BibitemOpen
  \bibfield  {author} {\bibinfo {author} {\bibfnamefont {R.}~\bibnamefont {Raussendorf}}, \bibinfo {author} {\bibfnamefont {D.~E.}\ \bibnamefont {Browne}},\ and\ \bibinfo {author} {\bibfnamefont {H.~J.}\ \bibnamefont {Briegel}},\ }\href {https://doi.org/10.1103/PhysRevA.68.022312} {\bibfield  {journal} {\bibinfo  {journal} {Phys. Rev. A}\ }\textbf {\bibinfo {volume} {68}},\ \bibinfo {pages} {022312} (\bibinfo {year} {2003})}\BibitemShut {NoStop}%
\bibitem [{\citenamefont {Krishnan~Vijayan}\ \emph {et~al.}(2024)\citenamefont {Krishnan~Vijayan}, \citenamefont {Paler}, \citenamefont {Gavriel}, \citenamefont {Myers}, \citenamefont {Rohde},\ and\ \citenamefont {Devitt}}]{Vijayan+2024compilation}%
  \BibitemOpen
  \bibfield  {author} {\bibinfo {author} {\bibfnamefont {M.}~\bibnamefont {Krishnan~Vijayan}}, \bibinfo {author} {\bibfnamefont {A.}~\bibnamefont {Paler}}, \bibinfo {author} {\bibfnamefont {J.}~\bibnamefont {Gavriel}}, \bibinfo {author} {\bibfnamefont {C.~R.}\ \bibnamefont {Myers}}, \bibinfo {author} {\bibfnamefont {P.~P.}\ \bibnamefont {Rohde}},\ and\ \bibinfo {author} {\bibfnamefont {S.~J.}\ \bibnamefont {Devitt}},\ }\href {https://doi.org/10.1088/2058-9565/ad1f39} {\bibfield  {journal} {\bibinfo  {journal} {Quantum Science and Technology}\ }\textbf {\bibinfo {volume} {9}},\ \bibinfo {pages} {025005} (\bibinfo {year} {2024})}\BibitemShut {NoStop}%
\bibitem [{\citenamefont {Kaldenbach}\ and\ \citenamefont {Heller}(2024)}]{KaldenbachHeller2025mapping}%
  \BibitemOpen
  \bibfield  {author} {\bibinfo {author} {\bibfnamefont {T.~N.}\ \bibnamefont {Kaldenbach}}\ and\ \bibinfo {author} {\bibfnamefont {M.}~\bibnamefont {Heller}},\ }\href {https://doi.org/10.1088/2058-9565/ad802b} {\bibfield  {journal} {\bibinfo  {journal} {Quantum Science and Technology}\ }\textbf {\bibinfo {volume} {10}},\ \bibinfo {pages} {015010} (\bibinfo {year} {2024})}\BibitemShut {NoStop}%
\bibitem [{\citenamefont {Okay}\ \emph {et~al.}(2024)\citenamefont {Okay}, \citenamefont {Yucel},\ and\ \citenamefont {Ipek}}]{Okay2024classical}%
  \BibitemOpen
  \bibfield  {author} {\bibinfo {author} {\bibfnamefont {C.}~\bibnamefont {Okay}}, \bibinfo {author} {\bibfnamefont {A.~T.}\ \bibnamefont {Yucel}},\ and\ \bibinfo {author} {\bibfnamefont {S.}~\bibnamefont {Ipek}},\ }\href@noop {} {\bibinfo {title} {Classical simulation of universal measurement-based quantum computation using multipartite {B}ell scenarios}} (\bibinfo {year} {2024}),\ \Eprint {https://arxiv.org/abs/2410.23734} {arXiv:2410.23734 [quant-ph]} \BibitemShut {NoStop}%
\bibitem [{\citenamefont {Beverland}\ \emph {et~al.}(2022)\citenamefont {Beverland}, \citenamefont {Kliuchnikov},\ and\ \citenamefont {Schoute}}]{Beverland2022surface}%
  \BibitemOpen
  \bibfield  {author} {\bibinfo {author} {\bibfnamefont {M.}~\bibnamefont {Beverland}}, \bibinfo {author} {\bibfnamefont {V.}~\bibnamefont {Kliuchnikov}},\ and\ \bibinfo {author} {\bibfnamefont {E.}~\bibnamefont {Schoute}},\ }\href {https://doi.org/10.1103/PRXQuantum.3.020342} {\bibfield  {journal} {\bibinfo  {journal} {PRX Quantum}\ }\textbf {\bibinfo {volume} {3}},\ \bibinfo {pages} {020342} (\bibinfo {year} {2022})}\BibitemShut {NoStop}%
\bibitem [{\citenamefont {Chamberland}\ and\ \citenamefont {Campbell}(2022)}]{Chamberland2022universal}%
  \BibitemOpen
  \bibfield  {author} {\bibinfo {author} {\bibfnamefont {C.}~\bibnamefont {Chamberland}}\ and\ \bibinfo {author} {\bibfnamefont {E.~T.}\ \bibnamefont {Campbell}},\ }\href {https://doi.org/10.1103/PRXQuantum.3.010331} {\bibfield  {journal} {\bibinfo  {journal} {PRX Quantum}\ }\textbf {\bibinfo {volume} {3}},\ \bibinfo {pages} {010331} (\bibinfo {year} {2022})}\BibitemShut {NoStop}%
\bibitem [{\citenamefont {Peng}\ \emph {et~al.}(2020)\citenamefont {Peng}, \citenamefont {Harrow}, \citenamefont {Ozols},\ and\ \citenamefont {Wu}}]{PengHarrowOW2020simulating}%
  \BibitemOpen
  \bibfield  {author} {\bibinfo {author} {\bibfnamefont {T.}~\bibnamefont {Peng}}, \bibinfo {author} {\bibfnamefont {A.~W.}\ \bibnamefont {Harrow}}, \bibinfo {author} {\bibfnamefont {M.}~\bibnamefont {Ozols}},\ and\ \bibinfo {author} {\bibfnamefont {X.}~\bibnamefont {Wu}},\ }\href {https://doi.org/10.1103/PhysRevLett.125.150504} {\bibfield  {journal} {\bibinfo  {journal} {Phys. Rev. Lett.}\ }\textbf {\bibinfo {volume} {125}},\ \bibinfo {pages} {150504} (\bibinfo {year} {2020})}\BibitemShut {NoStop}%
\bibitem [{\citenamefont {Piveteau}\ and\ \citenamefont {Sutter}(2024)}]{Piveteau2024circuit}%
  \BibitemOpen
  \bibfield  {author} {\bibinfo {author} {\bibfnamefont {C.}~\bibnamefont {Piveteau}}\ and\ \bibinfo {author} {\bibfnamefont {D.}~\bibnamefont {Sutter}},\ }\href {https://doi.org/10.1109/tit.2023.3310797} {\bibfield  {journal} {\bibinfo  {journal} {IEEE Transactions on Information Theory}\ }\textbf {\bibinfo {volume} {70}},\ \bibinfo {pages} {2734–2745} (\bibinfo {year} {2024})}\BibitemShut {NoStop}%
\bibitem [{\citenamefont {Cao}\ and\ \citenamefont {Eisert}(2026)}]{CaoEisert2026mmtdriven}%
  \BibitemOpen
  \bibfield  {author} {\bibinfo {author} {\bibfnamefont {C.}~\bibnamefont {Cao}}\ and\ \bibinfo {author} {\bibfnamefont {J.}~\bibnamefont {Eisert}},\ }\href {https://doi.org/10.1103/4b99-xmqn} {\bibfield  {journal} {\bibinfo  {journal} {Phys. Rev. Lett.}\ }\textbf {\bibinfo {volume} {136}},\ \bibinfo {pages} {080601} (\bibinfo {year} {2026})}\BibitemShut {NoStop}%
\bibitem [{\citenamefont {Grilo}\ \emph {et~al.}(2026)\citenamefont {Grilo}, \citenamefont {Kashefi}, \citenamefont {Markham},\ and\ \citenamefont {Oliveira}}]{Grilo2026power}%
  \BibitemOpen
  \bibfield  {author} {\bibinfo {author} {\bibfnamefont {A.}~\bibnamefont {Grilo}}, \bibinfo {author} {\bibfnamefont {E.}~\bibnamefont {Kashefi}}, \bibinfo {author} {\bibfnamefont {D.}~\bibnamefont {Markham}},\ and\ \bibinfo {author} {\bibfnamefont {M.}~\bibnamefont {Oliveira}},\ }\bibfield  {journal} {\bibinfo  {journal} {ACM Transactions on Quantum Computing}\ }\href {https://doi.org/10.1145/3802819} {10.1145/3802819} (\bibinfo {year} {2026}),\ \bibinfo {note} {{Just Accepted}}\BibitemShut {NoStop}%
\bibitem [{\citenamefont {Lu}\ \emph {et~al.}(2022)\citenamefont {Lu}, \citenamefont {Lessa}, \citenamefont {Kim},\ and\ \citenamefont {Hsieh}}]{Lu2022mmt}%
  \BibitemOpen
  \bibfield  {author} {\bibinfo {author} {\bibfnamefont {T.-C.}\ \bibnamefont {Lu}}, \bibinfo {author} {\bibfnamefont {L.~A.}\ \bibnamefont {Lessa}}, \bibinfo {author} {\bibfnamefont {I.~H.}\ \bibnamefont {Kim}},\ and\ \bibinfo {author} {\bibfnamefont {T.~H.}\ \bibnamefont {Hsieh}},\ }\href {https://doi.org/10.1103/PRXQuantum.3.040337} {\bibfield  {journal} {\bibinfo  {journal} {PRX Quantum}\ }\textbf {\bibinfo {volume} {3}},\ \bibinfo {pages} {040337} (\bibinfo {year} {2022})}\BibitemShut {NoStop}%
\bibitem [{\citenamefont {Smith}\ \emph {et~al.}(2023)\citenamefont {Smith}, \citenamefont {Crane}, \citenamefont {Wiebe},\ and\ \citenamefont {Girvin}}]{Smith2023deterministic}%
  \BibitemOpen
  \bibfield  {author} {\bibinfo {author} {\bibfnamefont {K.~C.}\ \bibnamefont {Smith}}, \bibinfo {author} {\bibfnamefont {E.}~\bibnamefont {Crane}}, \bibinfo {author} {\bibfnamefont {N.}~\bibnamefont {Wiebe}},\ and\ \bibinfo {author} {\bibfnamefont {S.}~\bibnamefont {Girvin}},\ }\href {https://doi.org/10.1103/PRXQuantum.4.020315} {\bibfield  {journal} {\bibinfo  {journal} {PRX Quantum}\ }\textbf {\bibinfo {volume} {4}},\ \bibinfo {pages} {020315} (\bibinfo {year} {2023})}\BibitemShut {NoStop}%
\bibitem [{\citenamefont {DeCross}\ \emph {et~al.}(2023)\citenamefont {DeCross}, \citenamefont {Chertkov}, \citenamefont {Kohagen},\ and\ \citenamefont {Foss-Feig}}]{DeCross+2023qubitreuse}%
  \BibitemOpen
  \bibfield  {author} {\bibinfo {author} {\bibfnamefont {M.}~\bibnamefont {DeCross}}, \bibinfo {author} {\bibfnamefont {E.}~\bibnamefont {Chertkov}}, \bibinfo {author} {\bibfnamefont {M.}~\bibnamefont {Kohagen}},\ and\ \bibinfo {author} {\bibfnamefont {M.}~\bibnamefont {Foss-Feig}},\ }\href {https://doi.org/10.1103/PhysRevX.13.041057} {\bibfield  {journal} {\bibinfo  {journal} {Phys. Rev. X}\ }\textbf {\bibinfo {volume} {13}},\ \bibinfo {pages} {041057} (\bibinfo {year} {2023})}\BibitemShut {NoStop}%
\bibitem [{\citenamefont {Fang}\ \emph {et~al.}(2026)\citenamefont {Fang}, \citenamefont {Zhang}, \citenamefont {Shi},\ and\ \citenamefont {Li}}]{Fang+2026}%
  \BibitemOpen
  \bibfield  {author} {\bibinfo {author} {\bibfnamefont {K.}~\bibnamefont {Fang}}, \bibinfo {author} {\bibfnamefont {M.}~\bibnamefont {Zhang}}, \bibinfo {author} {\bibfnamefont {R.}~\bibnamefont {Shi}},\ and\ \bibinfo {author} {\bibfnamefont {Y.}~\bibnamefont {Li}},\ }\href {https://doi.org/10.1109/TC.2025.3643826} {\bibfield  {journal} {\bibinfo  {journal} {IEEE Transactions on Computers}\ }\textbf {\bibinfo {volume} {75}},\ \bibinfo {pages} {748} (\bibinfo {year} {2026})}\BibitemShut {NoStop}%
\bibitem [{\citenamefont {B\"aumer}\ \emph {et~al.}(2024)\citenamefont {B\"aumer}, \citenamefont {Tripathi}, \citenamefont {Seif}, \citenamefont {Lidar},\ and\ \citenamefont {Wang}}]{Baumer+2024qft}%
  \BibitemOpen
  \bibfield  {author} {\bibinfo {author} {\bibfnamefont {E.}~\bibnamefont {B\"aumer}}, \bibinfo {author} {\bibfnamefont {V.}~\bibnamefont {Tripathi}}, \bibinfo {author} {\bibfnamefont {A.}~\bibnamefont {Seif}}, \bibinfo {author} {\bibfnamefont {D.}~\bibnamefont {Lidar}},\ and\ \bibinfo {author} {\bibfnamefont {D.~S.}\ \bibnamefont {Wang}},\ }\href {https://doi.org/10.1103/PhysRevLett.133.150602} {\bibfield  {journal} {\bibinfo  {journal} {Phys. Rev. Lett.}\ }\textbf {\bibinfo {volume} {133}},\ \bibinfo {pages} {150602} (\bibinfo {year} {2024})}\BibitemShut {NoStop}%
\bibitem [{\citenamefont {Kissinger}\ and\ \citenamefont {van~de Wetering}(2020)}]{KissWet2020reduceT}%
  \BibitemOpen
  \bibfield  {author} {\bibinfo {author} {\bibfnamefont {A.}~\bibnamefont {Kissinger}}\ and\ \bibinfo {author} {\bibfnamefont {J.}~\bibnamefont {van~de Wetering}},\ }\href {https://doi.org/10.1103/PhysRevA.102.022406} {\bibfield  {journal} {\bibinfo  {journal} {Phys. Rev. A}\ }\textbf {\bibinfo {volume} {102}},\ \bibinfo {pages} {022406} (\bibinfo {year} {2020})}\BibitemShut {NoStop}%
\bibitem [{\citenamefont {van~de Wetering}(2020)}]{Wetering2020zxcalculus}%
  \BibitemOpen
  \bibfield  {author} {\bibinfo {author} {\bibfnamefont {J.}~\bibnamefont {van~de Wetering}},\ }\href@noop {} {\bibinfo {title} {{ZX-calculus for the working quantum computer scientist}}} (\bibinfo {year} {2020}),\ \Eprint {https://arxiv.org/abs/2012.13966} {arXiv:2012.13966 [quant-ph]} \BibitemShut {NoStop}%
\bibitem [{\citenamefont {Staudacher}\ \emph {et~al.}(2023)\citenamefont {Staudacher}, \citenamefont {Guggemos}, \citenamefont {Grundner-Culemann},\ and\ \citenamefont {Gehrke}}]{Staudacher2023}%
  \BibitemOpen
  \bibfield  {author} {\bibinfo {author} {\bibfnamefont {K.}~\bibnamefont {Staudacher}}, \bibinfo {author} {\bibfnamefont {T.}~\bibnamefont {Guggemos}}, \bibinfo {author} {\bibfnamefont {S.}~\bibnamefont {Grundner-Culemann}},\ and\ \bibinfo {author} {\bibfnamefont {W.}~\bibnamefont {Gehrke}},\ }\href {https://doi.org/10.4204/eptcs.394.3} {\bibfield  {journal} {\bibinfo  {journal} {Electronic Proceedings in Theoretical Computer Science}\ }\textbf {\bibinfo {volume} {394}},\ \bibinfo {pages} {29–45} (\bibinfo {year} {2023})}\BibitemShut {NoStop}%
\bibitem [{\citenamefont {Kaldenbach}\ and\ \citenamefont {Heller}(2025)}]{kaldenbach2025mapping}%
  \BibitemOpen
  \bibfield  {author} {\bibinfo {author} {\bibfnamefont {T.~N.}\ \bibnamefont {Kaldenbach}}\ and\ \bibinfo {author} {\bibfnamefont {M.}~\bibnamefont {Heller}},\ }\href {https://doi.org/10.1088/2058-9565/ad802b} {\bibfield  {journal} {\bibinfo  {journal} {Quantum Science and Technology}\ }\textbf {\bibinfo {volume} {10}},\ \bibinfo {pages} {015010} (\bibinfo {year} {2025})}\BibitemShut {NoStop}%
\bibitem [{\citenamefont {Peres}(2026)}]{GitHubPACTS}%
  \BibitemOpen
  \bibfield  {author} {\bibinfo {author} {\bibfnamefont {F.~C.~R.}\ \bibnamefont {Peres}},\ }\href@noop {} {\bibinfo {title} {{PACTS}}},\ \bibinfo {howpublished} {\url{https://github.com/fcrperes/PACTS}} (\bibinfo {year} {2026})\BibitemShut {NoStop}%
\bibitem [{\citenamefont {Aaronson}\ and\ \citenamefont {Gottesman}(2004)}]{AarosonGottesman2004improved}%
  \BibitemOpen
  \bibfield  {author} {\bibinfo {author} {\bibfnamefont {S.}~\bibnamefont {Aaronson}}\ and\ \bibinfo {author} {\bibfnamefont {D.}~\bibnamefont {Gottesman}},\ }\href {https://doi.org/10.1103/PhysRevA.70.052328} {\bibfield  {journal} {\bibinfo  {journal} {Phys. Rev. A}\ }\textbf {\bibinfo {volume} {70}},\ \bibinfo {pages} {052328} (\bibinfo {year} {2004})}\BibitemShut {NoStop}%
\end{thebibliography}

\begin{thebibliography}{17}%
\makeatletter
\providecommand \@ifxundefined [1]{%
 \@ifx{#1\undefined}
}%
\providecommand \@ifnum [1]{%
 \ifnum #1\expandafter \@firstoftwo
 \else \expandafter \@secondoftwo
 \fi
}%
\providecommand \@ifx [1]{%
 \ifx #1\expandafter \@firstoftwo
 \else \expandafter \@secondoftwo
 \fi
}%
\providecommand \natexlab [1]{#1}%
\providecommand \enquote  [1]{``#1''}%
\providecommand \bibnamefont  [1]{#1}%
\providecommand \bibfnamefont [1]{#1}%
\providecommand \citenamefont [1]{#1}%
\providecommand \href@noop [0]{\@secondoftwo}%
\providecommand \href [0]{\begingroup \@sanitize@url \@href}%
\providecommand \@href[1]{\@@startlink{#1}\@@href}%
\providecommand \@@href[1]{\endgroup#1\@@endlink}%
\providecommand \@sanitize@url [0]{\catcode `\\12\catcode `\$12\catcode `\&12\catcode `\#12\catcode `\^12\catcode `\_12\catcode `\%12\relax}%
\providecommand \@@startlink[1]{}%
\providecommand \@@endlink[0]{}%
\providecommand \url  [0]{\begingroup\@sanitize@url \@url }%
\providecommand \@url [1]{\endgroup\@href {#1}{\urlprefix }}%
\providecommand \urlprefix  [0]{URL }%
\providecommand \Eprint [0]{\href }%
\providecommand \doibase [0]{https://doi.org/}%
\providecommand \selectlanguage [0]{\@gobble}%
\providecommand \bibinfo  [0]{\@secondoftwo}%
\providecommand \bibfield  [0]{\@secondoftwo}%
\providecommand \translation [1]{[#1]}%
\providecommand \BibitemOpen [0]{}%
\providecommand \bibitemStop [0]{}%
\providecommand \bibitemNoStop [0]{.\EOS\space}%
\providecommand \EOS [0]{\spacefactor3000\relax}%
\providecommand \BibitemShut  [1]{\csname bibitem#1\endcsname}%
\let\auto@bib@innerbib\@empty
%</preamble>
\bibitem [{\citenamefont {Bravyi}\ \emph {et~al.}(2016)\citenamefont {Bravyi}, \citenamefont {Smith},\ and\ \citenamefont {Smolin}}]{Bravyi2016trading-sm}%
  \BibitemOpen
  \bibfield  {author} {\bibinfo {author} {\bibfnamefont {S.}~\bibnamefont {Bravyi}}, \bibinfo {author} {\bibfnamefont {G.}~\bibnamefont {Smith}},\ and\ \bibinfo {author} {\bibfnamefont {J.~A.}\ \bibnamefont {Smolin}},\ }\href {https://doi.org/10.1103/PhysRevX.6.021043} {\bibfield  {journal} {\bibinfo  {journal} {Phys. Rev. X}\ }\textbf {\bibinfo {volume} {6}},\ \bibinfo {pages} {021043} (\bibinfo {year} {2016})}\BibitemShut {NoStop}%
\bibitem [{\citenamefont {Yoganathan}\ \emph {et~al.}(2019)\citenamefont {Yoganathan}, \citenamefont {Jozsa},\ and\ \citenamefont {Strelchuk}}]{Yoganathan2019quantum-sm}%
  \BibitemOpen
  \bibfield  {author} {\bibinfo {author} {\bibfnamefont {M.}~\bibnamefont {Yoganathan}}, \bibinfo {author} {\bibfnamefont {R.}~\bibnamefont {Jozsa}},\ and\ \bibinfo {author} {\bibfnamefont {S.}~\bibnamefont {Strelchuk}},\ }\href {https://doi.org/10.1098/rspa.2018.0427} {\bibfield  {journal} {\bibinfo  {journal} {Proc. R. Soc. A}\ }\textbf {\bibinfo {volume} {475}},\ \bibinfo {pages} {20180427} (\bibinfo {year} {2019})}\BibitemShut {NoStop}%
\bibitem [{\citenamefont {Pashayan}\ \emph {et~al.}(2022)\citenamefont {Pashayan}, \citenamefont {Reardon-Smith}, \citenamefont {Korzekwa},\ and\ \citenamefont {Bartlett}}]{pashayan2022fast-sm}%
  \BibitemOpen
  \bibfield  {author} {\bibinfo {author} {\bibfnamefont {H.}~\bibnamefont {Pashayan}}, \bibinfo {author} {\bibfnamefont {O.}~\bibnamefont {Reardon-Smith}}, \bibinfo {author} {\bibfnamefont {K.}~\bibnamefont {Korzekwa}},\ and\ \bibinfo {author} {\bibfnamefont {S.~D.}\ \bibnamefont {Bartlett}},\ }\href {https://doi.org/10.1103/PRXQuantum.3.020361} {\bibfield  {journal} {\bibinfo  {journal} {PRX Quantum}\ }\textbf {\bibinfo {volume} {3}},\ \bibinfo {pages} {020361} (\bibinfo {year} {2022})}\BibitemShut {NoStop}%
\bibitem [{\citenamefont {Raussendorf}(2013)}]{Raussendorf2013contextuality-sm}%
  \BibitemOpen
  \bibfield  {author} {\bibinfo {author} {\bibfnamefont {R.}~\bibnamefont {Raussendorf}},\ }\href {https://doi.org/10.1103/PhysRevA.88.022322} {\bibfield  {journal} {\bibinfo  {journal} {Phys. Rev. A}\ }\textbf {\bibinfo {volume} {88}},\ \bibinfo {pages} {022322} (\bibinfo {year} {2013})}\BibitemShut {NoStop}%
\bibitem [{\citenamefont {Aaronson}\ and\ \citenamefont {Gottesman}(2004)}]{AarosonGottesman2004improved-sm}%
  \BibitemOpen
  \bibfield  {author} {\bibinfo {author} {\bibfnamefont {S.}~\bibnamefont {Aaronson}}\ and\ \bibinfo {author} {\bibfnamefont {D.}~\bibnamefont {Gottesman}},\ }\href {https://doi.org/10.1103/PhysRevA.70.052328} {\bibfield  {journal} {\bibinfo  {journal} {Phys. Rev. A}\ }\textbf {\bibinfo {volume} {70}},\ \bibinfo {pages} {052328} (\bibinfo {year} {2004})}\BibitemShut {NoStop}%
\bibitem [{\citenamefont {He}\ \emph {et~al.}(2025)\citenamefont {He}, \citenamefont {Cowtan}, \citenamefont {Williamson},\ and\ \citenamefont {Yoder}}]{He2025extractors-sm}%
  \BibitemOpen
  \bibfield  {author} {\bibinfo {author} {\bibfnamefont {Z.}~\bibnamefont {He}}, \bibinfo {author} {\bibfnamefont {A.}~\bibnamefont {Cowtan}}, \bibinfo {author} {\bibfnamefont {D.~J.}\ \bibnamefont {Williamson}},\ and\ \bibinfo {author} {\bibfnamefont {T.~J.}\ \bibnamefont {Yoder}},\ }\href@noop {} {\bibinfo {title} {Extractors: {QLDPC} {Architectures} for {Efficient} {Pauli}-{Based} {Computation}}} (\bibinfo {year} {2025}),\ \Eprint {https://arxiv.org/abs/2503.10390} {arXiv:2503.10390 [quant-ph]} \BibitemShut {NoStop}%
\bibitem [{\citenamefont {Nielsen}\ and\ \citenamefont {Chuang}(2011)}]{NielsenChuang-sm}%
  \BibitemOpen
  \bibfield  {author} {\bibinfo {author} {\bibfnamefont {M.~A.}\ \bibnamefont {Nielsen}}\ and\ \bibinfo {author} {\bibfnamefont {I.~L.}\ \bibnamefont {Chuang}},\ }\href@noop {} {\emph {\bibinfo {title} {{Quantum Computation and Quantum Information: 10th Anniversary Edition}}}},\ \bibinfo {edition} {10th}\ ed.\ (\bibinfo  {publisher} {Cambridge University Press},\ \bibinfo {address} {USA},\ \bibinfo {year} {2011})\BibitemShut {NoStop}%
\bibitem [{\citenamefont {Shende}\ and\ \citenamefont {Markov}(2009)}]{ShendeMarkov08-sm}%
  \BibitemOpen
  \bibfield  {author} {\bibinfo {author} {\bibfnamefont {V.~V.}\ \bibnamefont {Shende}}\ and\ \bibinfo {author} {\bibfnamefont {I.~L.}\ \bibnamefont {Markov}},\ }\href {https://doi.org/10.26421/QIC8.5-6-8} {\bibfield  {journal} {\bibinfo  {journal} {Quantum Info. Comput.}\ }\textbf {\bibinfo {volume} {9}},\ \bibinfo {pages} {461–486} (\bibinfo {year} {2009})}\BibitemShut {NoStop}%
\bibitem [{\citenamefont {Peres}(2026)}]{GitHubPACTS-sm}%
  \BibitemOpen
  \bibfield  {author} {\bibinfo {author} {\bibfnamefont {F.~C.~R.}\ \bibnamefont {Peres}},\ }\href@noop {} {\bibinfo {title} {{PACTS}}},\ \bibinfo {howpublished} {\url{https://github.com/fcrperes/PACTS}} (\bibinfo {year} {2026})\BibitemShut {NoStop}%
\bibitem [{\citenamefont {Tang}(2019)}]{Tang2019quantuminspired-sm}%
  \BibitemOpen
  \bibfield  {author} {\bibinfo {author} {\bibfnamefont {E.}~\bibnamefont {Tang}},\ }in\ \href {https://doi.org/10.1145/3313276.3316310} {\emph {\bibinfo {booktitle} {Proceedings of the 51st Annual ACM SIGACT Symposium on Theory of Computing}}},\ \bibinfo {series and number} {STOC 2019}\ (\bibinfo  {publisher} {Association for Computing Machinery},\ \bibinfo {address} {New York, NY, USA},\ \bibinfo {year} {2019})\ p.\ \bibinfo {pages} {217–228}\BibitemShut {NoStop}%
\bibitem [{\citenamefont {Bravyi}\ \emph {et~al.}(2021)\citenamefont {Bravyi}, \citenamefont {Shaydulin}, \citenamefont {Hu},\ and\ \citenamefont {Maslov}}]{Bravyi2021clifford-sm}%
  \BibitemOpen
  \bibfield  {author} {\bibinfo {author} {\bibfnamefont {S.}~\bibnamefont {Bravyi}}, \bibinfo {author} {\bibfnamefont {R.}~\bibnamefont {Shaydulin}}, \bibinfo {author} {\bibfnamefont {S.}~\bibnamefont {Hu}},\ and\ \bibinfo {author} {\bibfnamefont {D.}~\bibnamefont {Maslov}},\ }\href {https://doi.org/10.22331/q-2021-11-16-580} {\bibfield  {journal} {\bibinfo  {journal} {{Quantum}}\ }\textbf {\bibinfo {volume} {5}},\ \bibinfo {pages} {580} (\bibinfo {year} {2021})}\BibitemShut {NoStop}%
\bibitem [{\citenamefont {Bravyi}\ and\ \citenamefont {Maslov}(2021)}]{Bravyi2021hadamardfree-sm}%
  \BibitemOpen
  \bibfield  {author} {\bibinfo {author} {\bibfnamefont {S.}~\bibnamefont {Bravyi}}\ and\ \bibinfo {author} {\bibfnamefont {D.}~\bibnamefont {Maslov}},\ }\href {https://doi.org/10.1109/TIT.2021.3081415} {\bibfield  {journal} {\bibinfo  {journal} {IEEE Transactions on Information Theory}\ }\textbf {\bibinfo {volume} {67}},\ \bibinfo {pages} {4546} (\bibinfo {year} {2021})}\BibitemShut {NoStop}%
\bibitem [{\citenamefont {Peres}\ and\ \citenamefont {Galv{\~{a}}o}(2023)}]{Peres2023quantum-sm}%
  \BibitemOpen
  \bibfield  {author} {\bibinfo {author} {\bibfnamefont {F.~C.~R.}\ \bibnamefont {Peres}}\ and\ \bibinfo {author} {\bibfnamefont {E.~F.}\ \bibnamefont {Galv{\~{a}}o}},\ }\href {https://doi.org/10.22331/q-2023-10-03-1126} {\bibfield  {journal} {\bibinfo  {journal} {{Quantum}}\ }\textbf {\bibinfo {volume} {7}},\ \bibinfo {pages} {1126} (\bibinfo {year} {2023})}\BibitemShut {NoStop}%
\bibitem [{\citenamefont {Peres}\ and\ \citenamefont {Galv\~ao}(2025)}]{Peres2025reducing-sm}%
  \BibitemOpen
  \bibfield  {author} {\bibinfo {author} {\bibfnamefont {F.~C.~R.}\ \bibnamefont {Peres}}\ and\ \bibinfo {author} {\bibfnamefont {E.~F.}\ \bibnamefont {Galv\~ao}},\ }\href {https://doi.org/10.1103/d3x5-cgky} {\bibfield  {journal} {\bibinfo  {journal} {Phys. Rev. A}\ }\textbf {\bibinfo {volume} {112}},\ \bibinfo {pages} {062604} (\bibinfo {year} {2025})}\BibitemShut {NoStop}%
\bibitem [{\citenamefont {Kissinger}\ and\ \citenamefont {van~de Wetering}(2020)}]{KissWet2020reduceT-sm}%
  \BibitemOpen
  \bibfield  {author} {\bibinfo {author} {\bibfnamefont {A.}~\bibnamefont {Kissinger}}\ and\ \bibinfo {author} {\bibfnamefont {J.}~\bibnamefont {van~de Wetering}},\ }\href {https://doi.org/10.1103/PhysRevA.102.022406} {\bibfield  {journal} {\bibinfo  {journal} {Phys. Rev. A}\ }\textbf {\bibinfo {volume} {102}},\ \bibinfo {pages} {022406} (\bibinfo {year} {2020})}\BibitemShut {NoStop}%
\bibitem [{\citenamefont {van~de Wetering}(2020)}]{Wetering2020zxcalculus-sm}%
  \BibitemOpen
  \bibfield  {author} {\bibinfo {author} {\bibfnamefont {J.}~\bibnamefont {van~de Wetering}},\ }\href@noop {} {\bibinfo {title} {{ZX-calculus for the working quantum computer scientist}}} (\bibinfo {year} {2020}),\ \Eprint {https://arxiv.org/abs/2012.13966} {arXiv:2012.13966 [quant-ph]} \BibitemShut {NoStop}%
\bibitem [{\citenamefont {Staudacher}\ \emph {et~al.}(2023)\citenamefont {Staudacher}, \citenamefont {Guggemos}, \citenamefont {Grundner-Culemann},\ and\ \citenamefont {Gehrke}}]{Staudacher2023-sm}%
  \BibitemOpen
  \bibfield  {author} {\bibinfo {author} {\bibfnamefont {K.}~\bibnamefont {Staudacher}}, \bibinfo {author} {\bibfnamefont {T.}~\bibnamefont {Guggemos}}, \bibinfo {author} {\bibfnamefont {S.}~\bibnamefont {Grundner-Culemann}},\ and\ \bibinfo {author} {\bibfnamefont {W.}~\bibnamefont {Gehrke}},\ }\href {https://doi.org/10.4204/eptcs.394.3} {\bibfield  {journal} {\bibinfo  {journal} {Electronic Proceedings in Theoretical Computer Science}\ }\textbf {\bibinfo {volume} {394}},\ \bibinfo {pages} {29–45} (\bibinfo {year} {2023})}\BibitemShut {NoStop}%
\end{thebibliography}
\end{document}